%% file: article.tex
\documentclass[a4paper, 10pt]{article}
\usepackage{srcltx,graphicx,epstopdf}
\usepackage{amsmath, amssymb, amsthm}
\usepackage{color}
\usepackage{multirow}
\usepackage{bm}
\usepackage{subfig} 
\usepackage{hyperref}
\usepackage{scalerel}
\usepackage[margin=1in]{geometry}
\usepackage{cleveref}
\crefname{section}{Section}{Sections}
\crefname{subsection}{Subsection}{Subsections}
\crefname{appendix}{Appendix}{Appendix}
\crefname{figure}{Figure}{Figures}
\crefname{table}{Table}{Tables}
\crefname{property}{Property}{Properties}
\crefname{theorem}{Theorem}{Theorem}

\graphicspath{{images/}}

\newtheorem{theorem}{Theorem}
\newtheorem{lemma}[theorem]{Lemma}
\newtheorem{deduction}[theorem]{Deduction}
\newtheorem{property}[theorem]{Property}

\newtheorem{definition}[theorem]{Definition}

{\theoremstyle{remark} }

\newcommand\bbR{\mathbb{R}}
\newcommand\bbN{\mathbb{N}}

\newcommand\br{\boldsymbol{r}}
\newcommand\bF{ {\boldsymbol{F}}}
\newcommand\bJ{ {\bf J}}
\newcommand\bS{ {\boldsymbol{S}}}
\newcommand\bU{ {\boldsymbol{U}}}

\newcommand\dd{\,\mathrm{d}}
\newcommand\eq{\mathrm{eq}}
\newcommand\sgn{\mathrm{sgn}}

\newcommand\out{\mathrm{out}}
\newcommand\ctend{c\,t_{\mathrm{end}}}
\newcommand\cell{\mathrm{cell}}

\newcommand\PN{{$P_N$ }}
\newcommand\MN{{$M_N$ }}
\newcommand\MP[1]{{${M\!P}_{#1}$ }}
\newcommand\MPN{{\MP{\!N}}}
\newcommand\Mone{{$M_1$ }}
\newcommand\MPtwo{{\MP{2}}}

\newcommand\mE{{\mathcal{E}}}
\newcommand\mS{{\mathcal{S}}}

\newcommand\weight{\omega^{[\alpha]}}
\newcommand\pl{\phi^{[\alpha]}}
\newcommand\Pl{\Phi^{[\alpha]}}

\newcommand\pd[2]{\dfrac{\partial {#1}}{\partial {#2}}}
\newcommand\od[2]{\dfrac{\dd {#1}}{\dd {#2}}}

\newcommand\inner[2]{\left\langle {#1}, ~{#2}\right\rangle^{[\alpha]}}
\newcommand\moment[1]{{\langle#1\rangle}}

\numberwithin{equation}{section}

\title
{A Nonlinear Moment Model for Radiative Transfer Equation in Slab Geometry}

\author{ 
Yuwei Fan\thanks{Department of Mathematics, Stanford University,
    Stanford, CA 94305, email: {\tt ywfan@stanford.edu}},~~
Ruo Li\thanks{HEDPS \& CAPT, LMAM \& School of Mathematical
    Sciences, Peking University, Beijing, China, email: {\tt rli@math.pku.edu.cn}}, 
  ~~and
Lingchao Zheng\thanks{School of Mathematical Sciences, Peking
    University, Beijing, China, email: {\tt lczheng@pku.edu.cn}} }

\begin{document}
\maketitle

\input{introduction}
\input{preliminaries}
\input{M1PNmodel}
\input{analysis}
\input{numerical}
\input{conclusion}
\input{appendix}
\bibliographystyle{abbrv}
\bibliography{../../article,../references}

\end{document}

%% file: introduction.tex
\begin{abstract}
    This paper is concerned with the approximation of the radiative transfer
    equation for a grey medium in the slab geometry by the moment method.
    We develop a novel moment model inspired by the classical \PN model and \MN
    model. The new model takes the ansatz of the \Mone model as the weight
    function and follows the primary idea of the \PN model to approximate the
    specific intensity by expanding it around the weight function in terms of
    orthogonal polynomials. The weight function uses the information of the first
    two moments, which brings the new model the capability to approximate an
    anisotropic distribution. Mathematical properties of the moment model are
    investigated, and particularly the hyperbolicity and the characteristic
    structure of the Riemann problem of the model with three moments are studied
    in detail. Some numerical simulations demonstrate its
    numerical efficiency and show its superior in comparison to the \PN model.

  \vspace*{4mm}
\end{abstract}

\section{Introduction}
In kinetic theory, the radiative transfer equation (RTE), which describes the
particle propagation and interaction with a background medium, has applications
in a wide variety of subjects, such as neutron transport in reactor physics
\cite{pomraning1973equations, duderstadt1979transport}, light transport in
atmospheric radiative transfer \cite{marshak20053d}, heat transfer
\cite{koch2004evaluation} and optical imaging
\cite{klose2002optical,tarvainen2005hybrid}.
The RTE is a high-dimensional integro-differential kinetic transport equation
for the specific intensity of radiation, and its high-dimensionality makes it
particularly challenging to solve numerically.  Currently, numerical methods for
solving the RTE can be categorized into two groups: the probabilistic
approaches, for example, the direct simulation Monte Carlo methods \cite{Bird,
hayakawa2007coupled, densmore2012hybrid} and the deterministic schemes
\cite{broadwell1964study, larsen2010advances, Stamnes1988Electromagnetic,
jeans1917stars, Davison1960on, dubroca1999theoretical, minerbo1978maximum,
alldredge2016approximating, fan2018fast}.
The Monte Carlo method solves the RTE by simulating a lot of individual
particles and determine the intensity by taking a statistical average over
particles. This method has made remarkable successes in solving the RTE, but the
statistical noise is an important issue for its accuracy.

Among the deterministic methods, two major approaches are the discrete ordinates
method \cite{broadwell1964study, larsen2010advances, Stamnes1988Electromagnetic}
and the moment methods \cite{jeans1917stars, Davison1960on,
dubroca1999theoretical, alldredge2016approximating}.
As the most popular deterministic method in solving the RTE, the
discrete-ordinates method ($S_N$) solves the transfer equation along a discrete
set of angular directions. The main flaw of this method is the so-called
\emph{ray effects} \cite{larsen2010advances} because the number of the discrete
angular directions is finite and the particles are only allowed to move along
these directions.

Moment method depicts the evolution of a finite number of moments of the
specific intensity. Typically, the governing equations of a lower order moments
depend on higher order moments. Hence a moment closure is required to close the
moment system. A common method for the moment closure is to construct an ansatz
to approximate the specific intensity. Based on this idea, the two most popular
moment methods are the the spherical harmonics method ($P_N$)
\cite{pomraning1973equations} and the maximum entropy method ($M_N$)
\cite{levermore1996moment, dubroca1999theoretical, minerbo1978maximum}. The \PN
model constructs the ansatz by expanding the specific intensity around the
equilibrium in terms of spherical harmonics in the velocity direction.
The resulting model is a linear symmetric hyperbolic system and easy to
implement, but it may lead to nonphysical oscillations, which may
lead to negative particle concentration \cite{brunner2001one, brunner2005two,
mcclarren2008solutions}.
The \MN model constructs the ansatz is based on the principle of maximum
entropy\cite{levermore1996moment, dubroca1999theoretical}.
The resulting model retains fundamental properties of the underlying kinetic
equations such as hyperbolicity, entropy dissipation, and positivity of the
intensity. However, for the the case $N\geq 2$, there is no algebraic
expression of the closed moments, and one has to solve an ill-conditioned
optimization problem to obtain the closed moments in the implementation, which
strongly limits the applications of the \MN model.


In this paper, we first propose two points of view --- the entropy-based
viewpoint and the weighted expansion viewpoint --- on the relationship between
the \PN model and the \MN model. 
Based on the entropy-based viewpoint, we show that both the models can be
attributed to a minimization problem with specific object functions.  The
weighted expansion viewpoint reveals that the two models can also be treated as
approximating the intensity by expanding it around a given
weight function in terms of orthogonal polynomials.
These viewpoints indicate that one can construct a new moment model by choosing
a new weight function.

Although there is no algebraic expression of the closed moment for the \MN
model for $N\geq 2$, the expression of the \Mone model is simple.
We take the ansatz of the \Mone model as the weight function to develop a novel
arbitrary order moment model (we call it \MPN here and put the explanation in
\cref{sec:M1Pndeduction}) for the RTE. Since the weight function contains the
information of the zeroth-order moment and the first-order moment, the ansatz of
the new model is expected to have the capability to approximate an anisotropic
distribution. The new moment model has a simple algebraic expression and easy
to implement. We study the hyperbolicity and the characteristic structure of the
\MPtwo model in detail.  This model is hyperbolic, and the judging criteria on
the wave-type are investigated.
Comparison with the $P_2$ and $M_2$ models, the \MPtwo lies in the betweenness
of these two models and can be viewed as an approximation of the $M_2$ model.
Numerical simulations are performed to study the numerical behavior of the new
moment model and show that the \MPN model has the power to simulate strong
anisotropic intensity, and has superiority on the \PN model.

The rest of this paper is arranged as follows. In \cref{sec:preliminiaries}, we
briefly introduce the radiative transfer equation and the moment method, and
present two viewpoints on the relationship between the \PN and \MN models.
In \cref{sec:M1Pndeduction}, we derive the 
arbitrary order \MPN model and investigate the \MPtwo model in detail.
Numerical issues, including the numerical scheme, details on implementation, and
numerical results, are presented in \cref{sec:numericalresults}. 
The paper ends with a conclusion in \cref{sec:conclusion}.


%% file: preliminaries.tex
\section{Radiative transfer equation and moment method} \label{sec:preliminiaries}
In this paper, we study the time-dependent radiative transfer equation (RTE)
for a grey medium in the slab geometry as
\begin{equation} \label{eq:radiativetransfer}
    \frac{1}{c}\pd{I}{t}+\mu\pd{I}{z} = \mS(I),
\end{equation}
where $c$ is the speed of light, $I=I(z,t,\mu)$ is the \emph{specific intensity}
of radiation, and $\mu \in [-1, 1]$ is the velocity related variable such that
$\arccos(\mu)$ represents the angle between the photon velocity and the
$z$-axis. The right hand side $\mS(I)$ denotes the actions by the background
medium on the photons. Here we adopt a common form of $\mS(I)$ given in
\cite{Bru02, McClarren2008Semi} as
\begin{equation} \label{eq:source}
    \mS(I) = -\sigma_t I + \frac{1}{2}ac\sigma_a T^4 
    + \frac{1}{2}\sigma_s \int_{-1}^1 I \dd\mu + \frac{s}{2}, 
\end{equation}
where $a$ is the radiation constant, and $s=s(z)$ is an isotropic external
source of radiation.
The scattering coefficient $\sigma_s$ and the absorption coefficient $\sigma_a$
depend on the position $z$ and the material temperature $T(z,t)$, and the total
opacity coefficient is $\sigma_t=\sigma_a+\sigma_s$.
Denote the $k$-th moment of the specific intensity by
\begin{equation} \label{eq:momentsdefine} 
    \moment{I}_k \triangleq \int_{-1}^1 \mu^{k} I(\mu)\dd\mu, \quad k \in \bbN.
\end{equation} 
For simplicity of notations, in \eqref{eq:momentsdefine} and the following
discussion, the explicit dependence of the specific intensity and the moments
on spatial coordinate and time has been suppressed (i.e., ${\moment{I}_k=
\moment{I}_k(z,t)}$, $I(\mu)=I(z,t,\mu)$).
The evolution equation of the internal energy $e$ of the background medium is
\begin{equation} \label{eq:internalenergy}
    \pd{e}{t} = \sigma_a\left(\moment{I}_0 - acT^4\right).
\end{equation}
The relationship between the temperature $T$ and the internal energy $e$ is
problem dependent. We will assign it in the numerical examples when necessary.

For the whole system, in the absence of any external source of radiation
i.e. $s=0$, the total energy $e+\moment{I}_0/c$ is conserved:
\begin{equation}\label{eq:conservation}
    \pd{e}{t}+\frac{1}{c}\pd{\moment{I}_0}{t} +\pd{\moment{I}_1}{z}=0.
\end{equation}


Moreover, for later usage, we introduce the equilibrium of the RTE as
\begin{equation}\label{eq:equilibrium}
    I_{\eq}=\frac{\moment{I}_0}{2}.
\end{equation}

\subsection{Moment method for RTE}
Multiplying \eqref{eq:radiativetransfer} by $\mu^k$ and integrating it with
respect to $\mu$ over $[-1,1]$ yields the following equations
\begin{equation} \label{eq:momentequations} 
    \frac{1}{c}\pd{\moment{I}_k}{t}
    + \pd{ \moment{I}_{k+1}}{z} = \moment{\mS(I)}_k, \quad k \in \bbN,
\end{equation}
where 
\[
    \moment{\mS(I)}_k = -\sigma_t \moment{I}_k
    +\frac{1-(-1)^{k+1}}{2k+2}(\sigma_a acT^4 + \sigma_s \moment{I}_0+s).
\]
In \eqref{eq:momentequations} the governing equation of $\moment{I}_k$ also
depends on $\moment{I}_{k+1}$, which means that the full system 
contains infinite number of equations.
To derive a moment model for \eqref{eq:radiativetransfer}, we first truncate
the system by discarding all the governing equations of $\moment{I}_k$,
$k>N$, for a given integer $N \in \bbN$. Clearly, the truncated system is
not closed, due to its dependence on the ($N+1$)-th moment $\moment{I}_{N+1}$,
thus we have to apply a so-called \emph{moment closure} to this system. 
Normally, the moment closure is to find an approximation for $\moment{I}_{N+1}$
formulated as
\begin{equation}\label{eq:momentclosure}
    \moment{I}_{N+1} \approx E_{N+1} = E_{N+1}(\moment{I}_0, \cdots,
    \moment{I}_N).
\end{equation}
To achieve this goal, a popular method is to construct an ansatz for the
specific intensity. Precisely, let $E_k$, $k=0, \cdots, N$, be the known
moments for a certain unknown specific intensity $I$. Then one proposes an
approximation $\hat{I}(\mu;E_0,\dots,E_N)$, which is called \emph{ansatz}
of $I$, such that
\begin{align} \label{eq:momentrelation}
    \moment{\hat{I}(\cdot;E_0,\dots,E_N)}_k = E_k, \quad k=0, \cdots, N,
\end{align}
and meanwhile $\hat{I}$ is uniquely determined by \eqref{eq:momentrelation}.
For the ($N+1$)-th moment of $I$, it is then directly approximated by
the ($N+1$)-th moment of $\hat{I}$, i.e.,
\begin{equation}\label{eq:closure}
    E_{N+1} = \moment{\hat{I}(\cdot; E_0, \cdots, E_N)}_{N+1}.
\end{equation}
The resulting moment system is
\[
  \frac{1}{c} \pd{E_k}{t} + \pd{E_{k+1}}{z} = \langle \mS(\hat{I})
  \rangle_k, \quad k = 0, \cdots, N.
\]
%

In the following, we briefly review the two most popular moment models
for the RTE: 
the \PN model \cite{jeans1917stars} and 
the \MN model \cite{dubroca1999theoretical, levermore1996moment,
minerbo1978maximum}, by specifying a
certain ansatz $\hat{I}$.

\subsubsection{\PN model}
The \PN model is a counterpart of the Grad's moment method \cite{Grad} in the
RTE. It expands the specific intensity around the equilibrium in terms of
orthogonal polynomials with respect to $\mu$.
For the RTE, the normalized equilibrium is a constant, hence the corresponding
orthogonal polynomials are the Legendre polynomials. Denote
the monic Legendre polynomial of degree $m$ by $P^{(m)}(\mu)$, then the 
ansatz for the \PN model, denoted by $\hat{I}_P$, is
\begin{equation}\label{eq:PnAnsatz}
  \hat{I}_P(\mu; E_0, \cdots, E_N) = \sum_{m=0}^N f_m(E_0, \cdots, E_N) P^{(m)}(\mu).
\end{equation}
Due to the moment constraints \eqref{eq:momentrelation}, we have 
\begin{equation}\label{eq:constrains}
  \sum_{m=0}^N \langle P^{(m)} \rangle_k f_m = E_k, \quad k=0, \cdots, N.
\end{equation}
Hence, the expansion coefficients $f_m$, $m=0,\dots,N$ are uniquely determined
by the moments $E_k$, $k=0,\dots,N$ through the linear system \eqref{eq:constrains}.
The moment closure is then given as
\[
  E_{N+1} = \sum_{m=0}^N f_m \langle P^{(m)} \rangle_{N+1},
\]
which is a linear function of $E_k$, $k=0,\dots,N$.


The \PN model is widely used in the numerical simulations of RTE due to its 
good mathematical properties.
For example, it has a simple analytical form such that the evaluation of the
flux is fast. Its system can be transformed into a symmetric hyperbolic
system \cite{Fischer1972The}.
Moreover, the \PN model formally converges in an $L^2$ setting to the solution
of the transport equation as $N \rightarrow \infty$ \cite{Davison1960on}.

Meanwhile, the \PN model also suffers some drawbacks. Its solution may have
undesirable oscillations, which may lead to negative particle concentration
\cite{brunner2005two, mcclarren2008solutions}.  Its approximation rate to the
RTE is low, such that a lot of moments are required in numerical simulations
\cite{Benassi1984High,Eshagh2009On}. 

\subsubsection{\MN model}
The \MN model takes the solution of the entropy minimization problem, denoted
by $\hat{I}_M$, to close the system. Precisely, we have
\begin{equation} \label{eq:entropy-basedequation}
  \begin{array}{rrl}
      \hat{I}_M = & \text{argmin} & \displaystyle\int_{-1}^1 \eta(\hat{I}) \dd\mu, \\
                  & \text{s.t.} & \moment{\hat{I}}_k = E_k, \quad k=0,1,\cdots,N.
  \end{array}
\end{equation}
Here $\eta: \bbR \rightarrow \bbR$ is the Bose-Einstein entropy
\begin{equation}\label{eq:entropy}
    \eta(f) = f\log f - (1+f)\log (1+f),
\end{equation}
for photon. The intensity takes the form
\cite{levermore1996moment,dubroca1999theoretical}
\begin{equation}\label{eq:Mnansatz_general}
  \frac{2\hbar\nu^3}{c^2}
  \left( \exp\left( \frac{\hbar\nu}{k_B}\sum_{i=0}^{N}\alpha_i\mu^i 
  \right) -1\right)^{-1},
\end{equation}
where $\alpha_i$, $i=0,\dots,N$ are the Lagrange multipliers to be determined,
$\hbar$ is the Planck constant, and $k_B$ is the Boltzmann constant.
Integrating \eqref{eq:Mnansatz_general} with respect to the frequency $\nu$ on
$[0,\infty]$, we obtain the intensity for the grey medium case
\begin{equation}\label{eq:Mnansatz}
  \hat{I}_M = \frac{\sigma}{\left(\sum\limits_{i=0}^N \alpha_i \mu^i\right)^{4}},
\end{equation}
where $\sigma$ is the Stefan-Boltzmann constant. Clearly, the Lagrange
multipliers $\alpha_i$, $i=0,\dots,N$ are uniquely determined by the moments
$E_k$, $k=0,\dots,N$, and the moment closure ${E_{N+1} = \langle \hat{I}_M
\rangle_{N+1}}$ follows.

Particularly, if $N=0$, the ansatz is the equilibrium, i.e.,
\begin{equation}
    \hat{I}_M=I_{\eq} =\frac{E_0}{2}.
\end{equation}
If $N=1$, the Lagrange multipliers $\alpha_i$, $i=0,1$ can be directly solved,
and the solution of the minimizing entropy problem is given as
\begin{equation}\label{eq:M1model}
  \hat{I}_M(\mu)= E_0 
  \frac{\varepsilon}{(1+\alpha\mu)^4},
\end{equation}
where 
\begin{align} \label{eq:alphaandsigma}
    \alpha = -\frac{3E_1/E_0}{2+\sqrt{4-3(E_1/E_0)^2}},\quad
    \varepsilon = \frac{3(1-\alpha^2)^3}{2(3+\alpha^2)}.
\end{align}
The corresponding moment closure is
\begin{equation*}
    E_2 = E_0\frac{3+4(E_1/E_0)^2}{5+2\sqrt{4-3(E_1/E_0)^2}}.
\end{equation*}

However, for $N\geq2$, there is no algebraic expression of the Lagrange
multipliers $\alpha_i$ with respect to the moments $E_k$. Thus, an expensive
iterative procedure is required to solve the Lagrange multipliers. This drawback
of the \MN model strongly limited its applications with $N\geq2$, though it has
been demonstrated that the \MN models yield promising results
\cite{Hauck2011high}.

As for the properties of the system, the \MN model retains many fundamental
properties from the kinetic formalism. The characteristic speed of this model is 
no larger than the speed of light, 
which agrees with the
fact that information cannot travel faster than the speed of light. The ansatz
$\hat{I}_M$ is always positive, the \MN model is equipped with entropy, and the
resulting system of equations can be transformed into a symmetric hyperbolic
system.

For the \MN model, though the drawback in the numerical simulations hinders its
application, its praised properties motivate researchers to construct new moment
models.  The moment model developed in the next section is partially inspired by
the \MN model.


\subsubsection{Relationship between \PN and \MN model}\label{sec:PN_MN}
In this subsection, we propose two points of view --- the
entropy-based viewpoint and the weighted expansion viewpoint --- on the
relationship between the \PN and \MN model, and show our motivation on the novel
moment model developed in this paper. All the discussion in this subsection is
formal, not rigorous.

\paragraph{Entropy-based viewpoint}
If the specific intensity $I$ is close to its equilibrium, i.e., $I\sim
I_{\eq}$, using the Taylor expansion on $\ln(1+x)$, we can approximate the
Bose-Entropy entropy \eqref{eq:entropy} as
\begin{equation}
    \begin{aligned}
        \eta(I)\approx \eta_{app}(I)\triangleq &
        I_{\eq}\ln(I_{\eq})+(1+\ln(I_{\eq}))(I-I_{eq})
        - (1+I_{\eq})\ln(1+I_{\eq})\\
        &-(1+\ln(1+I_{\eq}))(I-I_{eq})
        +\frac{(I-I_{\eq})^2}{2I_{\eq}(1+I_{\eq})},
    \end{aligned}
\end{equation}
by discarding high order term with respect to $O((I-I_{\eq})^3)$. Let
$\eta_{\scaleto{P}{4pt}}(I)=\eta_{app}(I)$ or equivalently
$\dfrac{I^2}{I_{\eq}}$, then one can easily check that the ansatz of the \PN
model can be obtained by minimizing $\eta_{\scaleto{P}{4pt}}$ as
\begin{equation} 
  \begin{array}{rrl}
      \hat{I}_P = & \text{argmin} & \displaystyle\int_{-1}^1 \eta_P(\hat{I}) \dd\mu, \\
                  & \text{s.t.} & \moment{\hat{I}}_k = E_k, \quad k=0,1,\cdots,N.
  \end{array}
\end{equation}
This indicates that the \PN model can also be brought into the framework of the
``entropy'' minimization problem by choosing a proper object function $\eta$.
It motivates us to construct new moment models by selecting a different object
function.
On the other hand, the start point of the \PN model is expanding the velocity
distribution around the equilibrium, where the distribution is implicitly
assumed to be close to the equilibrium. It is consistent with the assumption in the
approximation of $\eta_{app}$.

\paragraph{Weighted expansion viewpoint}
For the \PN model, the approach to construct its ansatz can be generalized as
an expansion of the velocity distribution function around a weight function
$\omega(\mu)$ in terms of orthogonal polynomials.
Precisely, given a weight function $\omega(\mu)$ satisfying
\begin{equation}
    \omega(\mu)\geq0,\quad \int_{-1}^1\omega(\mu)\dd\mu=1,
\end{equation}
we denote the monic orthogonal polynomial of degree $m$ by $\phi_m(\mu)$ such
that
\begin{equation}
    \int_{-1}^1\omega(\mu)\phi_m(\mu)\phi_n(\mu)\dd\mu=\delta_{mn}c_m,
\end{equation}
where $\delta$ is the Kronecker delta and $c_m$ are non-zero constants.
It is worth to point out that the weight function $\omega$ is allowed to depend
on the moments $E_k$.
Then we construct the ansatz as
\begin{equation}\label{eq:ansatz_gene}
    \hat{I}=\omega(\mu)\sum_{m=0}^Nf_m(E_0,\dots,E_N)\phi_m(\mu).
\end{equation}
Due to the moment constraints \eqref{eq:momentrelation}, we obtain the system
\begin{equation}\label{eq:momentconstrains}
    \sum_{m=0}^N \int_{-1}^1\omega(\mu)\phi_m(\mu)\mu^k f_m \dd\mu 
    = E_k, \quad k=0, \cdots, N,
\end{equation}
which uniquely determines the ansatz \eqref{eq:ansatz_gene}.

For the \PN model, the weight function is $\omega_P= 1 / 2$, and the orthogonal
polynomials are the Legendre polynomials.
For the \MN model, the weight function is the normalization of the ansatz
$\hat{I}_M$ itself, and the coefficients satisfy $f_m=0$, $m=1,\dots,N$.
To check it is not difficult but rather complex. We refer readers to
\cite[Section 5.4]{Fan2015} and \cite[Section 4.2.2]{framework} for details.

In the above, we also bring the \MN model in the framework of the \PN model.
The key point of this framework is the weight function. Once the weight
function is given, one can directly obtain the corresponding moment model
following the above routines.

These viewpoints build a bridge between the \PN and \MN models and also
present two methods to construct new models. 
The following part of this paper will focus on the new model and its analysis.

%% file: M1PNmodel.tex
\section{$M_1$-based moment model} \label{sec:M1Pndeduction}
The viewpoints in \cref{sec:PN_MN} provide methods to construct new moment
models. The remaining issue is how to choose the weight function. A weight
function, which contains much information of the moments, usually resulting
in a strong nonlinear system, for instance, the \MN model. Such a system is
expected to have a good approximation to the RTE, but the evaluation of its flux is,
in general, expensive. Hence, one has to make a trade-off between the numerical
efficiency and the approximation rate.

\subsection{$M_1$-based moment system}\label{sec:MPNmodel}
The ansatz of the \MN model with $N=1$ (we call it \Mone hereafter) contains the
zeroth-order moment $E_0$ and the first order moment $E_1$. This brings it the
capability to approximate an anisotropic distribution. If we take the ansatz of
the \Mone model as the weight function to construct a new model, the
corresponding ansatz is expected to have a better approximation on the
anisotropic distributions than that of the \PN model.
In this section, based on this idea, we develop a novel moment model and study
its properties in detail.

Let the weight function be
\begin{equation}\label{eq:weight}
    \weight(\mu)= \frac{\varepsilon}{(1+\alpha\mu)^4}, 
    \quad \alpha\in(-1,1), \quad
    \text{such that} \quad 
    \int_{-1}^1\weight(\mu)\dd\mu=1,
\end{equation}
so $\varepsilon = \frac{3(1-\alpha^2)^3}{2(3+\alpha^2)} > 0$.
Here $\alpha$ is a parameter to be determined.
Denote by $\pl_j(\mu)$ the monic orthogonal polynomial of $\mu$ with degree $j$
respect to the weight function $\weight(\mu)$.
We introduce the moments of the weight function and the inner product as
\begin{equation}
    \begin{aligned}
        \mE_k & \triangleq \moment{\weight(\mu)}_k,\quad
        \inner{f}{g} &\triangleq \int_{-1}^1 f(\mu)g(\mu)\weight(\mu) \dd\mu.
    \end{aligned}
\end{equation}
%
%
Then by Gram-Schmidt orthogonalization, the polynomials $\pl_j(\mu)$
are obtained recursively as
\begin{align}\label{eq:Gram-Schmidt}
  \pl_0(\mu) = 1, \quad \pl_j(\mu) = \mu^j - \sum_{k=0}^{j-1}
  \frac{A_{j,k}} {A_{k,k}} \pl_k(\mu), \quad j\geq 1,
\end{align}
where $A_{j,k} = \inner{\mu^j}{\pl_k(\mu)}$. The orthogonality of
$\pl_j$ indicates that
\[
    A_{k,k} = \inner{\pl_k}{\pl_k}, \quad A_{k,j}=0,\ k<j.
\]
Applying the inner product on \eqref{eq:Gram-Schmidt} and $\mu^i$ yields
the relationship between $A_{i,j}$ and $\mE_{k}$ as
\[
    A_{0,0}=1,\quad
  A_{i,j}=\mathcal{E}_{i+j} - \sum_{k=0}^{j-1}
  \frac{A_{j,k} A_{i,k}}{A_{k,k}}, \quad 1\leq j\leq i.
\]
The corresponding ansatz $\hat{I}$ is defined as
\begin{equation}\label{eq:ansatz}
    \hat{I}(\mu; E_0, \cdots, E_N) \triangleq \sum_{i=0}^N f_i \Pl_i(\mu), 
\end{equation}
where $\Pl_i(\mu)=\pl_i(\mu)\weight(\mu), i = 0,1,\cdots,N$ 
are the basis functions, and $f_i$ are the expansion coefficients 
to be determined by the moment
constraints \eqref{eq:momentconstrains}. Thanks to the orthogonality of $\pl_i$,
we have
\[
  f_i = \frac{1}{A_{i,i}} \int_{-1}^1 \pl_i(\mu)\hat{I}(\mu)\dd\mu.
\]
Substituting the recursive relationship \eqref{eq:Gram-Schmidt} into the upper
equation yields the following recursive formulation for $f_i$, which are
functions dependent on $E_i$,
\begin{equation}\label{eq:fk_Ek}
    f_0 = E_0,\quad
    f_i = \frac{1}{A_{i,i}} \left(E_i - \sum_{j=0}^{i-1}
    A_{i,j} f_j\right),\quad 0\leq i\leq N.
\end{equation}
Thus the explicit formation of the ansatz \eqref{eq:ansatz} is obtained. 
The moment closure is then given as
\[ 
E_{N+1} = \sum_{k=0}^N f_k A_{N+1,k}.
\]

If $\alpha\equiv 0$, then $\weight=1 / 2$, the orthogonal polynomials are the
Legendre polynomials, and the resulting system is the \PN model.
If we set 
\begin{equation}
  \label{eq:alphachoose}
    \alpha = -\frac{3E_1/E_0}{2+\sqrt{4-3(E_1/E_0)^2}},
\end{equation}
the weight function is the normalization of the ansatz of the \Mone model.
Some calculations yield
\begin{equation}
    f_0=E_0,\quad f_1=0,\quad
    \mathcal{E}_0 = 1, \quad \mathcal{E}_1 = \frac{E_1}{E_0}.
\end{equation}
In the following part of the paper, we always use this setup.

Notice that the moment model uses the ansatz of the \Mone model as the weight
function, and generates the arbitrary order models following the idea of the \PN
model. We call the moment model as the \emph{\MPN}model in the following.

In the \cref{sec:PN_MN}, we have put the \PN model into the framework of the
entropy minimization form with the corresponding object function
$\eta_{\scaleto{P}{4pt}}(I)=\frac{I^2}{I_{\eq}}$.
For the \MPN model, let $\eta_{\omega}=\frac{I^2}{\weight}$, then one can easily
check that the ansatz \eqref{eq:ansatz} is also the minimizer of the following
problem
\begin{equation} 
  \begin{array}{rrl}
      & \text{argmin} & \displaystyle\int_{-1}^1 \eta_{\omega}(\hat{I}) \dd\mu, \\
      & \text{s.t.} & \moment{\hat{I}}_k = E_k, \quad k=0,1,\cdots,N.
  \end{array}
\end{equation}

\begin{figure}[ht]
    \centering
    \subfloat[specific intensity of $I^{(1)}$ in \eqref{eq:dis2}]{
        \includegraphics[clip, width=0.4\textwidth]{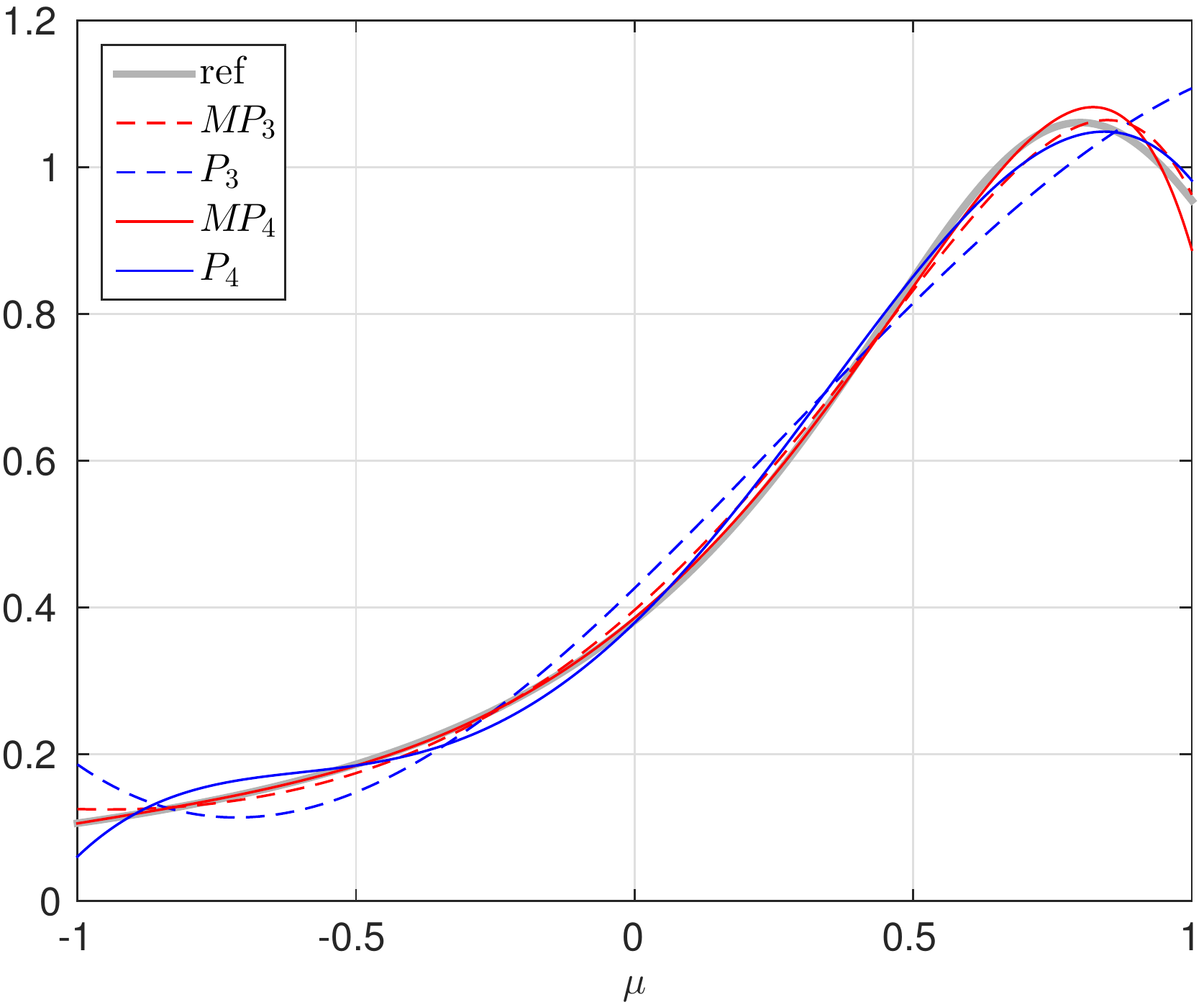}
    }\qquad\qquad
    \subfloat[specific intensity of $I^{(2)}$ in \eqref{eq:dis2}]{
        \includegraphics[trim={2mm 0mm 0mm 0mm},clip, width=0.4\textwidth]{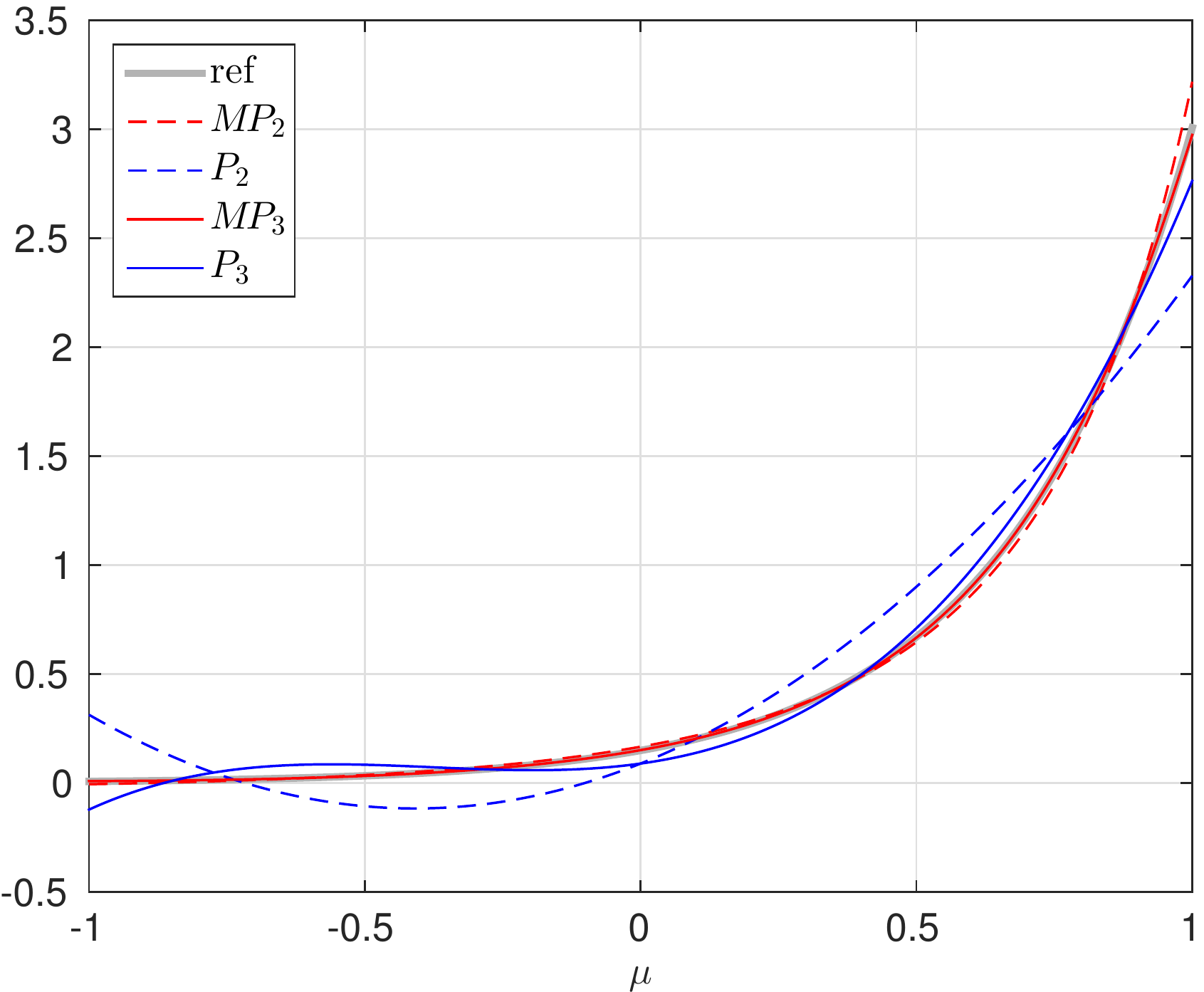}
    }
    \caption{Profile of the ansatz of the \MPN and \PN models for two anisotropic 
    distributions.}
    \label{fig:distribution}
\end{figure}

We compare the approximation efficiency of the \MPN model and the \PN model, and
select the following two intensity as examples:
\begin{align}
    I^{(1)} = \frac{6}{5\pi} \frac{1}{1-\frac{8}{5}\mu+\mu^2},\qquad
    I^{(2)} = \frac{3}{2\sinh(3)} \exp(3\mu).\label{eq:dis2}
\end{align}
The corresponding profiles of the intensity are presented in
\cref{fig:distribution}. Clearly, the ansatz of the \MPN model shows advantage
in the approximating such anisotropic distributions, because of its specific
weight function.


%% file: analysis.tex
\subsection{\MPtwo model} \label{sec:M1P2Model}
The complex form of \MPN \eqref{eq:fk_Ek} makes it not easy to investigate the
\MPN model with arbitrary order.
Here we provide a perspective on the \MPN model by studying the simplest
non-trivial case, i.e., $N=2$ case in detail.

\subsubsection{\MPtwo moment system}
Direct calculation on \eqref{eq:Gram-Schmidt} yields
the orthogonal polynomials 
\[  
  \begin{aligned}
    \pl_0(\mu) &= 1,\\
    \pl_1(\mu) &=\mu - \mE_1,\\
    \pl_2(\mu) &= \mu^2-\mE_2 - 
    \beta(\mu-\mE_1),
    \quad
    \beta = \frac{\mE_3-\mE_2\mE_1} {\mE_2-\mE_1^2}.
  \end{aligned}
\]
The expansion coefficients are given as
\begin{equation}
  \label{eq:expansioncoefficients}
  f_0 = E_0,\quad f_1=0, \quad
  f_2=\frac{E_2-E_0\mE_2}{\mE_4-\mE_2^2-\beta(\mE_3 -\mE_1\mE_2)},
\end{equation}
and the ansatz $\hat{I}(\mu)$ is
\begin{equation}\label{eq:ansatzM1P2}
  \hat{I}(\mu) = \weight(\mu) (f_0 + f_2\pl_2(\mu)).
\end{equation}
Then, the moment closure is directly obtained by the third moment of the
intensity
\begin{equation} \label{eq:E3ofM1P2}
  E_3 = \langle \hat{I}(\mu) \rangle_3 = E_0\mE_3
  + f_2 \left( \mE_5 - \mE_2\mE_3 -\beta\left(\mE_4 - \mE_1\mE_3\right) \right).
\end{equation}

For the RTE, the positivity of the specific intensity $I$ provides constraints on
the moments. Based on this, we introduce 
\emph{realizable domain} for the RTE. 
\begin{definition} 
    The realizable domain is 
    the set of moments where each point corresponds to a 
    positive intensity, i.e.,
    \begin{equation}
        \Omega_R \triangleq \{ (E_0,E_1,E_2)^T : \exists I(\mu)> 0, \langle I
        \rangle_k = E_k, k=0,1,2 \}.
    \end{equation}
\end{definition}
By the Cauchy-Schwarz inequality, $|E_1|< E_0$, and $E_1^2< E_0 E_2$ have to be
fulfilled for a positive intensity $I$. Direct calculations indicate
that the realizable region is
\begin{equation} \label{eq:realizablityspace}
  \Omega_R = \{(E_0,E_1,E_2)^T : E_0>0, E_2< E_0, E_1^2< E_0E_2\}.
\end{equation}
Hereafter we focus on the \MPtwo model in the
realizable domain $\Omega_R$.
Let $\sgn(x)$ be the sign function 
\begin{equation*}
    \sgn(x)=\left\{ \begin{array}{ll}
        -1, &   x<0,\\
        0, &   x=0,\\
        1, &   x>0.
    \end{array} \right.
\end{equation*}
Then we have the following properties on the 
closure of \MPtwo.
\begin{property}\label{pro:M1P2}
    For the \MPtwo moment system, the closed moment $E_3=E_3(E_0,E_1,E_2)$,
    $(E_0,E_1,E_2)\in\Omega_R$, satisfies the following properties: 
    \begin{enumerate}
        \renewcommand{\theenumi}{\alph{enumi})}
        \item $\sgn\left(\pd{E_3}{E_0}\right)=-\sgn(E_1)$; \label{itm:E3E0}
        \item $\pd{E_3}{E_1}\Big|_{E_1=0}>0$; \label{itm:E3E1}
        \item $\sgn\left(\pd{E_3}{E_2}\right)=\sgn{(E_1)}$; \label{itm:E3E2_1}
        \item $E_3$ is linear dependent on $E_2$, i.e.  $\pd{^2E_3}{E_2^2}=0$;\label{itm:E3E2}
        \item $\sgn\left(\mE_3-\mE_1^3-(\mE_2-\mE_1^2)\pd{E_3}{E_2}\right)=\sgn(E_1)$.
            \label{itm:E3geE1}
    \end{enumerate}
\end{property}
These properties depict the behavior of the closed moment $E_3$. One can
directly check these properties with the help of \eqref{eq:E3ofM1P2} and
\cref{fig:comparsionE3}. But the rigorous proof is tedious, and we put it in the
\cref{sec:proof_M1P2}.

\begin{figure}[ht]
  \centering
  \subfloat[$P_2$ model]{
  \includegraphics[trim={2mm 0mm 14mm 8mm}, clip, width=0.30\textwidth]{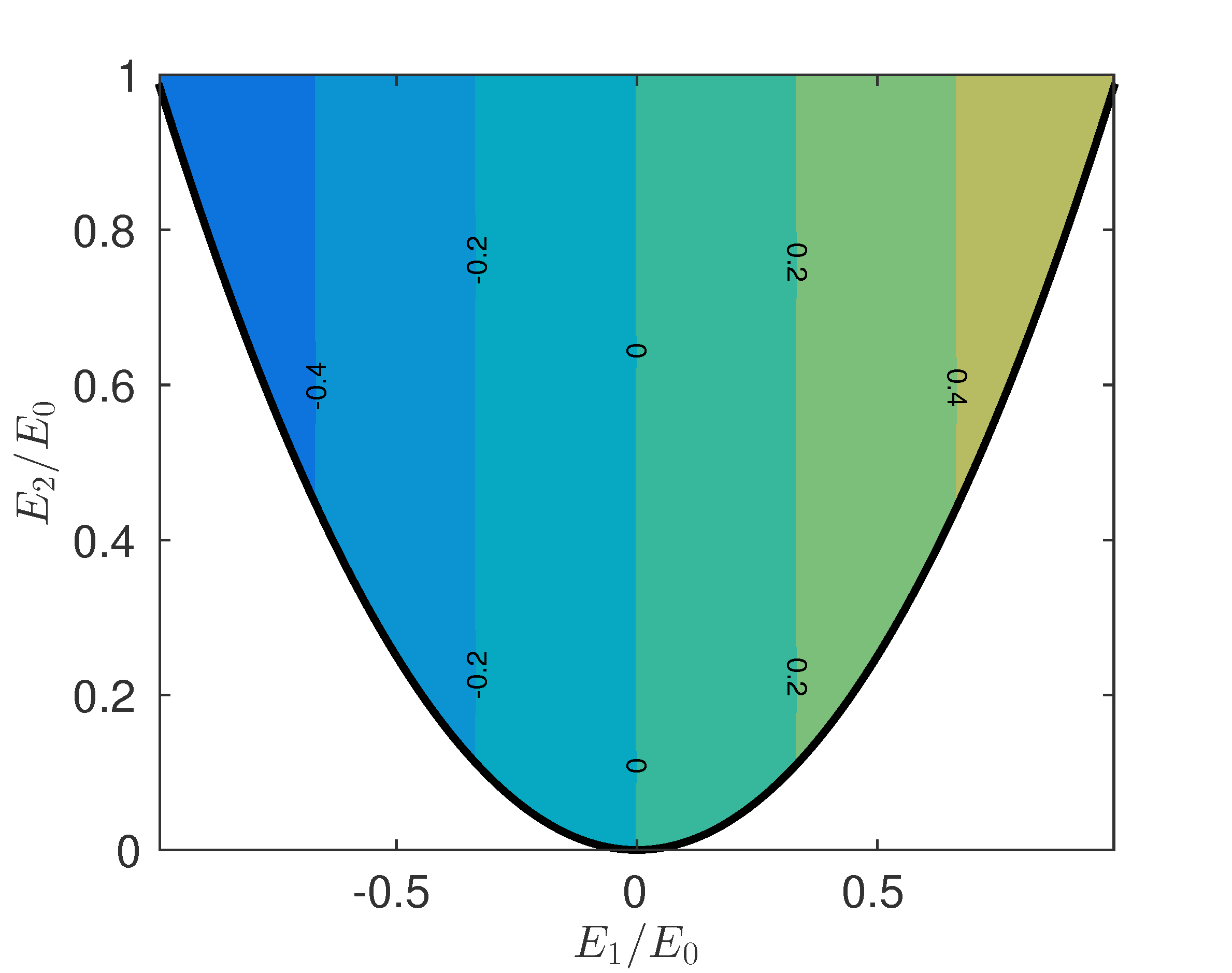}
  }\quad
  \subfloat[\MPtwo model]{
  \includegraphics[trim={2mm 0mm 14mm 8mm}, clip, width=0.30\textwidth]{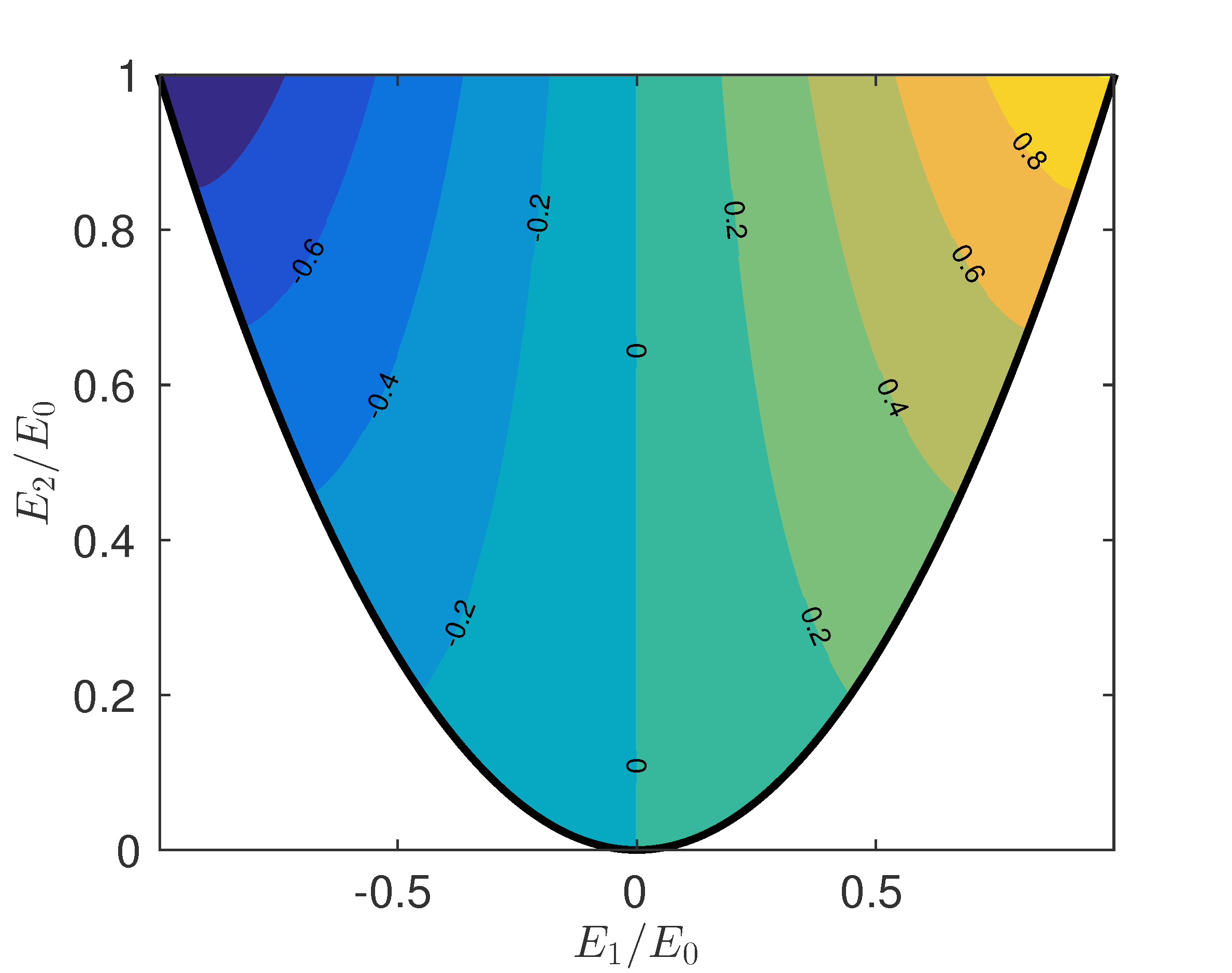}
  }\quad
  \subfloat[$M_2$ model]{
  \includegraphics[trim={0mm 0mm 14mm 8mm}, clip, width=0.30\textwidth]{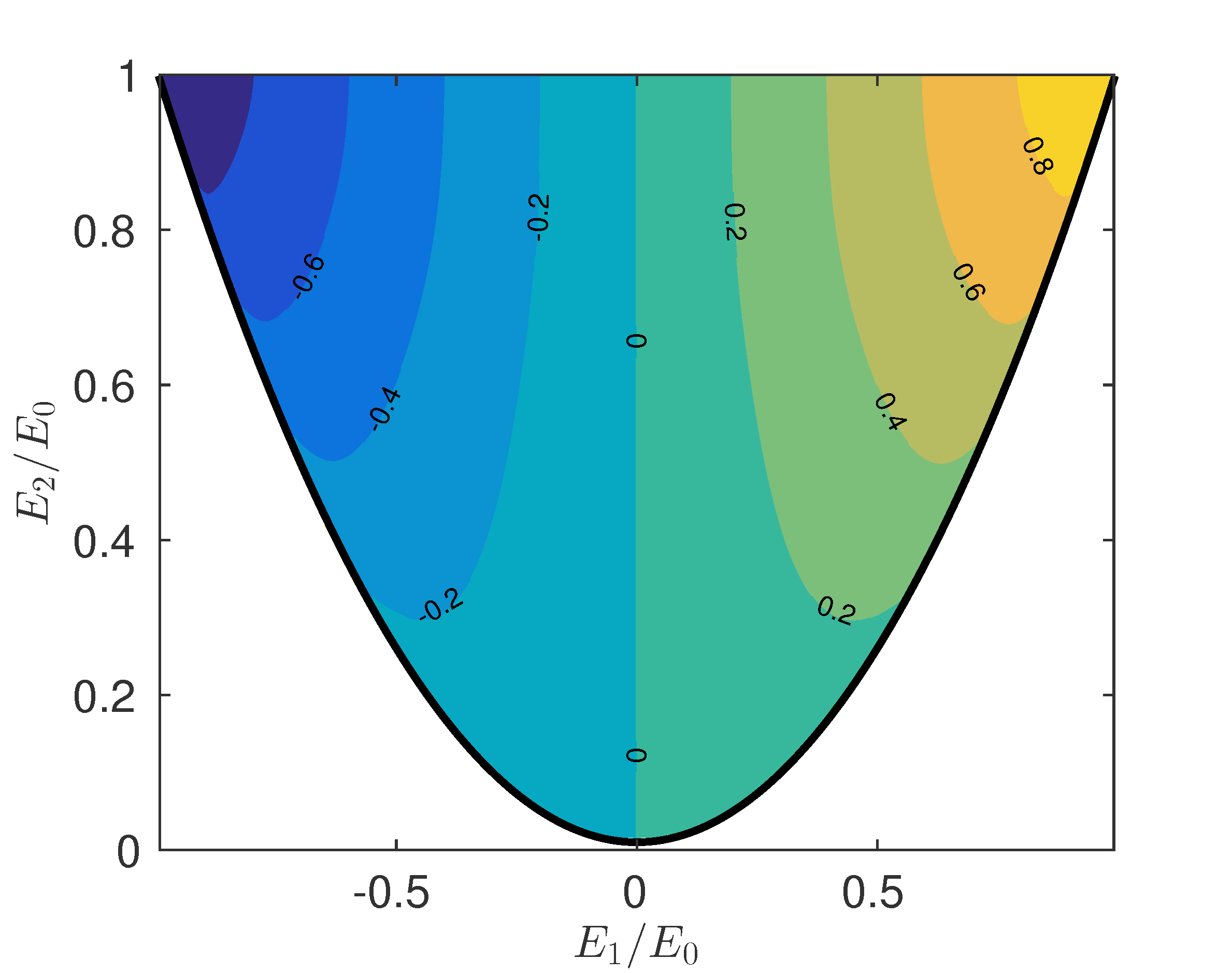}
  }
  \caption{\label{fig:comparsionE3}
      Contour of $\frac{E_3}{E_0}$ of the $P_2$ model, \MPtwo model, and $M_2$
      model.  The range of $\frac{E_3}{E_0}$ for the \MPtwo and $M_2$ model is
      $(-1,1)$ in the realizable domain, while that for the $P_2$ model is
      $(-3/5,3/5)$.
      }
\end{figure}
\begin{figure}[ht]
  \centering
  \subfloat[$\frac{E_2}{E_0}=1/5$]{
  \includegraphics[width=0.32\textwidth,height=0.16\textheight]{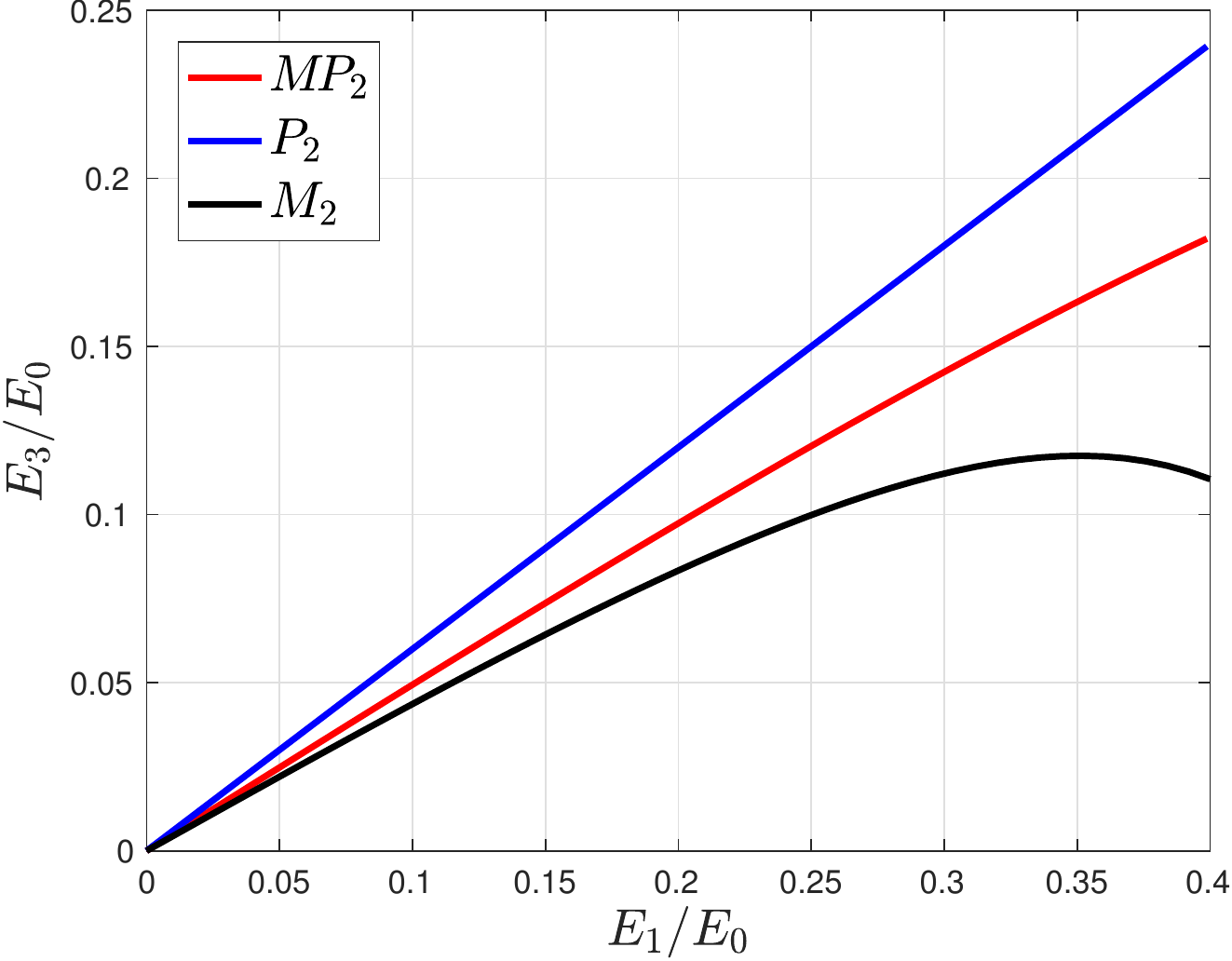}
  }
  \subfloat[$\frac{E_2}{E_0}=1/3$]{
  \includegraphics[width=0.32\textwidth,height=0.16\textheight]{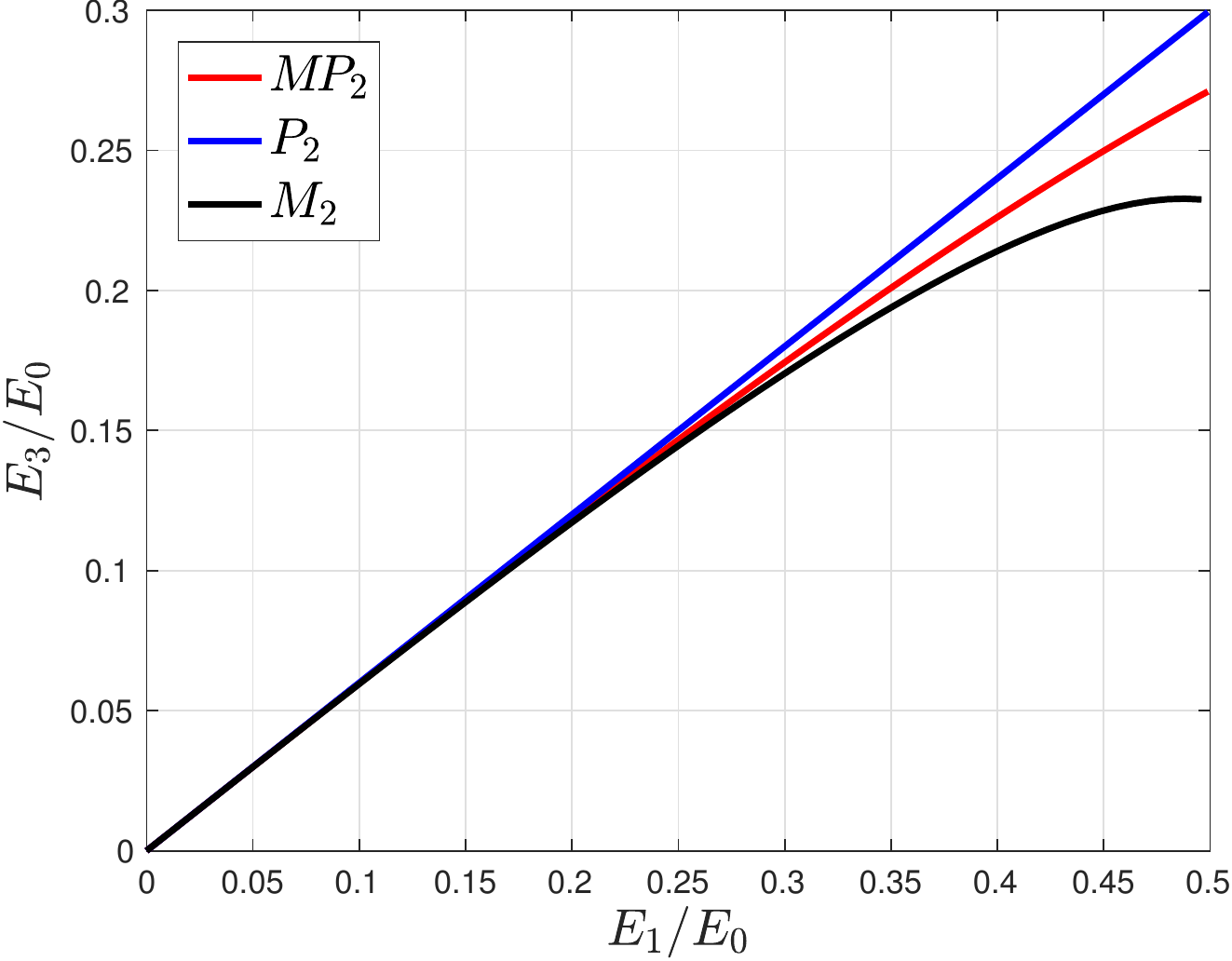}
  }
  \subfloat[$\frac{E_2}{E_0}=81/100$]{
  \includegraphics[width=0.32\textwidth,height=0.16\textheight]{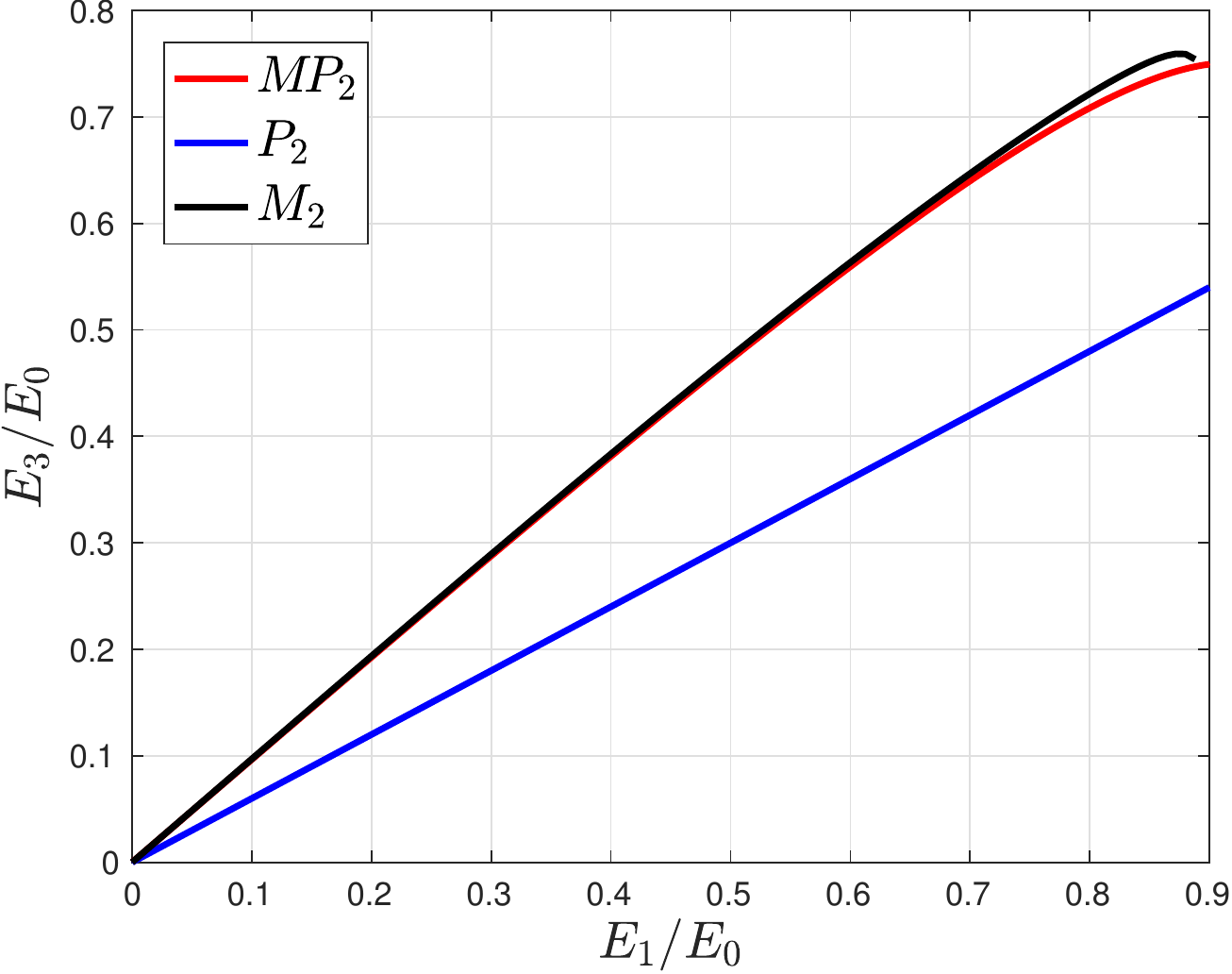}
  }
  \caption{\label{fig:comparsionE3_E2}
  Profile of $E_3/E_0$ with respect to $E_1 / E_0$ for given $E_2 / E_0$ for the
  $P_2$ model, the \MPtwo model, and the $M_2$ model.
  }
\end{figure}

Now we compare the \MPtwo model with the $P_2$ model and the $M_2$ model.
\Cref{fig:comparsionE3} presents the contour of $\frac{E_3}{E_0}$ in $\Omega_R$
for the three models, and \cref{fig:comparsionE3_E2} presents some cross
sections with respect to $\frac{E_2}{E_0}=1/5,1/3,81/100$.
The $P_2$ model is a linear system, hence its closed moment $E_3$ is linearly
dependent on the moments $E_0$, $E_1$, and $E_2$. The $M_2$ model is a nonlinear
system due to its complex ansatz. The \MPtwo model is also a nonlinear system
because of its weight function, which contains the information of $E_0$ and
$E_1$. Compared with the $P_2$ and $M_2$ model, the \MPtwo model falls in the
middle.
In the sense of the closed moment, the \MPtwo can be treated as an approximation
of the $M_2$ model.

\subsubsection{Hyperbolicity}
Denote the relevant moments and the flux by
\begin{align} \label{eq:moments}
  \bU = (E_0, E_1, E_2)^T,
  \quad
  \bF(\bU) = (E_1,E_2,E_3)^T,
\end{align}
then the moment equation is given by 
\begin{equation}\label{eq:M1P2_vec}
    \frac{1}{c} \pd{\bU}{t} + \pd{\bF(\bU)}{z} = \bm S,
\end{equation}
with $\bS=(S_0, S_1, S_2)^T$, and $S_k = \langle \mS(I) \rangle_k$.
We declare the main conclusion of this subsection in the following theorem.
\begin{theorem} \label{thm:hyperbolic}
    The \MPtwo system \eqref{eq:M1P2_vec} is hyperbolic for any $\bU\in\Omega_R$.
\end{theorem}
Before the proof of the theorem, we introduce the following notations and lemma.
The Jacobian of the $M_1$-$P_2$ model is
\begin{equation} \label{eq:JacobianofM1P2}
  \bJ = \pd{\bF}{\bU} = 
  \begin{pmatrix}
      0             & 1             & 0 \\
      0             & 0             & 1 \\
      \pd{E_3}{E_0} & \pd{E_3}{E_1} & \pd{E_3}{E_2}
  \end{pmatrix}.
\end{equation}
Its characteristic polynomial is
\begin{equation} \label{eq:cpofM1P2} 
    p(\lambda)=\lambda^3 - \pd{E_3}{E_2} \lambda^2 - \pd{E_3}{E_1} \lambda -
    \pd{E_3}{E_0},
\end{equation} 
which satisfies the following property:
\begin{lemma} \label{lem:plambda}
    For any $\bU \in \Omega$, we have
    \begin{equation}\label{eq:pE1E0}
    \sgn\left( p\left( \frac{E_1}{E_0} \right) \right)=-\sgn(E_1).
    \end{equation}
\end{lemma}
\begin{proof}
    If $E_1=0$, then $p(0)=-\pd{E_3}{E_0} = 0$ thanks to
    \cref{pro:M1P2}\ref{itm:E3E0}. So \eqref{eq:pE1E0} holds.

    Next we assume $E_1\neq0$. Noticing that $\dfrac{E_3}{E_0}$ is a function
    of $\dfrac{E_1}{E_0}$ and $\dfrac{E_2}{E_0}$, we have
    \[
        \pd{E_3}{E_0} = \frac{E_3}{E_0} - \pd{E_3}{E_2}\frac{E_2}{E_0}
        -\pd{E_3}{E_1}\frac{E_1}{E_0}.
    \]
    Therefore, $p\left(\dfrac{E_1}{E_0}\right)$ can be simplified as 
    \[
        p\left(\frac{E_1}{E_0}\right) = 
        \left(\frac{E_1}{E_0}\right)^3-
        \pd{E_3}{E_2}\left(\frac{E_1}{E_0}\right)^2-
        \pd{E_3}{E_1}\frac{E_1}{E_0}
        -\pd{E_3}{E_0} 
        = 
        \left(\frac{E_1}{E_0}\right)^3- \frac{E_3}{E_0}-
        \left(\left(\frac{E_1}{E_0}\right)^2-\frac{E_2}{E_0}
        \right)\pd{E_3}{E_2}. 
    \]
    According to \cref{pro:M1P2}\ref{itm:E3E2}, i.e., $\pd{^2E_3}{E_2^2}=0$, we obtain
    \[ 
        \pd{p\left(E_1/E_0\right)}{E_2} = -\frac{1}{E_0}\pd{E_3}{E_2}
        + \frac{1}{E_0}\pd{E_3}{E_2} = 0,
    \]
    which indicates that $p\left(\dfrac{E_1}{E_0}\right)$ is independent of $E_2$.
    Hence, we can set $E_2$ as any available value, for instance $E_2=E_0\mE_2$,
    where $f_2=0$, and then $E_3=E_0\mE_3$.
    In this case, we have
    \begin{equation} \label{eq:pE1}
			p\left(\frac{E_1}{E_0}\right) = \mE_1^3-\mE_3-\left( \mE_1^2-\mE_2 \right)\pd{E_3}{E_2},
    \end{equation}
    whose sign is different from that of $E_1$, due to \cref{pro:M1P2}\ref{itm:E3geE1}.
    This completes the proof.
\end{proof}

\begin{proof}[Proof of \cref{thm:hyperbolic}]
    To prove the system \eqref{eq:M1P2_vec} is hyperbolic, that is, the matrix
    $\bJ$ is real diagonalizable, we need only to show that the characteristic
    polynomial $p\left(\lambda\right)$ has three distinct real zeros.

    If $E_1=0$, then $p(\lambda) = \lambda^3-\pd{E_3}{E_1}\lambda$.
    \Cref{pro:M1P2}\ref{itm:E3E1} shows $\pd{E_3}{E_1}\Big|_{E_1=0}>0$, so
    $p(\lambda)$ has three distinct zeros.

    If $E_1\neq0$, without loss of generality, we can assume that
    $E_1<0$. According to Lemma \ref{lem:plambda}, $p\left( \dfrac{E_1}{E_0}
    \right) > 0$. \Cref{pro:M1P2}\ref{itm:E3E0} shows that $p(0) = -\pd{E_3}{E_0} < 0$. Hence
    $p(\lambda)$ has three distinct zeros, which satisfy $\lambda_1 <
    \dfrac{E_1}{E_0} < \lambda_2 < 0 <\lambda_3$.
    This completes the proof.
\end{proof}


Denote the three distinct zeros of $p(x)$ by $\lambda_i(\bU)$, $i=1,2,3$
satisfying $\lambda_1 (\bU) < \lambda_2(\bU) < \lambda_3(\bU)$. 
One can directly deduce the following conclusion from 
\cref{thm:hyperbolic} and its proof.
\begin{deduction} \label{lem:lambda}
For any $\bU \in \Omega_R$, we have
\begin{equation} \label{eq:eigenvaluescondition}
  \lambda_1 (\bU) <\frac{E_1}{E_0}<\lambda_3 (\bU).
\end{equation}
\end{deduction}
%

\subsubsection{Riemann problem}
The characteristic structure of the moment model is fundamental for
further investigations into the behavior of the solution of the system.
Meanwhile, the solution structure of the Riemann problem is instructional for
studying the approximate Riemann solver, which is the basis of the numerical
methods using Godunov type schemes.
Here we investigate the characteristic structure of the \MPtwo model 
by the following Riemann problem:
\begin{equation}\label{eq:Riemann}
    \left\{ 
        \begin{aligned}
            &\frac{1}{c} \pd{\bU}{t} + \pd{\bF(\bU)}{z} = 0,\\
            &\bU(z,t=0) = \left\{ \begin{array}{ll}
                \bU^L,  &   \text{ if } z<0,\\
                \bU^R,  &   \text{ if } z>0.
            \end{array} \right.
        \end{aligned}
        \right.
\end{equation}

Recall that the Jacobian $\bJ$ has three distinct eigenvalues. For the eigenvalue
$\lambda_k$, $k=1,2,3$, the corresponding eigenvector is 
\begin{equation}\label{eq:eigenvectors}
  \br_{k}=(1,\lambda_k,\lambda_k^2)^T.
\end{equation}
We have the following conclusion on the wave type of each characteristic field.
\begin{theorem}\label{thm:wavetype}
    The $1$- and $3$-characteristic field are genuinely nonlinear, and the
    $2$-characteristic field is neither genuinely nonlinear nor linearly
    degenerate.
\end{theorem}
\begin{proof}
    Denote $\Delta_k := \nabla_{\bU}\lambda_k\cdot \br_{k}$, $k=1,2,3$. Noticing
    that $\lambda_k$ is a zero of $p(\lambda)=0$, we can obtain the following
    formula by the implicit differentiation
    \[
        \Delta_k = \frac{1}{p'(\lambda_k)}
        \sum_{i,j=0}^2\frac{\partial^2E_3}{\partial E_i\partial
        E_j}\lambda_k^{i+j}, \quad k=1,2,3.
    \]
    The fact that $p(\lambda)$ has three distinct zeros indicates
    $p'(\lambda_1)>0$, $p'(\lambda_2)<0$ and $p'(\lambda_3)>0$.
    Hence, we need only to study the sign of 
    \begin{equation}\label{eq:def_q}
        q(\lambda) = \sum_{i,j=0}^2\frac{\partial^2E_3}{\partial E_i
        \partial E_j}\lambda^{i+j},
    \end{equation}
    on $\lambda = \lambda_k$, $k=1,2,3$.
    \Cref{pro:M1P2}\ref{itm:E3E2} shows $\pd{^2E_3}{E_2^2}=0$, so $q(\lambda)$
    is a polynomial of degree $3$.

    For the $1$- and $3$-characteristic fields, we can directly solve the values
    of $\lambda_1$ and $\lambda_3$ because the degree of $p(\lambda)$ is $3$,
    and then substitute them into $q(\lambda)$ to check their sign for
    $\bU\in\Omega_R$.  However, the calculation is too complex. With the help of
    the cylindrical algebraic decomposition in computer algebraic
    \cite{collins1975quantifier} and its implementation in Maple\footnote{Maple
    is a trademark of Waterloo Maple Inc.},  we validate that $q(\lambda_1)<0$
    and $q(\lambda_3)>0$ for all $\bU\in\Omega_R$.

    For the $2$-characteristic field, 
    if we set $E_1=0$, then $q(\lambda_2) = \pd{^2E_3}{E_0^2}=0$.  If we set
    $E_1 / E_0=\pm 1/10$ and $E_2 / E_0= 1/3$, then $q(\lambda_2)\approx \pm
    0.0293 \gtrless0$.
    Thus the sign of $q(\lambda_2)$ varies over
    $\Omega_R$, which indicates that the $2$-characteristic field is neither
    genuinely nonlinear nor linearly degenerate.
\end{proof}

\begin{figure}[ht]
  \begin{center}
      \subfloat[$q(\lambda_1)$]{
          \includegraphics[width=0.33\textwidth,height=0.16\textheight]{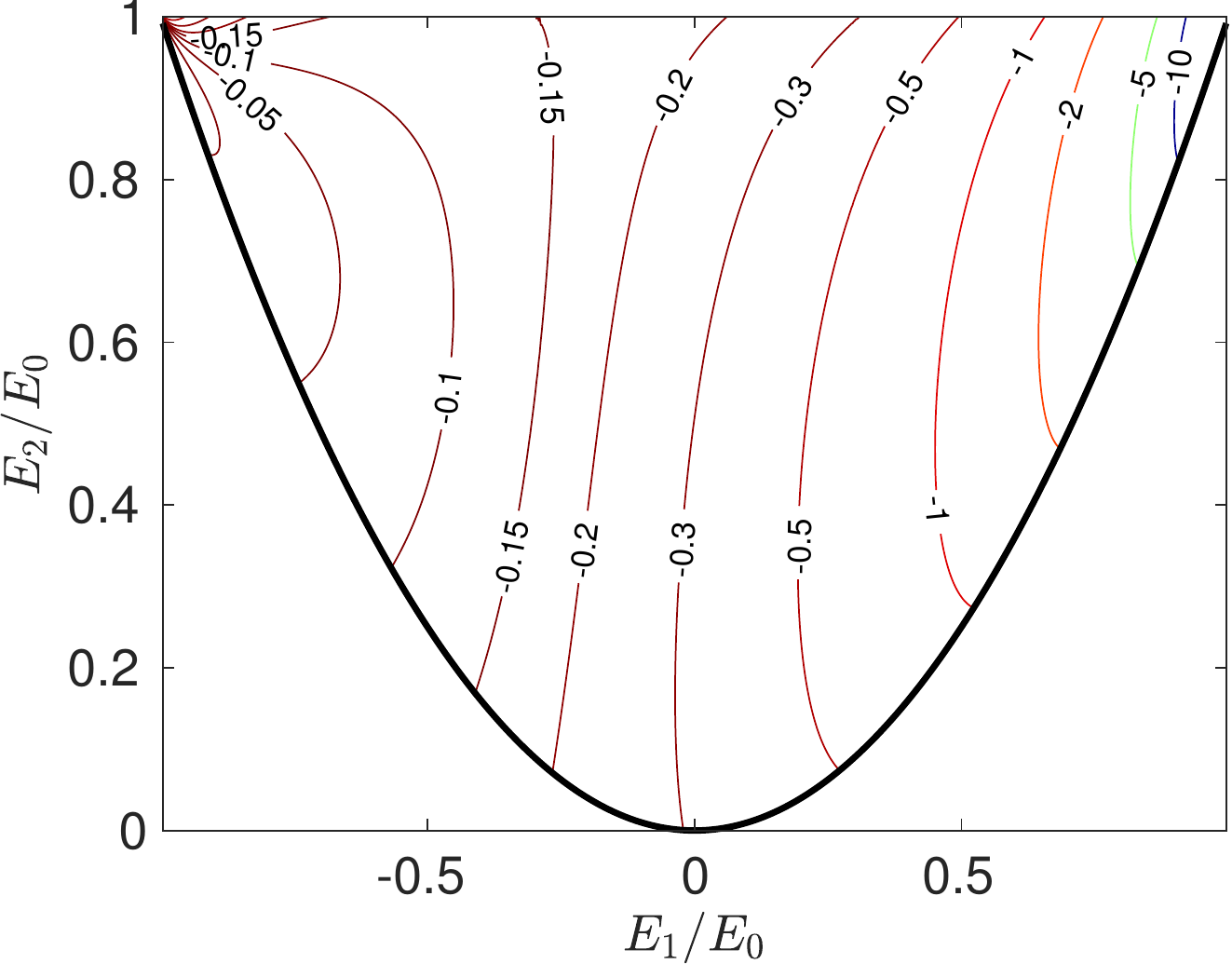}
      }
      \subfloat[$q(\lambda_2)$]{
          \includegraphics[width=0.33\textwidth,height=0.16\textheight]{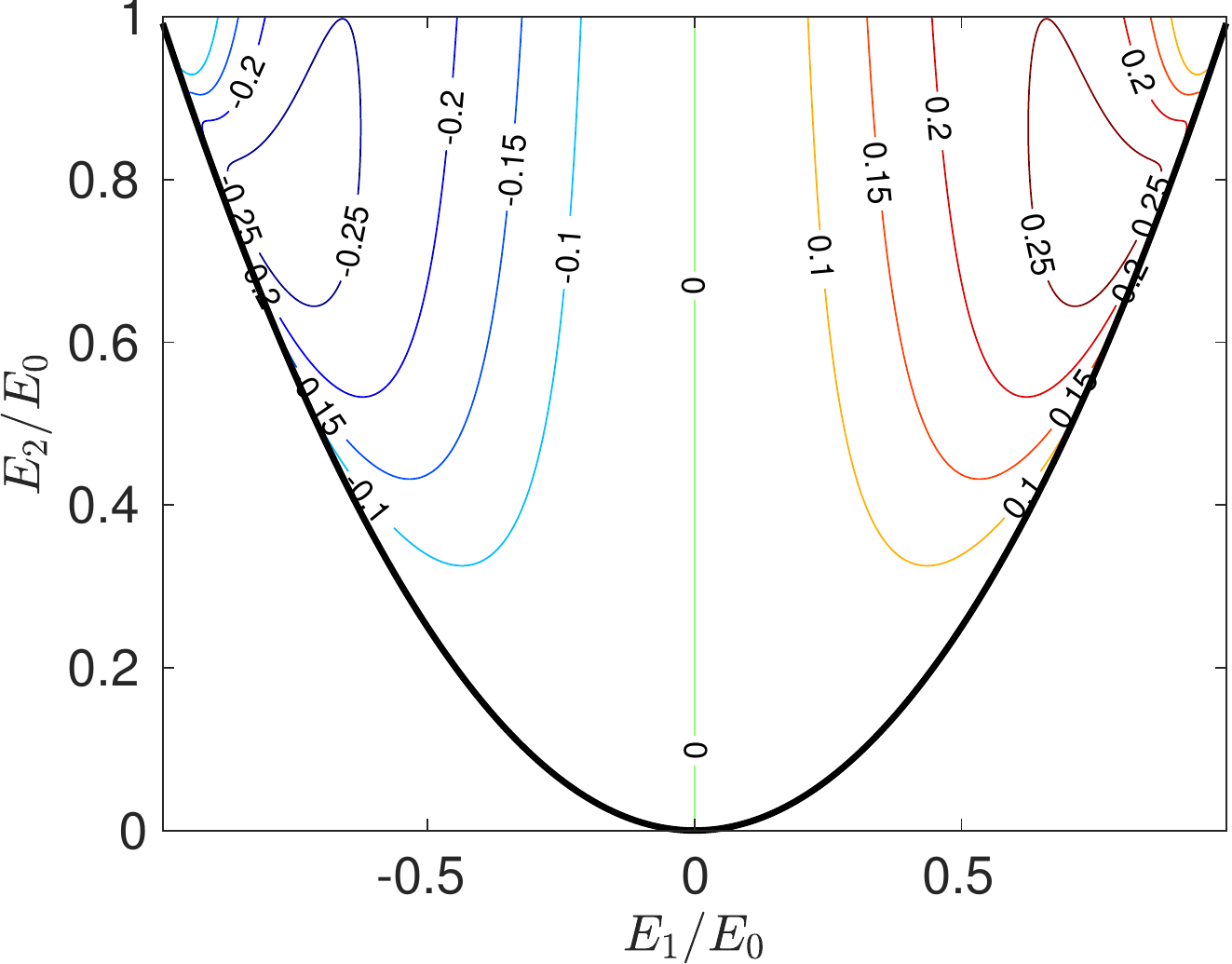}
      }
      \subfloat[$q(\lambda_3)$]{
          \includegraphics[width=0.33\textwidth,height=0.16\textheight]{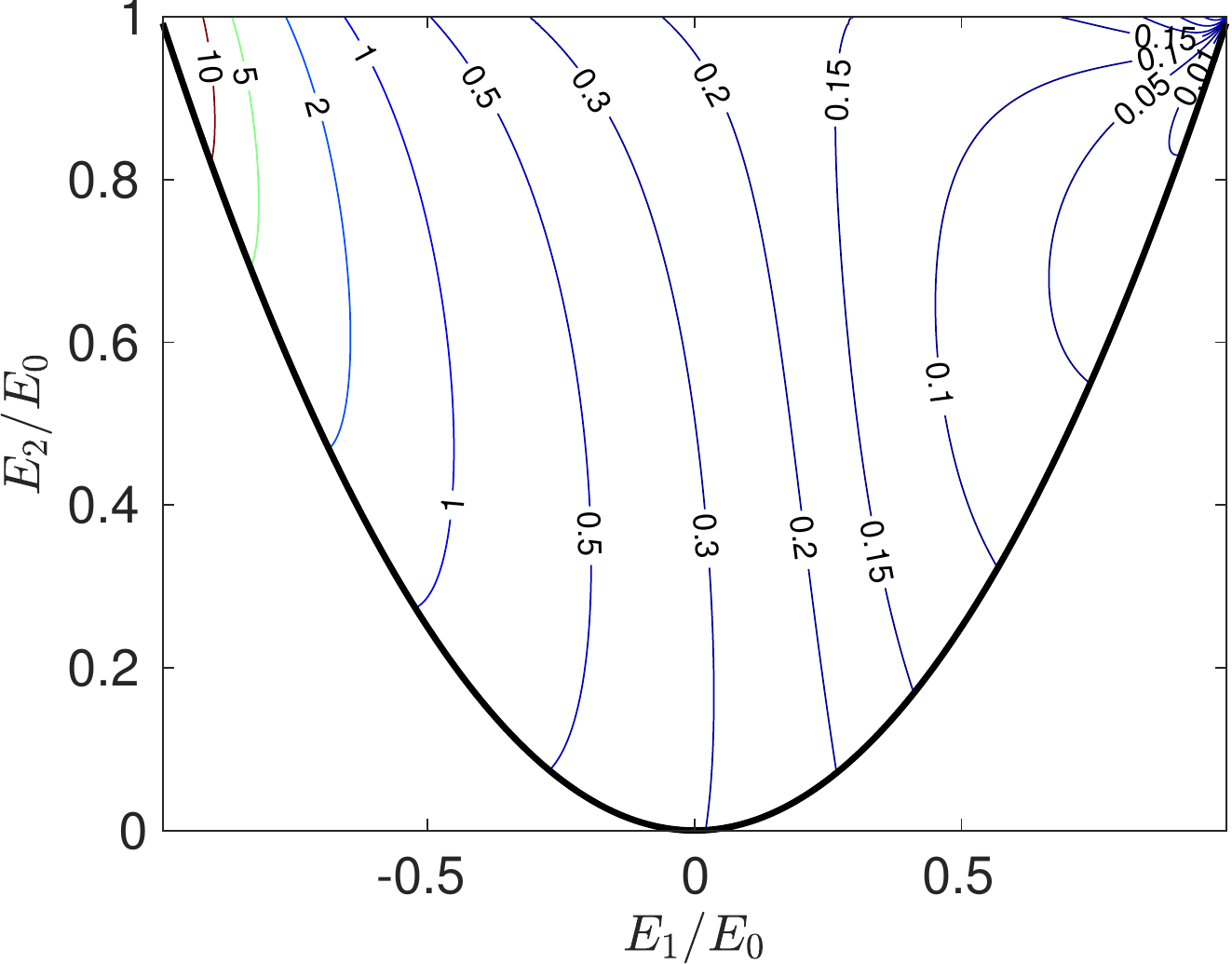}
      }
  \end{center}
  \caption{Contour of $q(\lambda)$ with $\lambda=\lambda_k$,
      $k=1,2,3$.\label{fig:delta13}.}
\end{figure}
To present a visualization on the sign of $q(\lambda)$ in \eqref{eq:def_q}, we
plot the contour of $q(\lambda_k)$, $k=1,2,3$ with $E_0=1$ as functions of $E_1$
and $E_2$ in \cref{fig:delta13}. The contour lines in the figures agree with the
conclusion in the \cref{thm:wavetype}.

The $1$- and $3$-characteristic fields are genuinely nonlinear, thus each field
associates with one wave, whose type is either rarefaction wave or shock.
However, for the $2$-characteristic fields, it corresponds to the nonconvex flux
\cite[Chapter 16.1]{leveque2002finite}, and one field may associate with more
than one wave. Investigation on nonconvex wave requires too many tools, so we
will not discuss the $2$-field too much. Below, we study the rarefaction waves
and shocks for the $1$- and $3$-fields.
%

\paragraph{Rarefaction waves}
For the system \eqref{eq:Riemann}, if two states $\bU^L$ and $\bU^R$ are
connected by a rarefaction wave in a genuinely nonlinear field, then 
the following two conditions must be satisfied:
\begin{enumerate}
    \item Constancy of the generalised Riemann invariants across the wave.
        That is to say, the integral curve $\tilde{\bU}(\zeta)=(\tilde{E}_0(\zeta),
        \tilde{E}_1(\zeta),\tilde{E}_2(\zeta))^T$ satisfies
        \begin{equation}\label{eq:integralcurve}
            \tilde{\bU}'(\zeta) = C(\zeta)\br_k(\tilde{\bU}(\zeta)),\quad k=1,3,
        \end{equation}
        where $C(\zeta)$ is a nonzero scalar factor with a fixed sign for any
        $\zeta$;
    \item Consistent condition:
        \begin{equation}\label{eq:consistent}
            \lambda_k(\bU^L)<\lambda_k(\bU^R).
        \end{equation}
\end{enumerate}

Since the characteristic speed is $\lambda_k=\lambda_k(\tilde{\bU})$, along the
integral curve, we have
\[
\dfrac{1}{C(\zeta)}\od{\lambda_k}{\zeta} = 
\nabla_{\tilde{\bU}}\lambda_k\cdot\br_k
  \left\{
    \begin{array}{ll}
      <0, &   k=1,\\
      >0, &   k=3,
    \end{array}\right.
\]
according to the proof of the \cref{thm:wavetype}.
Noticing $\bU^L$ and $\bU^R$ lie in an integral curve, we let $\bU^L =
\tilde{\bU}(0)$, then these exists a $\zeta_*$ such that $\bU^R =
\tilde{\bU}(\zeta_*)$.
According to \eqref{eq:consistent}, we have 
\begin{equation}\label{eq:signOfZeta}
  C(\zeta_*) \zeta_* \left\{
    \begin{array}{ll}
      <0, &   k=1,\\
      >0, &   k=3.
    \end{array} \right.
\end{equation}

For the $3$-rd characteristic field,
\eqref{eq:integralcurve} and \eqref{eq:signOfZeta} indicates
\(
\od{\tilde{E}_0(\zeta)}{\zeta} = C(\zeta), 
\)
and thus $$E_0^L < E_0^R.$$
Let $\tilde{u}(\zeta) = \dfrac{\tilde{E}_1(\zeta)}{\tilde{E}_0(\zeta)}$
and
$\tilde{p}(\zeta) = \tilde{E}_2(\zeta) -
\dfrac{\tilde{E}_1^2(\zeta)}{\tilde{E}_0(\zeta)}$.
By \eqref{eq:integralcurve}, we obtain
\begin{equation*}
    \begin{aligned}
\od{\tilde{u}(\zeta)}{\zeta} &=
\dfrac{\tilde{E}_0(\zeta)\od{\tilde{E}_1(\zeta)}{\zeta}
-\od{\tilde{E}_0(\zeta)}{\zeta}\tilde{E}_1(\zeta)
}{\tilde{E}^2_0(\zeta)} 
= C(\zeta)\left(\lambda_3 - \tilde{u}
\right),\\
\od{\tilde{p}(\zeta)}{\zeta}&=C(\zeta)(\lambda_3-\tilde{u})^2.
    \end{aligned}
\end{equation*}
According to \eqref{eq:eigenvaluescondition}, $\lambda_3-\tilde{u}>0$.
Thus we can obtain 
\[
  u^L < u^R, \quad p^L<p^R. 
\]

Analogously, for the $1$-st characteristic wave, we have
\[
  E_0^L > E_0^R,\quad u^L < u^R, \quad p^L > p^R.
\]

\paragraph{Shock waves}
If two states $\bU^L$ and $\bU^R$ are connected by a shock in a genuinely
nonlinear field, then we have the following two relationships:
\begin{enumerate}
    \item Rankine-Hugoniot condition:
        \begin{equation} \label{eq_shock_RH}
            \begin{aligned}
                E_1^R-E_1^L &= s_k (E_0^R-E_0^L),\\
                E_2^R-E_2^L &= s_k (E_1^R-E_1^L),\\
                E_3^R-E_3^L &= s_k (E_2^R-E_2^L), ~k=1,3,
            \end{aligned}
        \end{equation}
        where $s_k$ is the speed of the shock;
    \item Entropy condition:
        \begin{equation}\label{eq:entropycondition}
            \lambda_k^L>s_k>\lambda_k^R,\quad k=1,3.
        \end{equation}
\end{enumerate}
The first two equations in \eqref{eq_shock_RH} indicate
\begin{equation}\label{eq:shock_sk}
  s_k = \frac{E_1^R-E_1^L}{E_0^R-E_0^L} = \frac{E_2^R-E_2^L}{E_1^R-E_1^L}.
\end{equation}
Let $u = \dfrac{E_1}{E_0}$ and $p = E_2 - \dfrac{E_1^2}{E_0}$.
For the $3$-rd characteristic field, noticing
$s_3 > \lambda_3^R > u^R$, we can obtain
\begin{equation}\label{eq:shock_u}
  (E_0^R-E_0^L)(u^R-u^L)  
  = \frac{(E_0^R-E_0^L)^2}{E_0^L} (s_3-u^R)
  > 0,
\end{equation}
and
\begin{equation}\label{eq:shock_p}
  \begin{aligned}
    (E_0^R-E_0^L) (p^R-p^L) &= (E_1^R-E_1^L)^2-(E_0^R-E_0^L)
    \left(\frac{(E_1^R)^2}{E_0^R}
      -\frac{(E_1^L)^2}{E_0^L}\right) \\
      &= \frac{1}{E_0^LE_0^R}\left( E_0^R E_1^L -E_0^L E_1^R \right)^2 > 0.
  \end{aligned}
\end{equation}
The remaining work is to study the sign of $E_0^R-E_0^L$ for the
$3$-characteristic field. 

Denote the Hugoniot curves by $\tilde{\bU}(\tau)$, with $\tilde{\bU}(0)=\bU^L$
and $\tilde{\bU}(1)=\bU^R$. For a given $\tau_*\in[0,1)$ and a sufficient small
$\epsilon$, let $\bU^l=\tilde{\bU}(\tau_*)$ and 
$\bU^r=\tilde{\bU}(\tau_*+\epsilon)$, then $\bU^l$ and $\bU^r$ also satisfy
the Rankine-Hugoniot condition \eqref{eq_shock_RH} and the entropy condition
\eqref{eq:entropycondition}.

Let $d=E_0^r-E_0^l$, then $d\ll 1$ and
\[
   E_1^r=E_1^l + s_3 d,\quad E_2^r=E_2^l + s_3^2d.
\]
According to the Taylor expansion, we have
\begin{align} 
  \lambda_3^r-\lambda^l_3=d\left(
  \left(\pd{\lambda_3}{E_0}\right)^l+\left(\pd{\lambda_3}
  {E_1}\right)^ls_3
  +\left(\pd{\lambda_3}{E_2}\right)^ls_3^2
  \right)  + O(d^2).
\end{align}
Thus, $|\lambda_3^l-\lambda_3^r|\ll 1$, and $|s_3-\lambda_3^l|\ll 1$.
With 
\begin{align} 
\left(\pd{\lambda_3}{E_0}\right)^l+\left(\pd{\lambda_3}
{E_1}\right)^l\lambda_3^l
+\left(\pd{\lambda_3}{E_2}\right)^l(\lambda^l_3)^2  = (\nabla_{\bU}
\lambda_3 \cdot \br_3)|_{\bU=\bU^l} = \Delta_3 |_{\bU=\bU^l}
> 0,
\end{align}
and the entropy condition, we obtain that
\[  
  d < 0,\quad E_0^r<E_0^l.
\]
Because the upper relation holds for any $\tau\in[0,1)$, we have that 
for the $3$-characteristic field
\begin{equation}
    E_0^L>E_0^R,\quad u^L > u^R,\quad p^L > p^R.
\end{equation}
Analogously, for the $1$-characteristic field, we have
\begin{equation}
    E_0^L<E_0^R,\quad u^L > u^R,\quad p^L < p^R.
\end{equation}

We summarize all the conclusions on genuinely nonlinear waves in the following
theorem to close this section.
\begin{theorem}
    For the $1$- and $3$-characteristic fields, the variables
    $E_0$, $u=\frac{E_1}{E_0}$ and $p=E_2-\frac{E_1^2}{E_0}$ on both sides of the wave 
    have the relationship with the wave type as in the \cref{tab:wavetype}.
  \begin{table}[ht]
    \centering
    \begin{tabular}{|l|c|c|c|}
      \hline 
      Wave type & $E_0$ & $u$ & $p$ \\
      \hline 
      \multirow{2}{*}{Rarefaction wave}
      & $\frac{\mathstrut}{\mathstrut}$1-wave $E_0^L>E_0^R$
              & \multirow{2}{*}{$u^L < u^R$}
                    & 1-wave $p^L>p^R$ \\  
      \cline{2-2}
      \cline{4-4}
      & $\frac{\mathstrut}{\mathstrut}$3-wave $E_0^L<E_0^R$
              &     & 3-wave $p^L<p^R$ \\
      \hline \hline   
      \multirow{2}{*}{Shock wave}
      & $\frac{\mathstrut}{\mathstrut}$1-wave $E_0^L<E_0^R$
              & \multirow{2}{*}{$u^L > u^R$}
                    & 1-wave $p^L<p^R$ \\  
      \cline{2-2}
      \cline{4-4}
      & $\frac{\mathstrut}{\mathstrut}$3-wave $E_0^L>E_0^R$
              & & 3-wave $p^L>p^R$ \\
      \hline 
    \end{tabular}
    \caption{\label{tab:wavetype}The relationship between the wave type and the
variables $E_0$, $u$, and $p$.}
  \end{table}
\end{theorem}


%% file: numerical.tex
\section{Numerical simulation} \label{sec:numericalresults}
In this section, we discuss the numerical scheme for the \MPN model, and perform
numerical simulations on some typical examples to demonstrate its numerical
efficiency.

\subsection{Numerical scheme}
Because the convection part of the \MPN model is hyperbolic conservation laws in
the sense of balance laws, we discretize it by the finite volume method. The
source term and the governing equation of internal energy
\eqref{eq:internalenergy} both contain the term $T^4$, which is usually a stiff
term. Hence, an implicit scheme is adopted to deal with the stiff terms.

Precisely, we assume the spatial domain is $[z_l, z_r]$, and the number of
discretization cell is $N_{\cell}$. A uniform discretization yields
the spatial step $\Delta z = \frac{z_r-z_l}{N_{\cell}}$, discretization points $z_i = z_l +
(i-1/2)\Delta z$, $i=1,\cdots,N$, and mesh cells $[z_{i-1/2},z_{i+1/2}]$,
$i=1,\cdots,N$ with $z_{i-1/2}=z_i-\Delta z/2$.

Denote the approximation of the solution in $i$-cell at time step $t_n$ by
$\bU_i^n$, and analogous for the source $\bS$ and the internal energy $e$.
The numerical scheme for the \MPN system is
%
%
\begin{equation} \label{eq:discretization}
\begin{aligned}
    \dfrac{\bU_i^{n+1}-\bU_i^n}{c\Delta t}
    &+ \frac{\bF_{i+1/2}-\bF_{i-1/2}}{\Delta z} = \bS_i^{n+1},\\
    \dfrac{e_i^{n+1}-e_i^{n}}{\Delta t} &=
    \sigma_{a,i}^n\left(E_{0,i}^{n+1}-a c(T_i^{n+1})^4\right),
\end{aligned}
\end{equation}
where the $k$-th element of the source term $\bS_i^{n+1}$ has the form
\begin{equation}
    S_{k,i}^{n+1} = 
    -\sigma^n_{t,i} E_{k,i}^{n+1} 
    + \frac{1-(-1)^{k+1}}{2k+2} \left( ac \sigma_{a,i}^n (T_i^{n+1})^4
    +\sigma_{s,i}^n E_{0,i}^{n+1} + s^n_i \right).
\end{equation}
Here we adopt the Lax-Friedrich scheme in the numerical flux $\bF_{i+1/2}$.
Due to the fact that the source term and the governing equation of 
the internal energy is
implicitly discretized, the constraints of the time step is 
the CFL condition
\begin{equation}
    \Delta t = \mathrm{CFL} \cdot \min_{i}\frac{\Delta
    z}{\max_{k}|\lambda_k(\bU_i^n)|}.
\end{equation}
We set the CFL number to be $0.3$ in all the numerical simulations.
It is worth to point out that in the absence of any external source of
radiation, i.e. $s=0$, adding the first equation of the discretization of
moments in \eqref{eq:discretization} and the discretization of the internal
energy in \eqref{eq:discretization} yields
\begin{equation}
    \dfrac{e_i^{n+1}-e_i^n}{\Delta t}+
    \dfrac{E_{0,i}^{n+1}-E_{0,i}^n}{c\Delta t}+\dfrac{F_{0,i+1/2}-F_{0,i-1/2}}{\Delta
    z} = 0,
\end{equation}
which is the discretization version of the conservation of the total energy
\eqref{eq:conservation}.

\subsubsection{Boundary condition} 
The ansatz of the \MPN provides an injective function between the concerned
moments in the \MPN model and the intensity in the form \eqref{eq:ansatz}.
This allows us to construct the boundary condition of the \MPN model based on
the boundary condition of the RTE. Without loss of generality, we take the left
boundary as an example.
For the RTE, the intensity on the boundary is 
\begin{equation}
    I^B(t,\mu) = 
    \begin{cases}
        I(z=z_l,t,\mu),  &   \mu < 0,\\
        I_{\out}(t,\mu),  &  \mu > 0,
    \end{cases}
\end{equation}
where $I_{\out}$, which is the intensity outside of the domain, is problem
dependent.
For instance, 
the intensity outside the domain for the common used reflective boundary
condition is
\begin{equation}
    I_{\out}(t,\mu)=I(z=z_l, t,-\mu), \quad \mu>0,
\end{equation}
and for the vacuum boundary condition, we have
\begin{equation}
    I_{\out}(t,\mu) = 0, \quad \mu>0.
\end{equation}

Due to the ansatz \eqref{eq:ansatz}, for the boundary condition of the \MPN
model, the intensity close to the domain $I(z=z_l,t,\mu)$ is replaced by
the ansatz constructed by the moments close to the domain $\hat{I}(\mu;
\bU(z=z_l,t))$. Then, one can directly evaluate the moments on the boundary.
Particularly, the flux across the boundary for the $k$-th moment is
\begin{equation} \label{eq:boundaryflux}
    F^B_k=\int_{-1}^0 \mu^{k+1} \hat{I}(\mu;\bU(z=z_l,t))\dd\mu +
    \int_{0}^1 \mu^{k+1}I_{\out}(t,\mu)\dd\mu.
\end{equation}



\subsubsection{Implementation}
The implementation of the numerical scheme \eqref{eq:discretization} is
straightforward except the moment closure, where one has to evaluate $\mE_i$,
$i=0,\cdots,2N+1$ fast and accurately. However, a naive implementation on
$\mE_i$ will lose the accuracy when $\alpha$ varies in $(-1,1)$. In the
following, we discuss the details of the implementation.

Note that $\mE_k=\moment{\weight(\mu)}_k$. If $|\alpha|$ is close to $0$, using
Taylor expansion on the weight function $\weight(\mu)$, one can obtain
\begin{equation*}
    \begin{aligned}
        \mE_k &=
        \int_{-1}^1\varepsilon\sum_{m=0}^\infty\frac{(m+1)(m+2)(m+3)}{6}(-\alpha)^m\mu^{m+k}\dd\mu \\
        &=\sum_{m+k \text{ is even}}\varepsilon\frac{(m+1)(m+2)(m+3)}{3}\frac{(-\alpha)^m}{m+k+1}.
    \end{aligned}
\end{equation*}
Hence, we have
\begin{equation}\label{eq:Ek_taylor}
    \mE_k = \left\{ \begin{array}{ll}
    \sum\limits_{j=0}^\infty\frac{2(2j+1)(j+1)(2j+3)\alpha^{2j}}{3}\frac{\varepsilon}{2j+k+1}, & k \text{ is even},\\
    -\alpha\sum\limits_{j=0}^{\infty}\frac{4(j+1)(2j+3)(j+2)\alpha^{2j}}{3}\frac{\varepsilon}{2j+k+2}, & k \text{ is odd}.
    \end{array} \right.
\end{equation}
We truncate the summation of series as $j=0,\dots,n_s$, then its error  is
$O((n_s+3)^3\alpha^{2n_s+2})$.
Clearly, the formula \eqref{eq:Ek_taylor} is efficient for the evaluation of
$\mE_k$ if $|\alpha|$ is small.

On the other hand, if $|\alpha|$ is not small, we let
\begin{equation*}
    J_k^{(n)}=\int_{-1}^1\frac{\varepsilon}{(1+\alpha\mu)^n}\mu^k\dd\mu,
\end{equation*}
then $\mE_k=J_k^{(4)}$.
Direct calculation using the integral by part yields
\begin{equation}\label{eq:rec_J}
    J_k^{(n)} = \frac{k}{(n-1)\alpha}J_{k-1}^{(n-1)} + \frac{\varepsilon}{(n-1)\alpha}
    \left( \frac{(-1)^{k}}{(1-\alpha)^{n-1}}-\frac{1}{(1+\alpha)^{n-1}} \right).
\end{equation}
For the case $n=1$, using the following recursive relationship
\[ 
   J_k^{(1)} = 
   \begin{cases}
     -\frac{1}{\alpha}J_{k-1}^{(1)},& k \text{ is even},\\
     -\frac{1}{\alpha}J_{k-1}^{(1)} + \frac{\varepsilon}{\alpha k},& k \text{ is odd},
   \end{cases}
\]
 we obtain
\begin{equation}\label{eq:Jk1}
    J_k^{(1)} = \varepsilon\begin{cases}
        \frac{\ln(1+\alpha)-\ln(1-\alpha)}{\alpha^{k+1}} 
        -\sum_{j=1}^{\lfloor k/2\rfloor}\frac{2}{(k+1-2j)\alpha^{2j}}, & 
        k \text{ is even},\\
        \frac{\ln(1-\alpha)-\ln(1+\alpha)}{\alpha^{k+1}} 
        +\sum_{j=0}^{\lfloor k/2 \rfloor}\frac{2}{(k-2j)\alpha^{1+2j}}, &
        k \text{ is odd}.
    \end{cases}
\end{equation}
The equations \eqref{eq:rec_J} and \eqref{eq:Jk1} provide a formula to evaluate
$\mE_k$ for the case $|\alpha|$ not small.

In our implementation, the algorithm for $\mE_k$ is a combination of
\eqref{eq:Ek_taylor} and \eqref{eq:rec_J}.

\subsection{Bilateral inflow}
This example is used to study the behavior of the solution of the \MPN model,
hence the right hand side of \eqref{eq:radiativetransfer} vanishes, i.e., the RTE
degenerates into
\begin{align}\label{eq:equationinvacuum}
  \frac{1}{c}\pd{I}{t}+\mu\pd{I}{z} = 0.
\end{align}
Using the method of characteristics, we obtain the analytical solution  
\begin{align}\label{eq:solutioninvacuum}
  I(z,\mu,t) = I_0(z-c\mu t,\mu),
\end{align}
where $I_0$ is the initial value at $t=0$. Here we choose the initial value as
\begin{equation}
    I_0(z,\mu) = 
    \begin{cases}
        ac\delta(\mu-1), & z\leq 0.2, \\
        0,& 0.2<z\leq 0.8,\\
        \frac{1}{2}ac,& z>0.8,
    \end{cases}
\end{equation}
which is consist of two Riemann problems.
The initial intensity on the left is a Dirac delta function, which is an 
extremely anisotropic distribution. Generally, it is challenging to approximate
such a function for the method based on the polynomial expansion, including the
\PN and \MPN models. 
The initial intensity on the right is an equilibrium, however, because the
initial intensity in the middle is zero, the intensity for the right
Riemann problem is not continuous. 
In the following, we perform simulations to study the efficiency of the \MPN
model on this bilateral flow.


\begin{figure}[ht]
  \centering
  \subfloat[$E_0$ of \MPtwo and $P_2$]{
	\label{fig:delta2_E0}
  \includegraphics[width=0.33\textwidth,height=0.16\textheight]{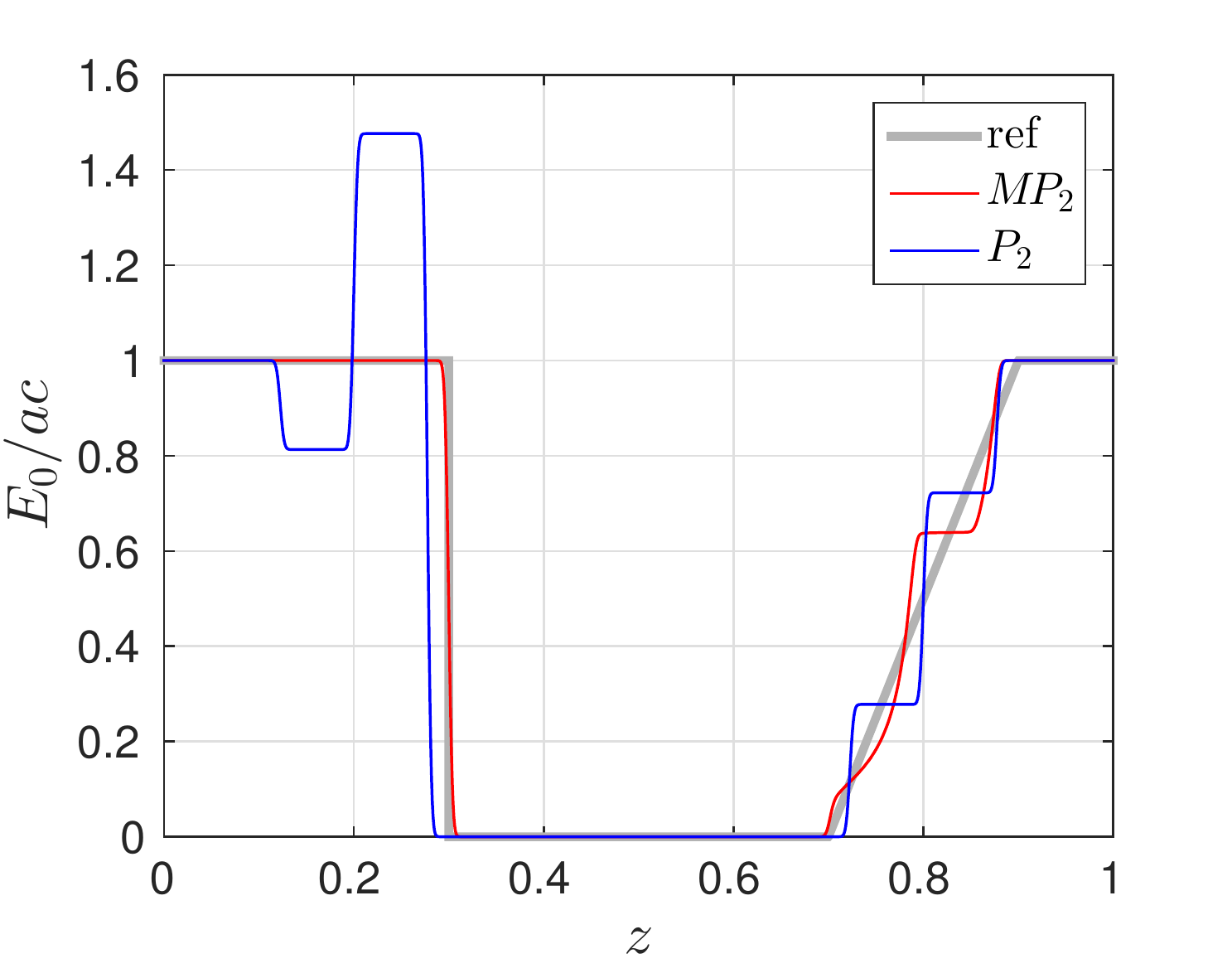}
  }
  \subfloat[$E_0$ of \MP{6} and $P_6$]{
	\label{fig:delta2_E06}
  \includegraphics[width=0.33\textwidth,height=0.16\textheight]{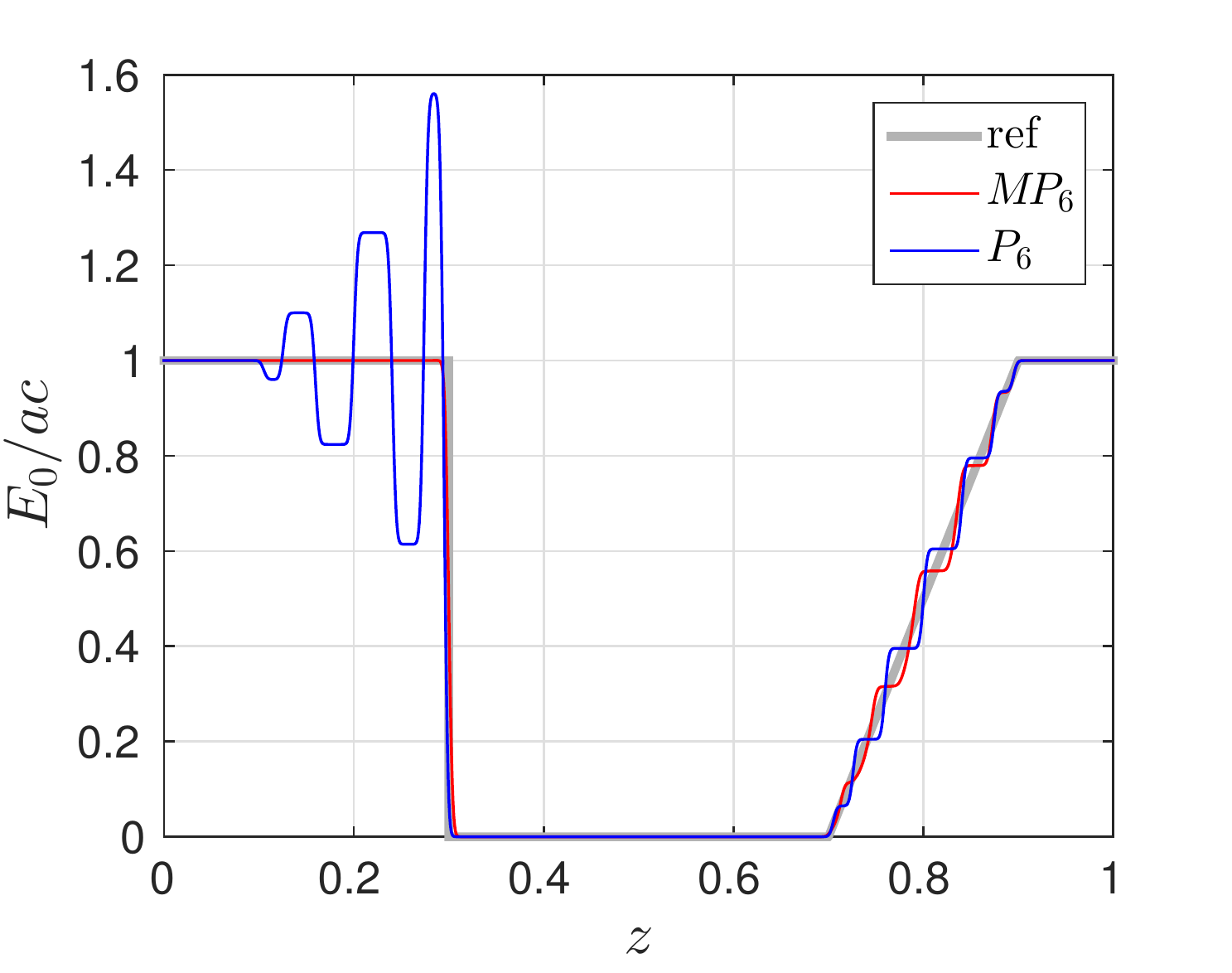}
  }
  \subfloat[$E_0$ of \MP{10} and $P_{10}$]{
  \includegraphics[width=0.33\textwidth,height=0.16\textheight]{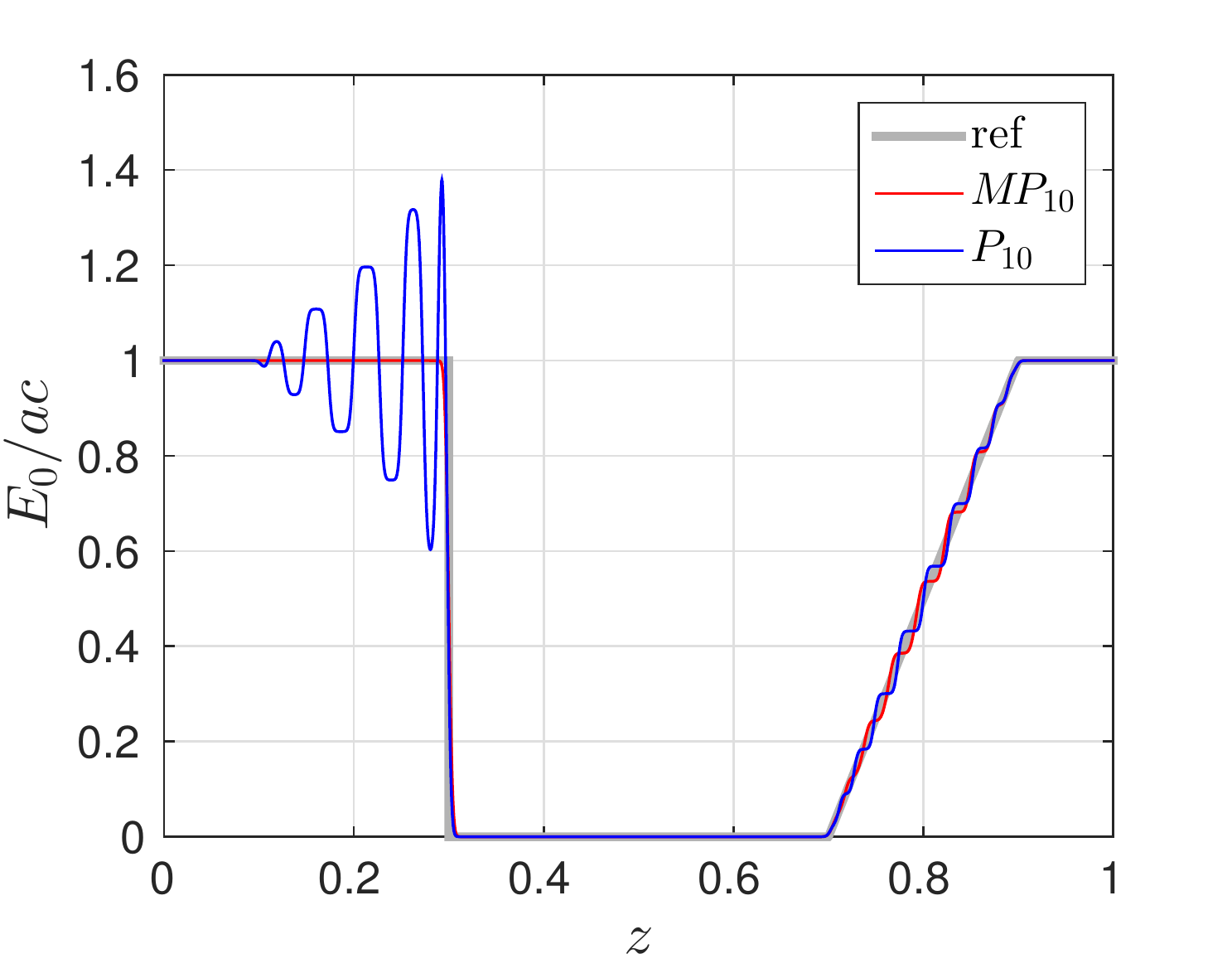}
  }\\
  \subfloat[${E_1}$ of \MPtwo and $P_2$]{
	\label{fig:delta2_E1}
  \includegraphics[width=0.33\textwidth,height=0.16\textheight]{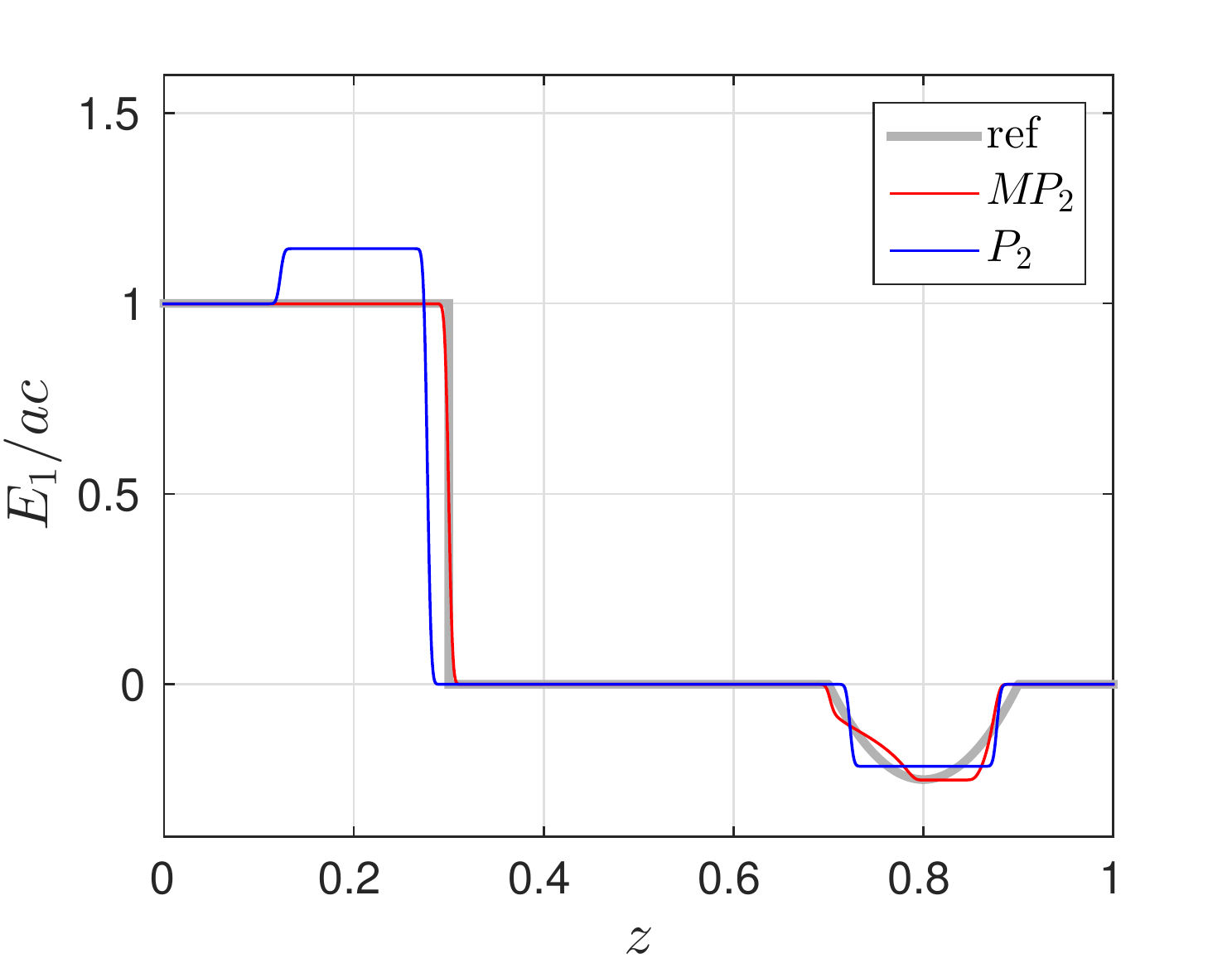}
  }
  \subfloat[${E_1}$ of \MP{6} and $P_6$]{
	\label{fig:delta2_E16}
  \includegraphics[width=0.33\textwidth,height=0.16\textheight]{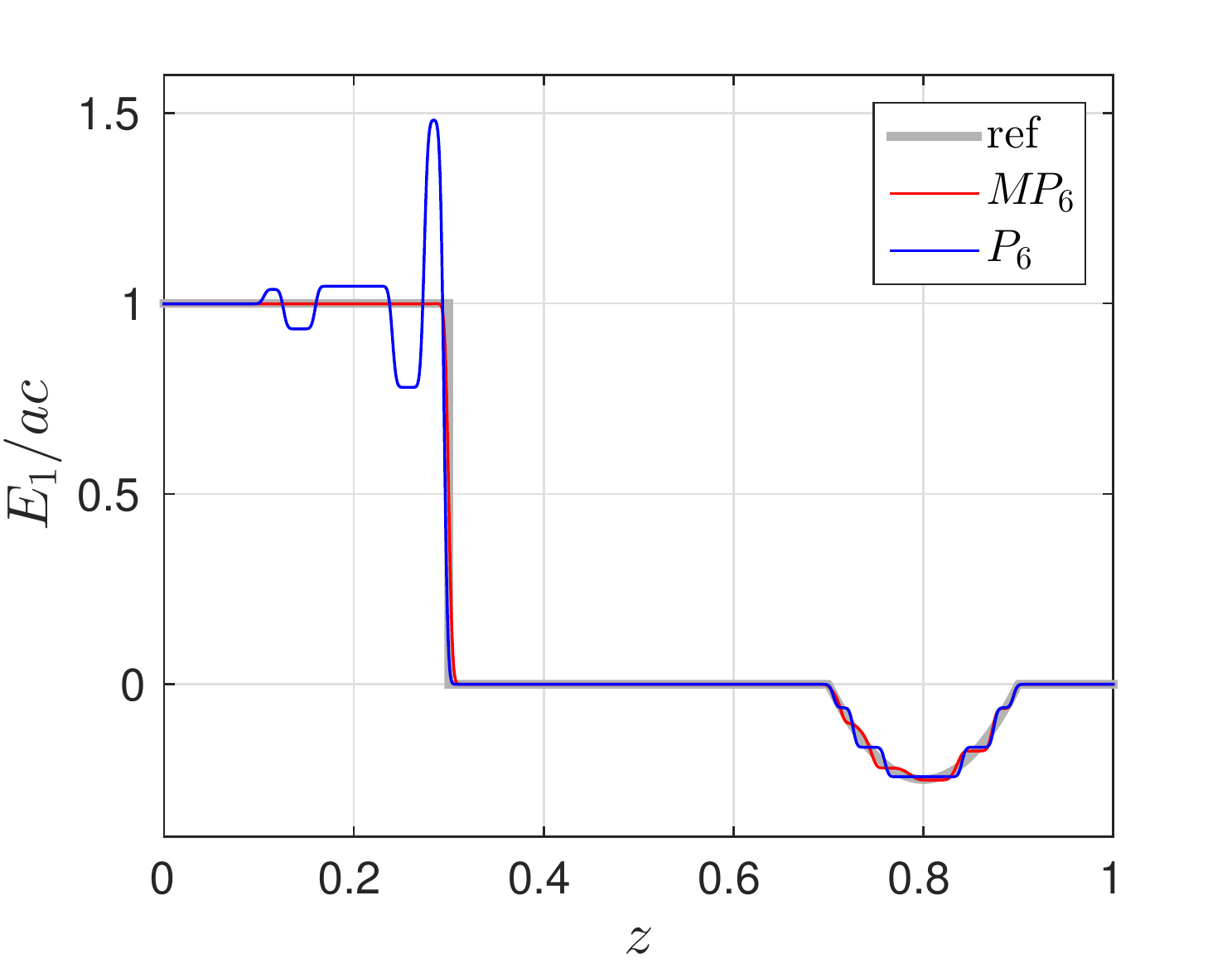}
  }
  \subfloat[${E_1}$ of \MP{10} and $P_{10}$]{
  \includegraphics[width=0.33\textwidth,height=0.16\textheight]{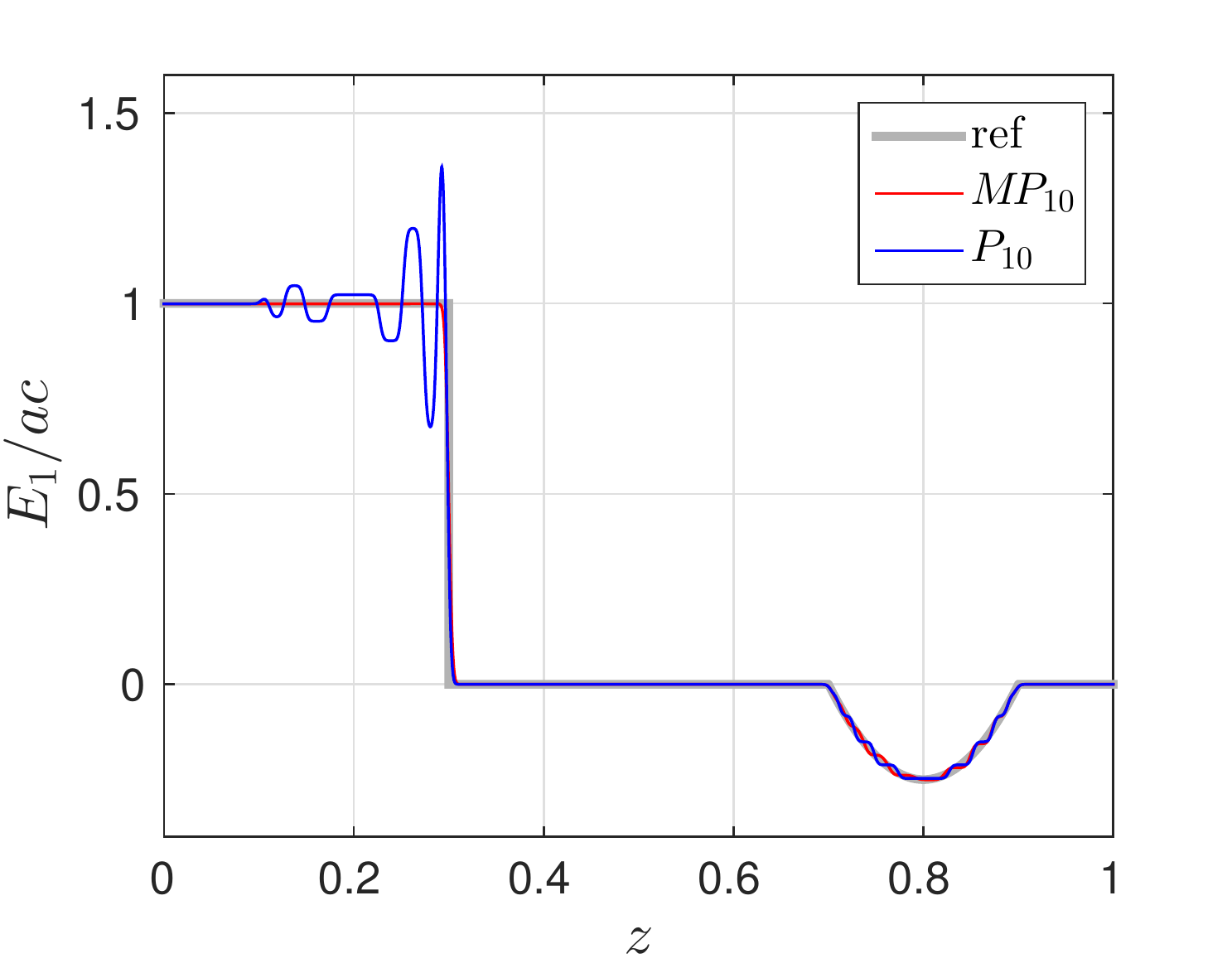}
  }
  \caption{Profile of $E_0$ and $E_1$ for the \MPN model and the \PN
  model for the bilateral inflow.}\label{fig:delta_E0}
\end{figure}


We simulate the problem by the \MPN model and the \PN model till $\ctend =0.1$.
Because the speed of light is finite, we can limit the computational domain in
$[0,1]$, which is uniformly discretized with the number of cells to be 
$N_{\cell}=100000$.
\cref{fig:delta_E0} presents the profile of $E_0$ and $E_1$ for the \MPN and \PN
models with $N=2,6,10$ and the reference solution. 

Clearly, from the left part ($z<0.5$) of each figure, there is an oscillation in
the results of both $E_0$ and $E_1$ of the \PN model, 
and as $N$ increases, the oscillation frequency increases. This Gibbs phenomenon
is caused by the failure of the approximation to the Dirac delta function by
Legendre series.
On the other hand, the results of the \MPN model agree with the reference
solution well, even when $N=2$, which indicates the \MPN has the ability to
simulate the strongly anisotropic problem. This supports the argument
at the end of \cref{sec:MPNmodel} that the ansatz of the \MPN model has the
ability to approximate anisotropic distributions because of its specific weight
function.
Moreover, according to \cref{fig:delta2_E0} and \cref{fig:delta2_E1}, one can
observe that the fastest wave of the $P_2$ model spread much slower than the
reference solution, because the characteristic speed of the $P_2$ model is less
than $1$. The \MPtwo does not have this issue due to its specific weight
function, which allows the characteristic speed reaches $1$.

From the right part ($z>0.5$) of each figure, both the \MPN and \PN models have
a good agreement with the reference solution, and it turns better as $N$
increases. This indicates the \MPN model also has the ability to simulate the
problem with the discontinuous intensity. Compared with the \PN model, one
can observe that the results of the \MPN model are a bit closer to the reference
solution. 


\subsection{Gaussian source problem} \label{sec:gaussian}
This example simulates particles with an initial specific intensity that is a
Gaussian distribution in space \cite{Frank2012Perturbed}:
\begin{align}
  I_0(z,\mu) =
  \frac{ac}{\sqrt{2\pi\theta}}e^{-\frac{z^2}{2\theta}},
  \quad \theta = \frac{1}{100},\quad  z\in(-L,L).
\end{align}
Here we set a large enough $L=\ctend+1$ to ensure that the energy reaching the
boundaries is negligible, and vacuum boundary conditions are prescribed at both
boundaries.
The medium is purely scattering with $\sigma_s=\sigma_t=1$, thus 
the material coupling term vanishes. We also set the external source to be zero.

\begin{figure}[htb]
  \centering
  \subfloat[$E_0$ of \MP{3} and $P_3$]{
  \includegraphics[trim={0mm 0mm 0mm 0mm},clip,width=0.33\textwidth,height=0.16\textheight]{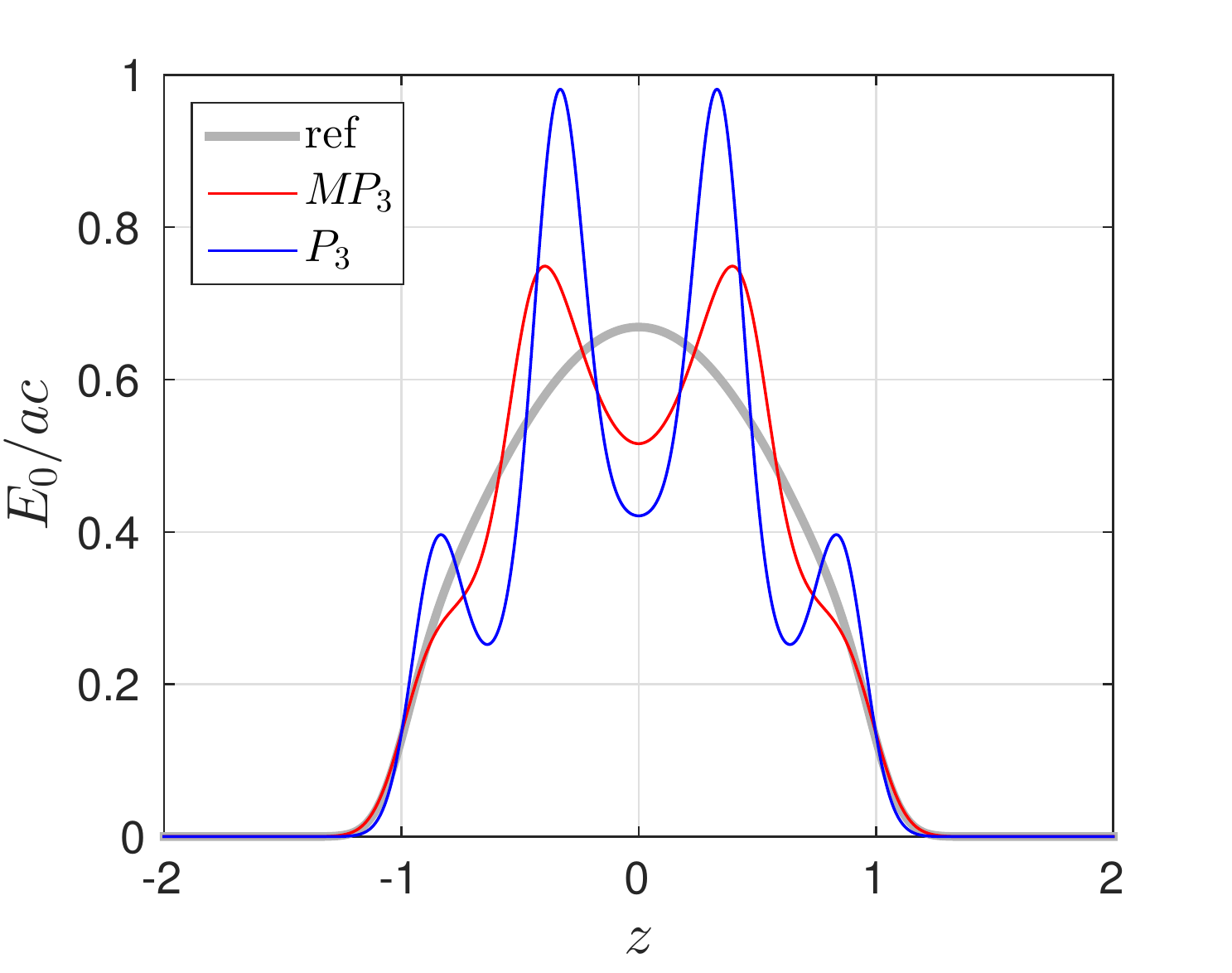}
  }
  \subfloat[${E_1}$ of \MP{3} and $P_3$]{
  \includegraphics[trim={0mm 0mm 0mm 0mm},clip,width=0.33\textwidth,height=0.16\textheight]{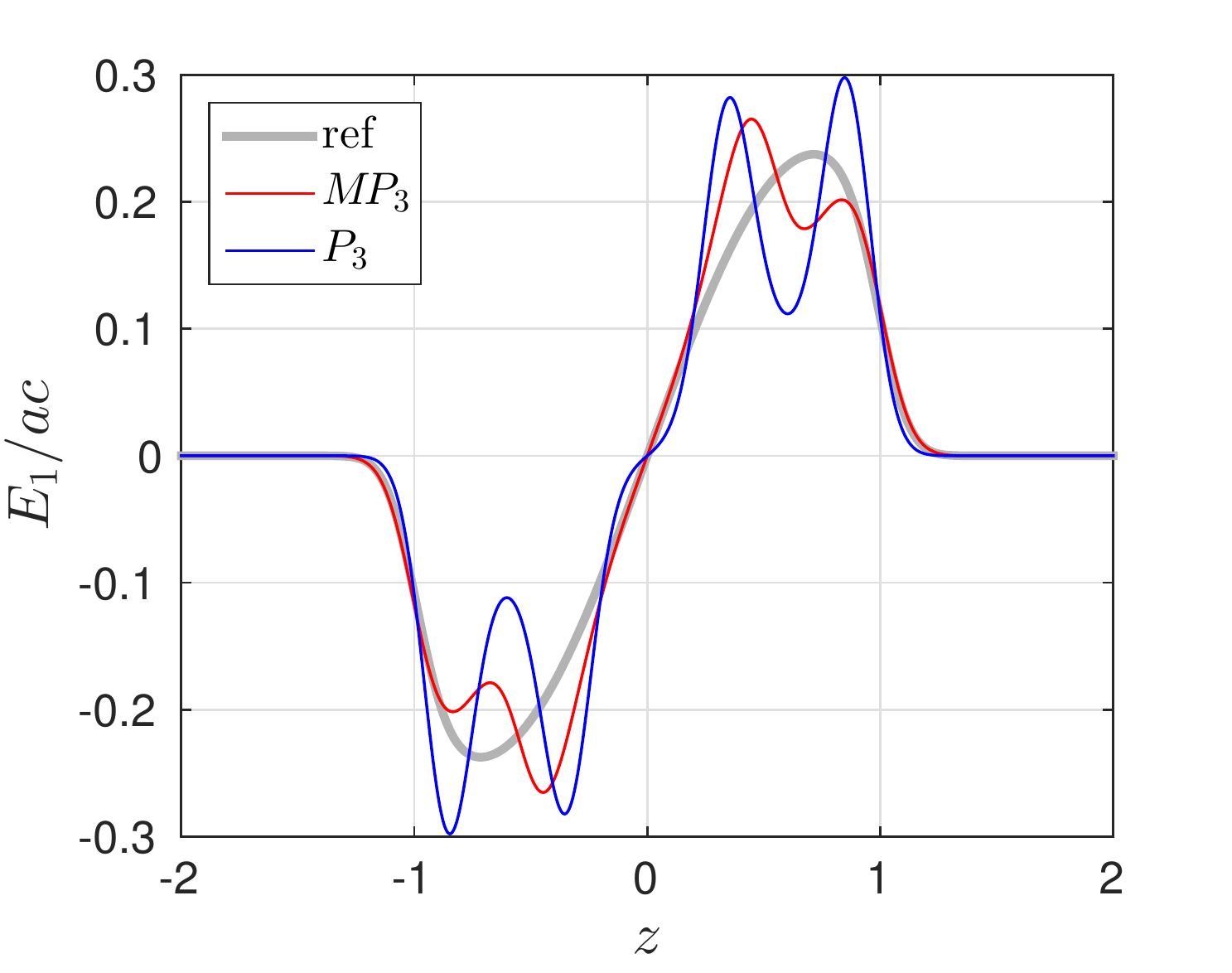}
  }
  \subfloat[${E_1}/{E_0}$ of \MP{3} and $P_3$]{
  \includegraphics[trim={0mm 0mm 0mm 0mm},clip,width=0.33\textwidth,height=0.16\textheight]{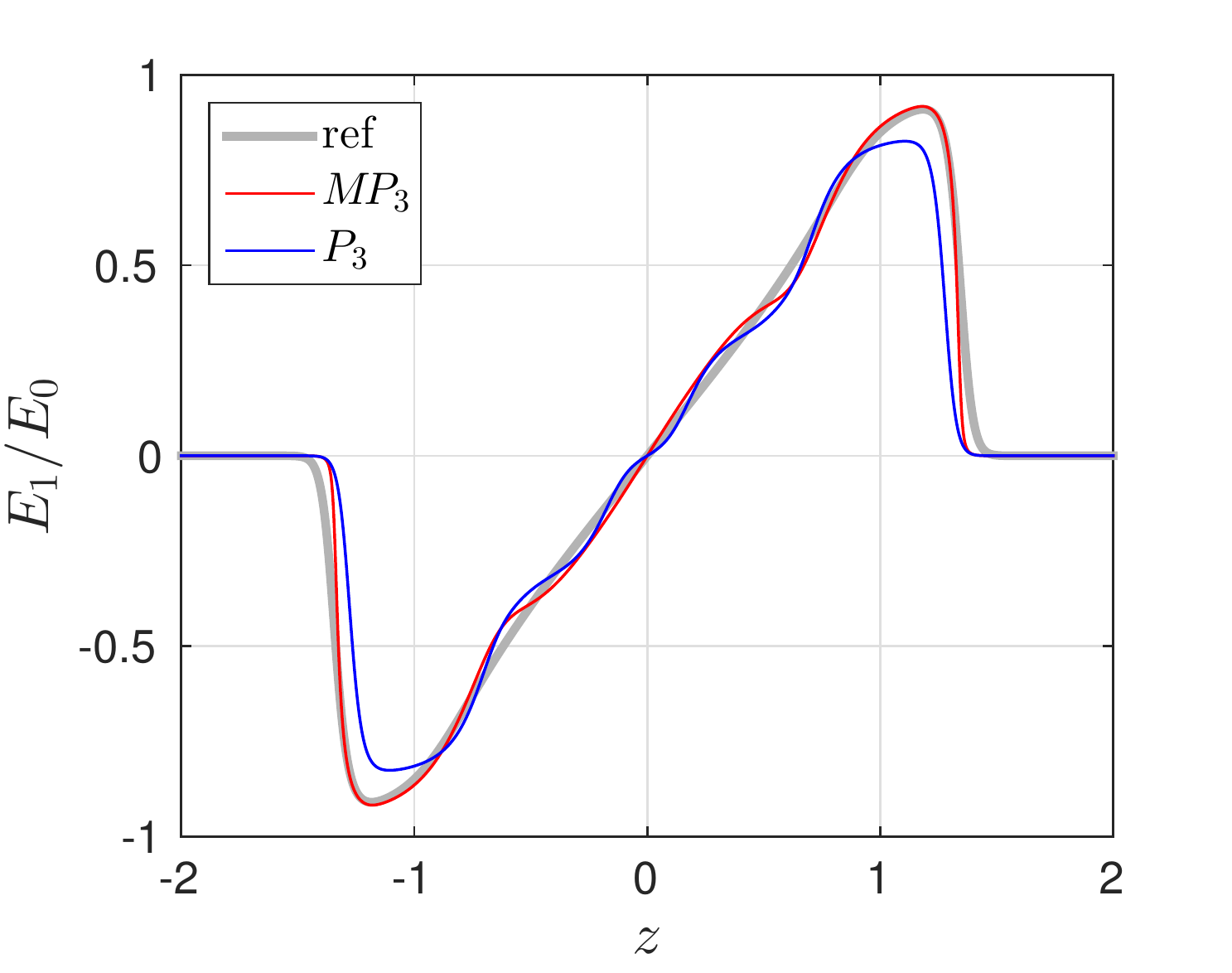}
  }\\
  \subfloat[$E_0$ of \MP{10} and $P_{10}$]{
  \includegraphics[trim={0mm 0mm 0mm 0mm},clip,width=0.33\textwidth,height=0.16\textheight]{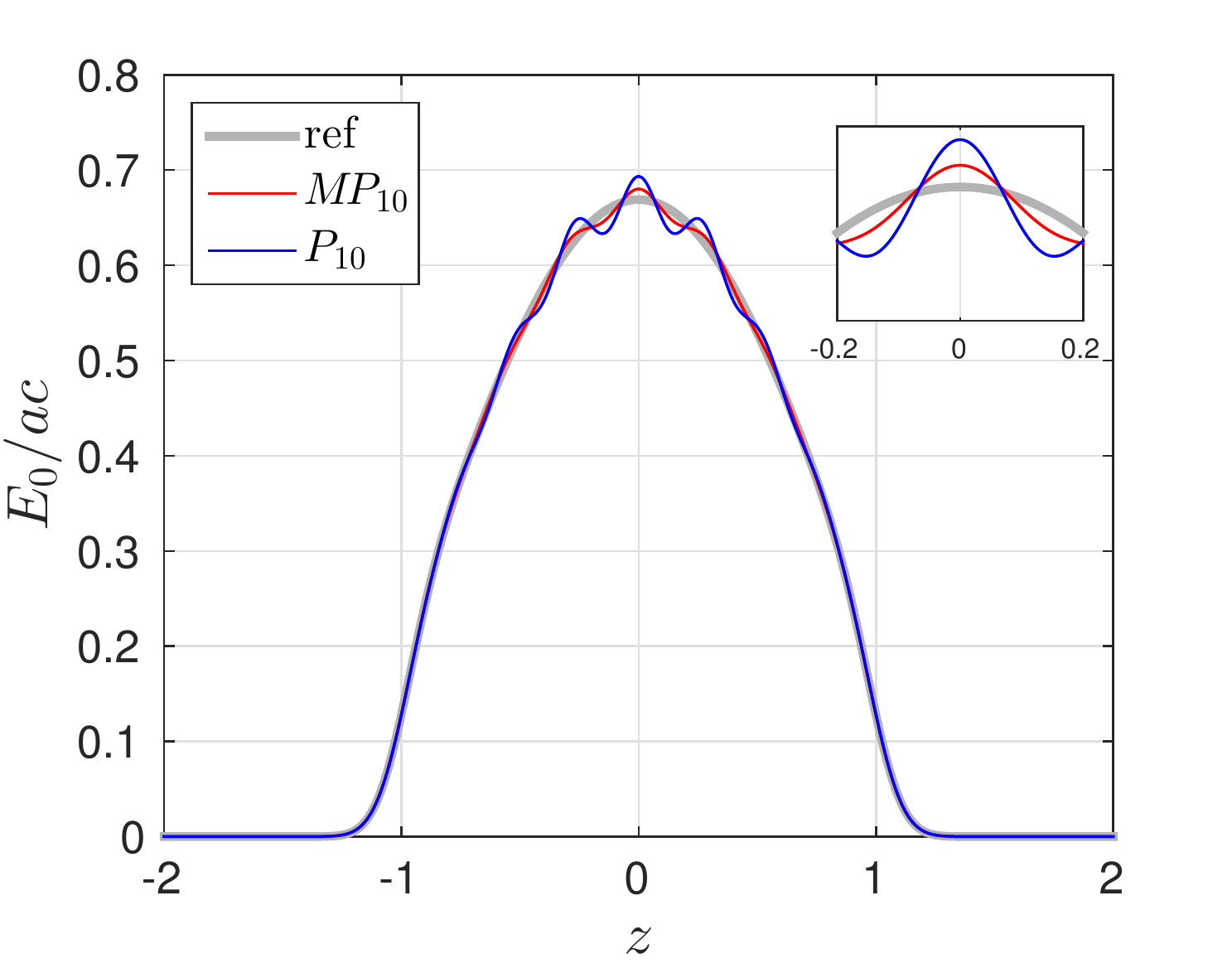}
  }
  \subfloat[$E_1$ of \MP{10} and $P_{10}$]{
  \includegraphics[trim={0mm 0mm 0mm 0mm},clip,width=0.33\textwidth,height=0.16\textheight]{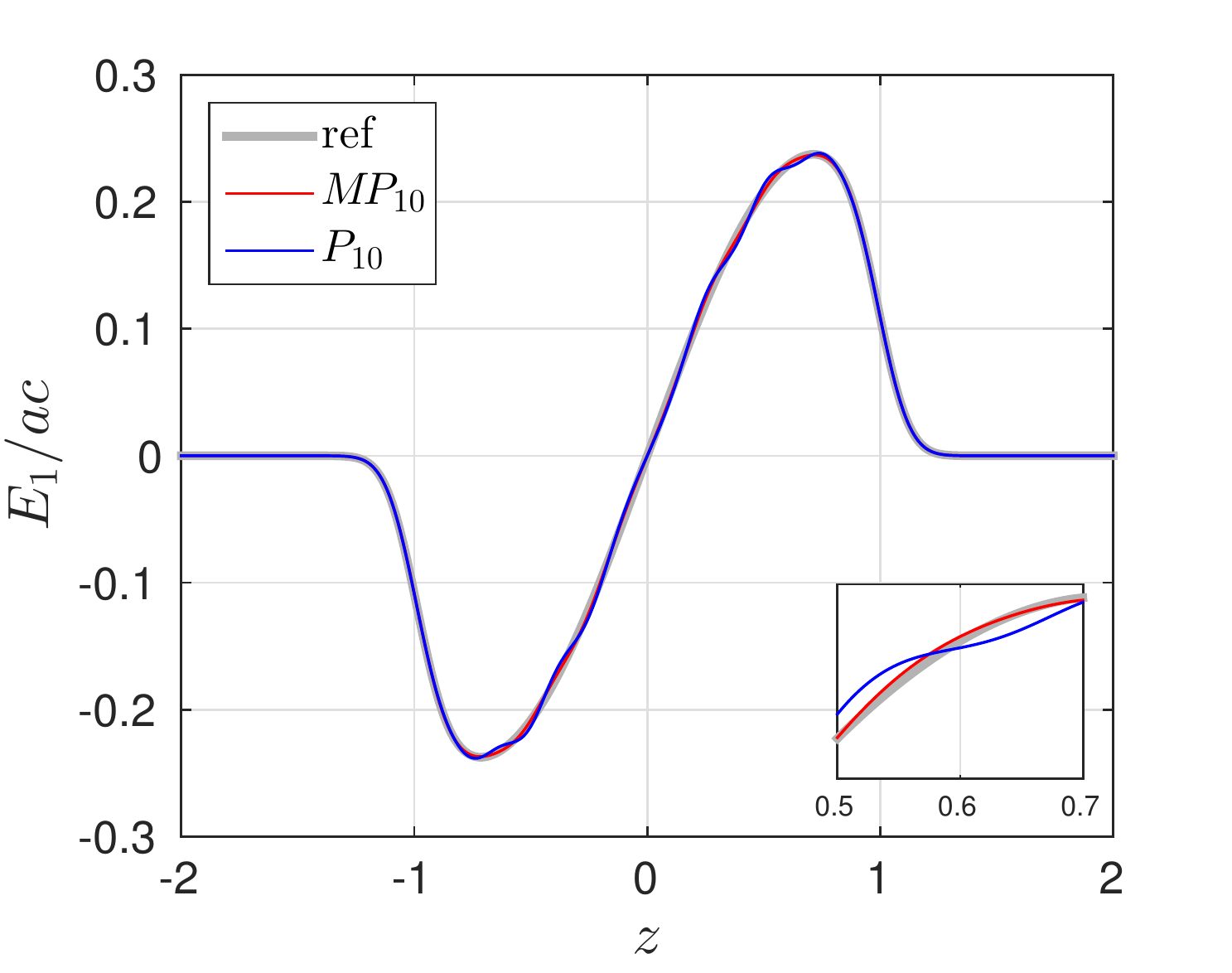}
  }
  \subfloat[${E_1}/{E_0}$ of \MP{10} and $P_{10}$]{
  \includegraphics[trim={0mm 0mm 0mm 0mm},clip,width=0.33\textwidth,height=0.16\textheight]{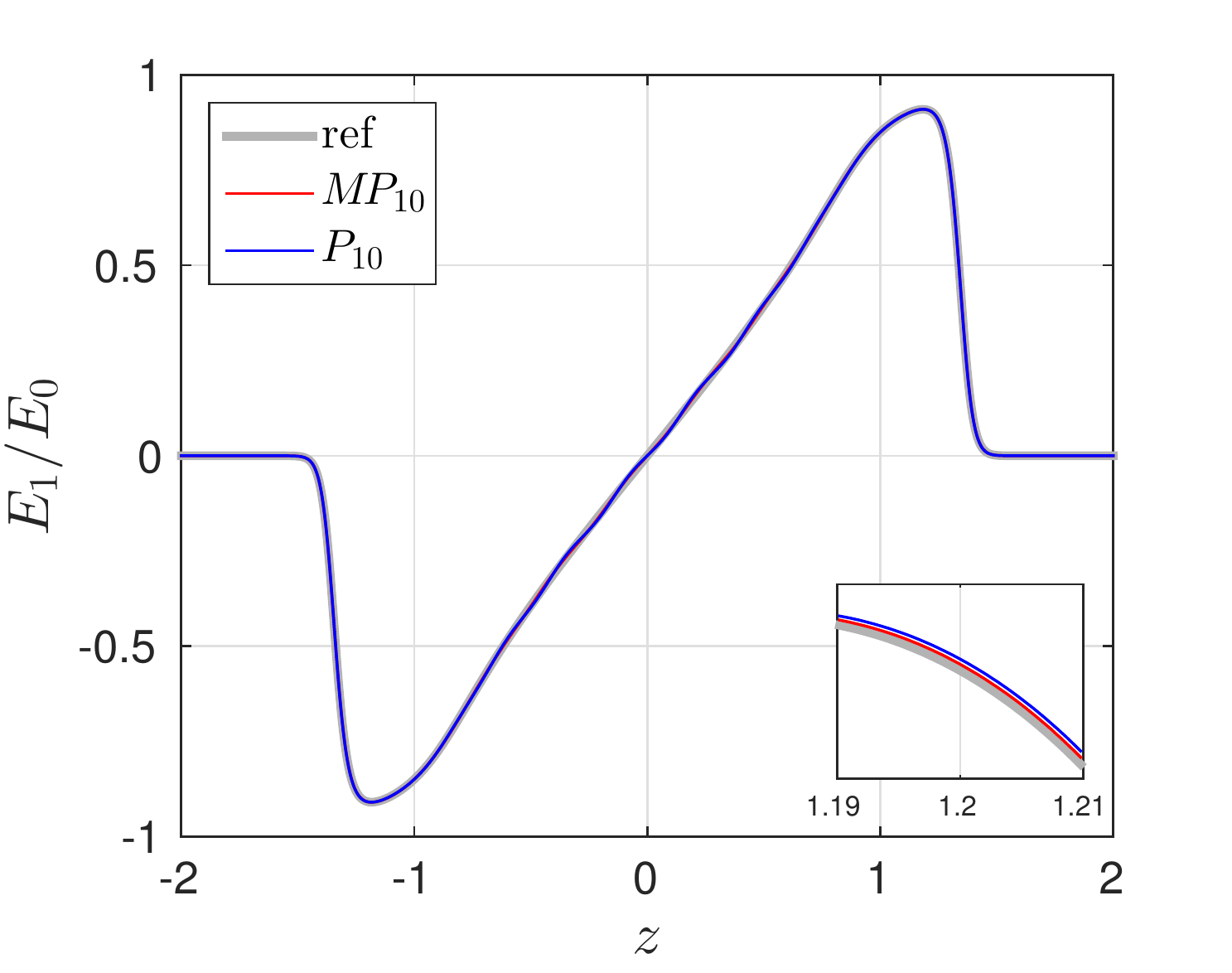}
  }
  \caption{Numerical results of the Gaussian source problem.}\label{fig:Gaussian}
\end{figure}

\begin{figure}[htb]
  \centering
  \subfloat{
  \includegraphics[width=0.4\textwidth,height=0.16\textheight]{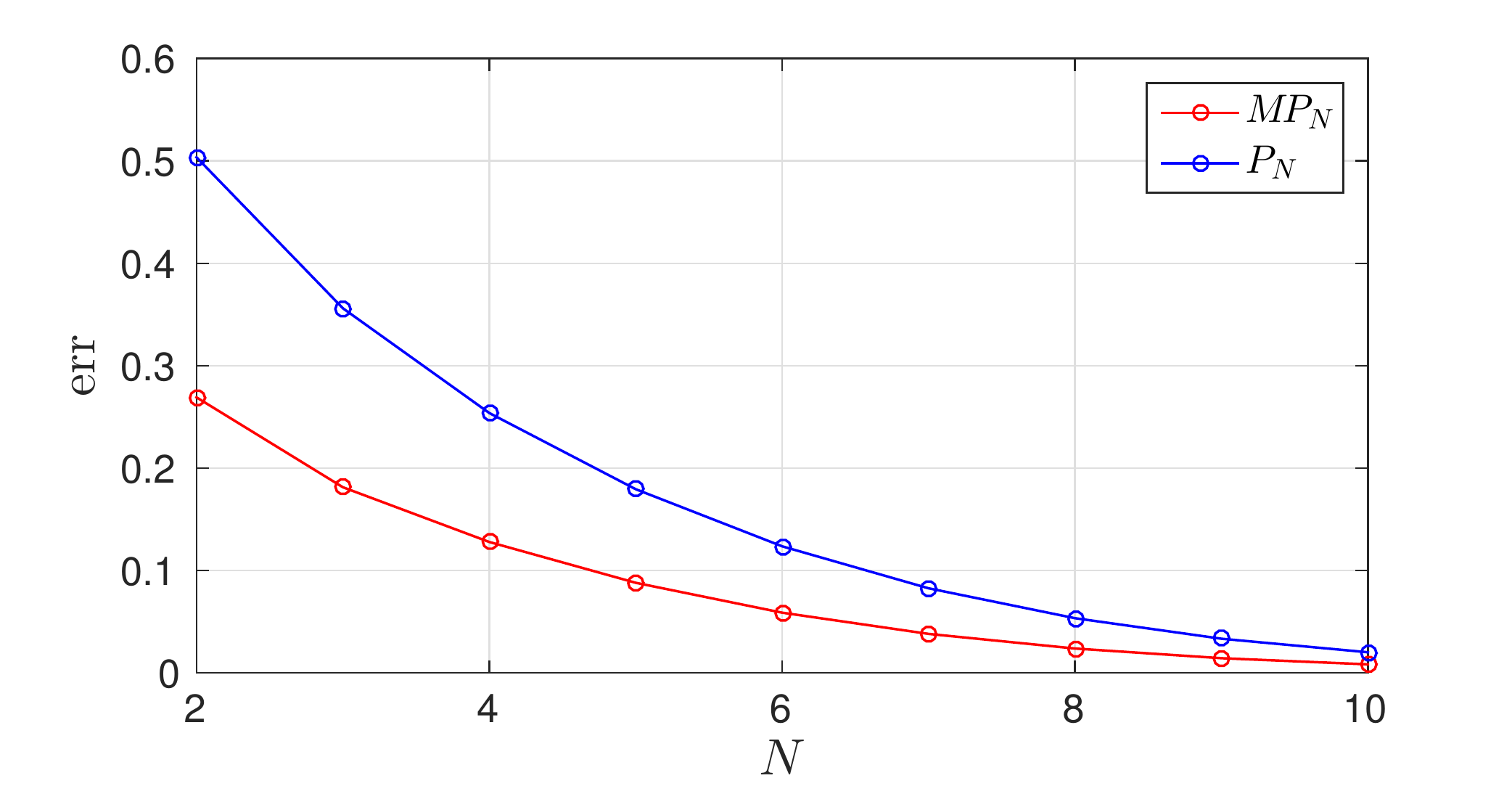}
  }\quad 
  \subfloat{
  \includegraphics[width=0.4\textwidth,height=0.16\textheight]{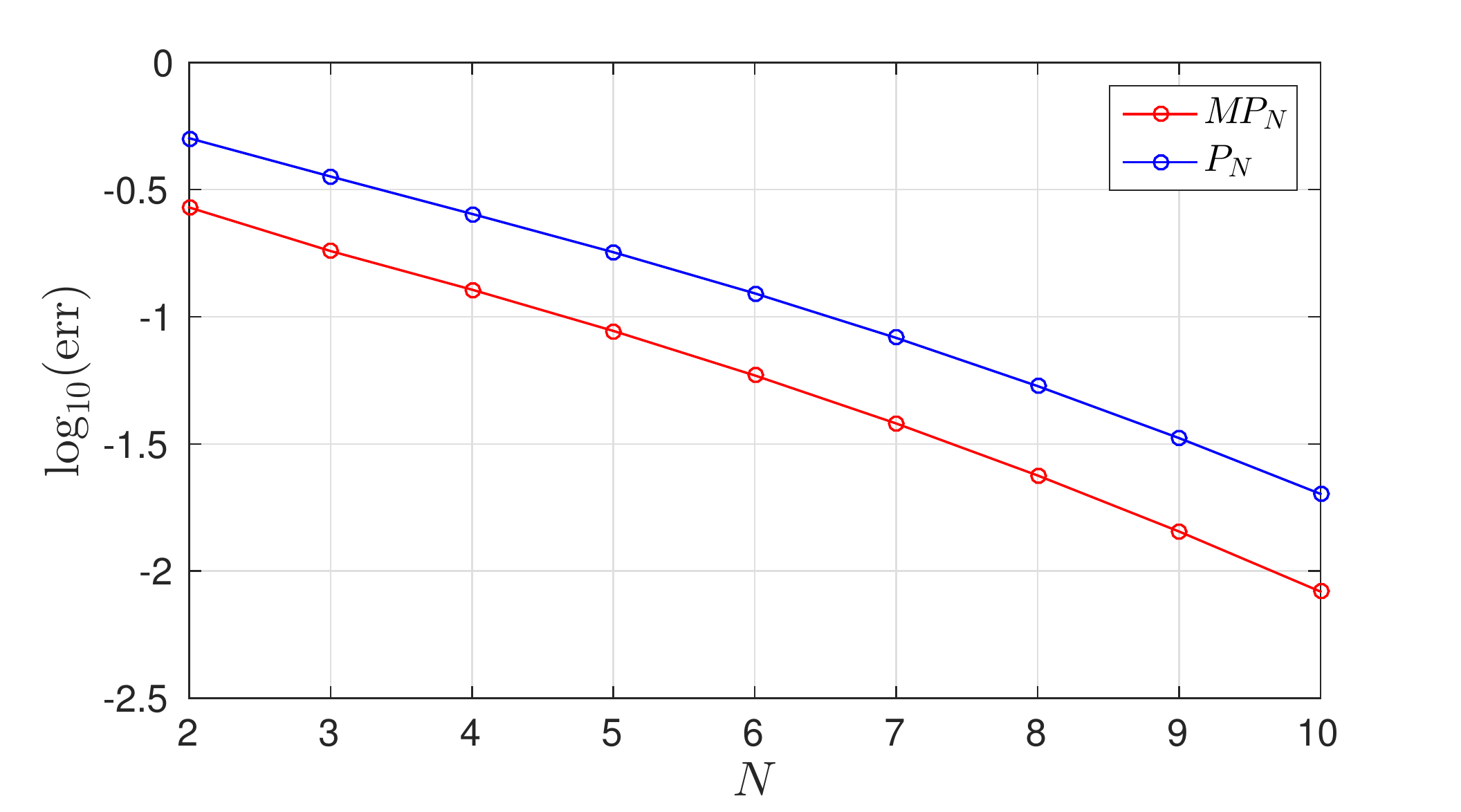}
  }
  \caption{Relative $\ell_2$ error (err) of $E_0$ of the Gaussian source problem
      in the linear scale (left figure) and the semi-logarithm scale (right figure).}
  \label{fig:Gaussian_error}
\end{figure}

In this problem, we set $\ctend=1$, thus the problem domain is
$[-2,2]$, and the number of cells is
$N_{\cell}=8000$, with $\Delta z = \frac{1}{2000}$. 
\Cref{fig:Gaussian} presents the numerical results, including
the profiles of $E_0$, $E_1$, and $\frac{E_1}{E_0}$ 
of the \MPN model, the \PN model and the
reference solution, which is the solution of the $P_{31}$ model. The relative
$\ell_2$ errors of $E_0$ of the \MPN and \PN models are shown in
\cref{fig:Gaussian_error}.
Spurious oscillations occur in the numerical solutions of $E_0$ of both the \MPN
model and the \PN model, and the oscillation amplitude decreases as the number
of moments increases. In the comparison of the \PN model, the \MPN model is more
effective in reducing the oscillations. 

Moreover, the $\frac{E_1}{E_0}$ in this problem can be sufficiently large (close
to $1$) to make the distribution function anisotropic. 
Clearly, both the numerical results of the \MPN model and the \PN model can
approximate the reference solution well while $N$ gets larger, and the
approximation of the \MPN model is much better than that of the \PN model,
according to both the results of the solutions and the relative error of the two
models. Precisely, \cref{fig:Gaussian_error} shows that the relative $\ell^2$
error of $E_0$ for the \MPN model is only half of that for the \PN model with
the same $N$.  This result also supports our argument in \cref{sec:MPNmodel} that the
\MPN model approximates better than the $P_N$ model for the anisotropic
distribution function.



\subsection{Pure absorbing problem}
This example is adopt to compare the approximation efficiency of \MPN model and
\PN model. 
We use the setup in \cite{buchan2005linear} as the computational domain
$[-5,5]$, the absorption coefficient $\sigma_a\equiv 0.5$, the scattering
coefficient $\sigma_s\equiv 0$, and the external source term
\[  
    s(z) = \begin{cases}
        a\,c,\quad &-1\leq z\leq 1,\\
        0,\quad &\text{otherwise}.
    \end{cases}
\]
The material coupling term is ignored, and the vacuum boundary condition are
prescribed at the both boundaries. We care about the steady-state solution of
this problem. In the practical simulation, we set the initial intensity as
$I_0(z,\mu) = 10^{-8}ac$ with the number of cells $N_{\cell}=20000$ and $\Delta
z = \frac{1}{2000}$.

\begin{figure}[htb]
  \centering
  \subfloat[$E_0$ of \MPtwo and $P_2$]{
  \includegraphics[trim={0mm 0mm 0mm 0mm},clip,width=0.33\textwidth,height=0.16\textheight]{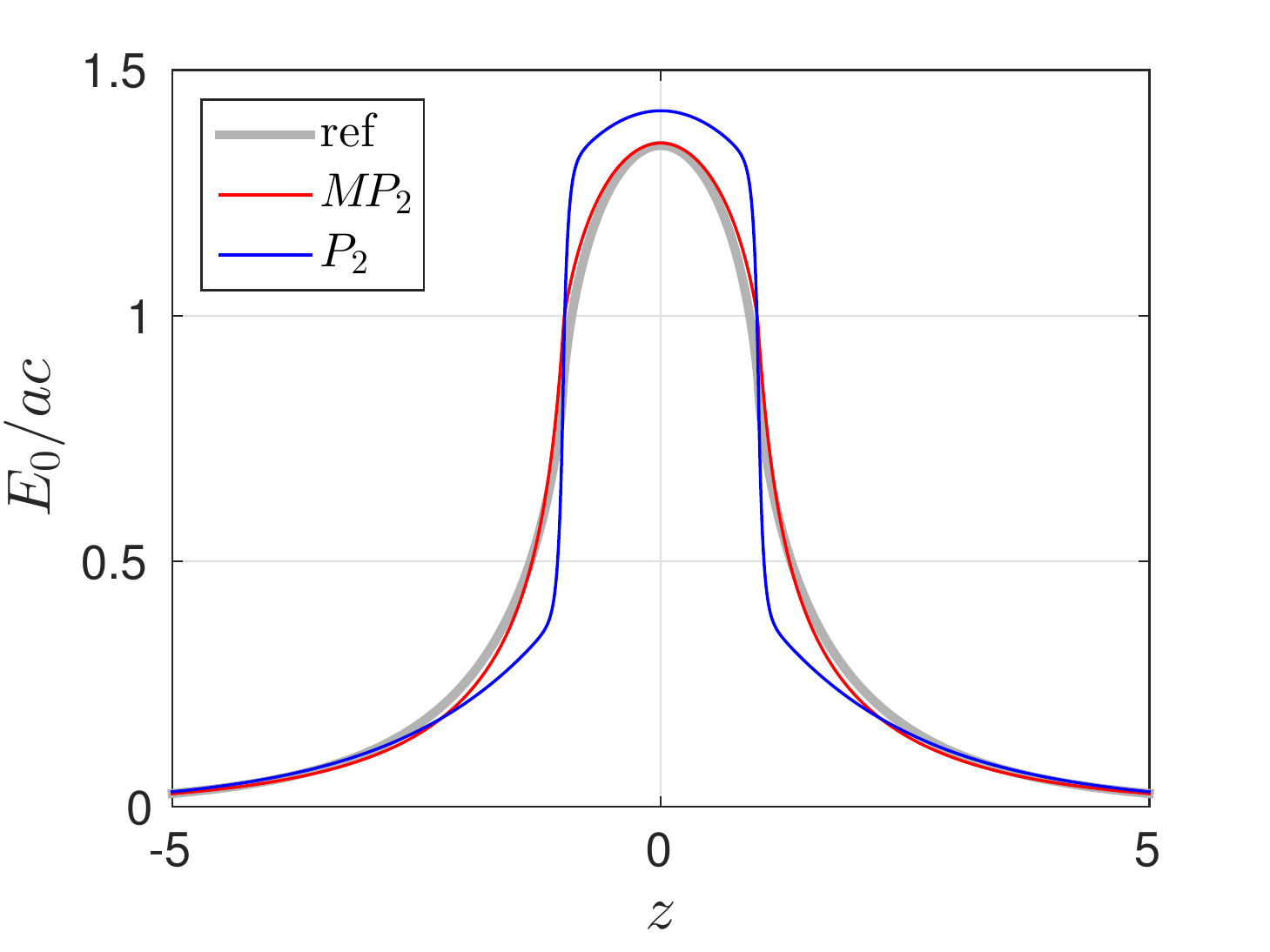}
  }
  \subfloat[$E_1$ of \MPtwo and $P_2$]{
  \includegraphics[trim={0mm 0mm 0mm 0mm},clip,width=0.33\textwidth,height=0.16\textheight]{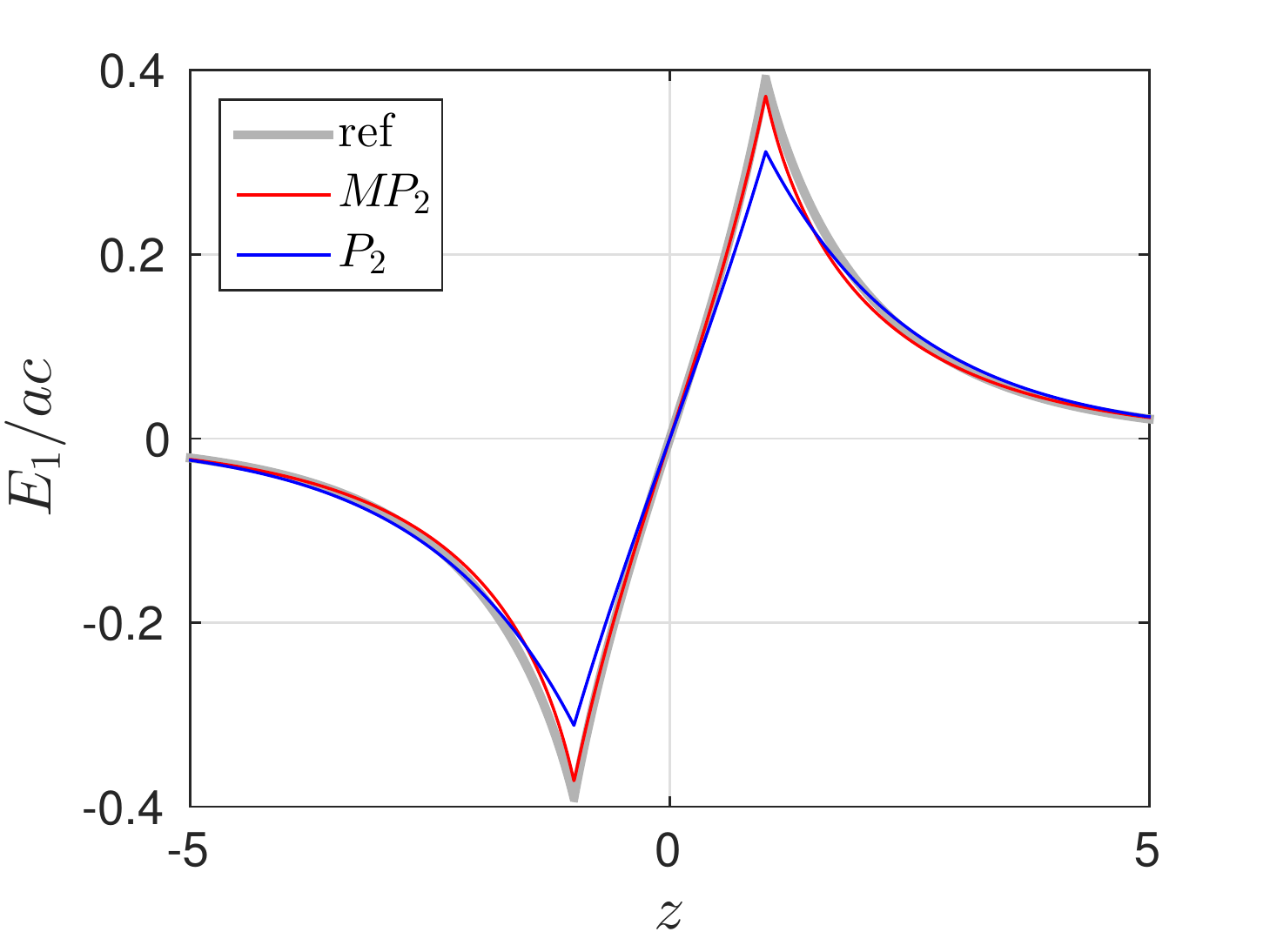}
  }
  \subfloat[$E_1/E_0$ of \MPtwo and $P_2$]{
  \includegraphics[trim={0mm 0mm 0mm 0mm},clip,width=0.33\textwidth,height=0.16\textheight]{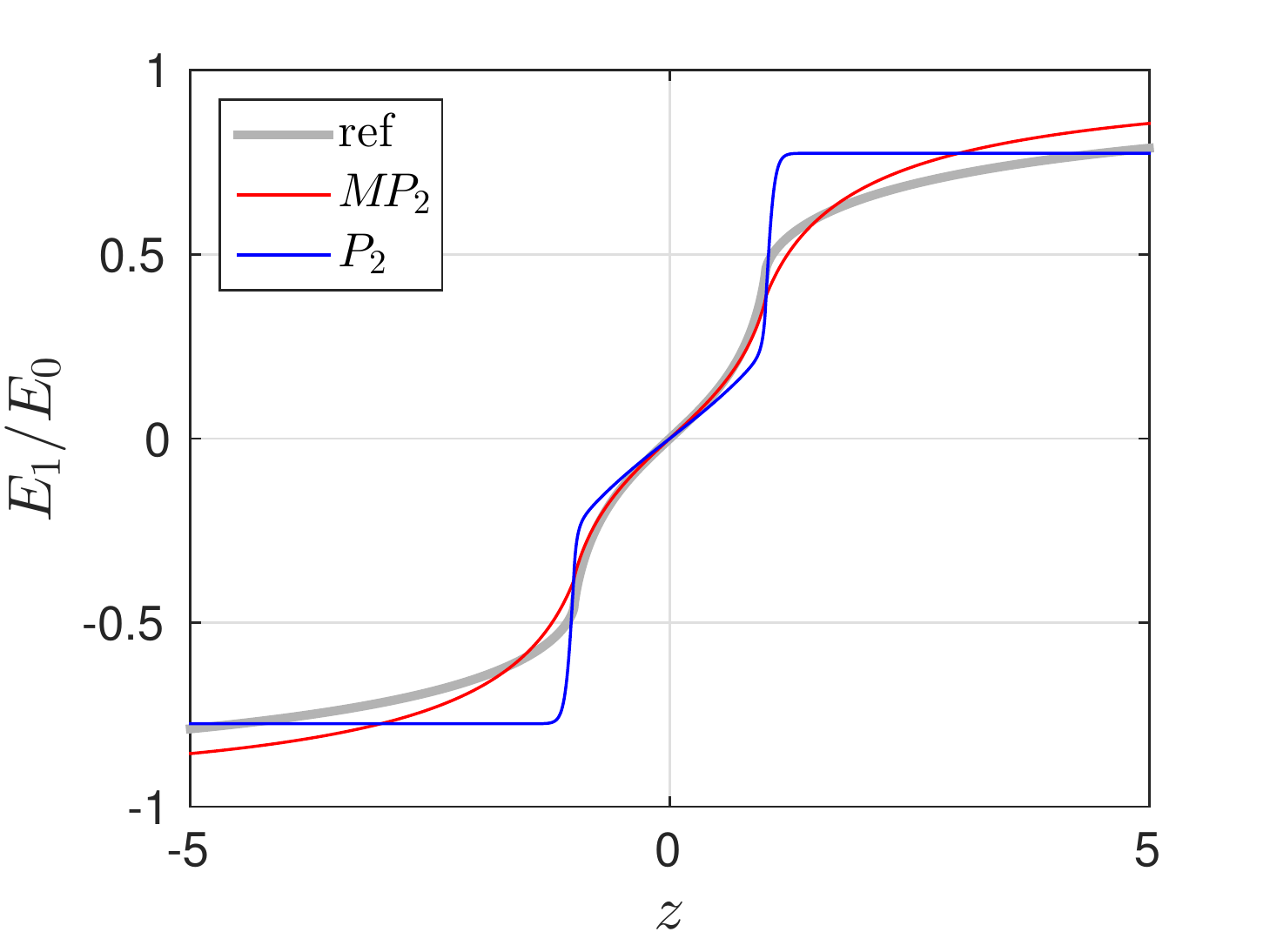}
  }\\
  \subfloat[${E_0}$ of \MP{10} and $P_{10}$]{
  \includegraphics[trim={0mm 0mm 0mm 0mm},clip,width=0.33\textwidth,height=0.16\textheight]{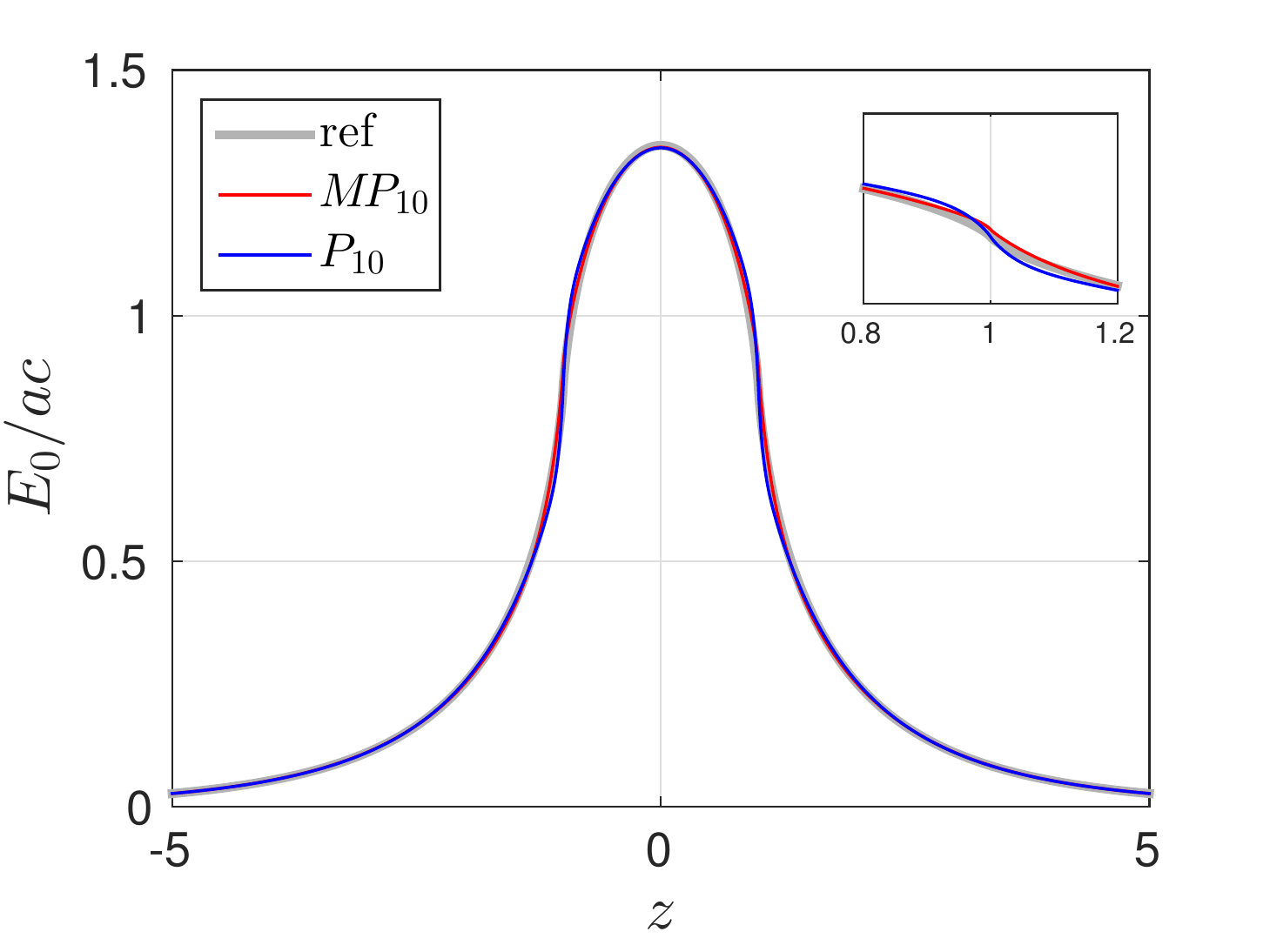}
  }
  \subfloat[${E_1}$ of \MP{10} and $P_{10}$]{
  \includegraphics[trim={0mm 0mm 0mm 0mm},clip,width=0.33\textwidth,height=0.16\textheight]{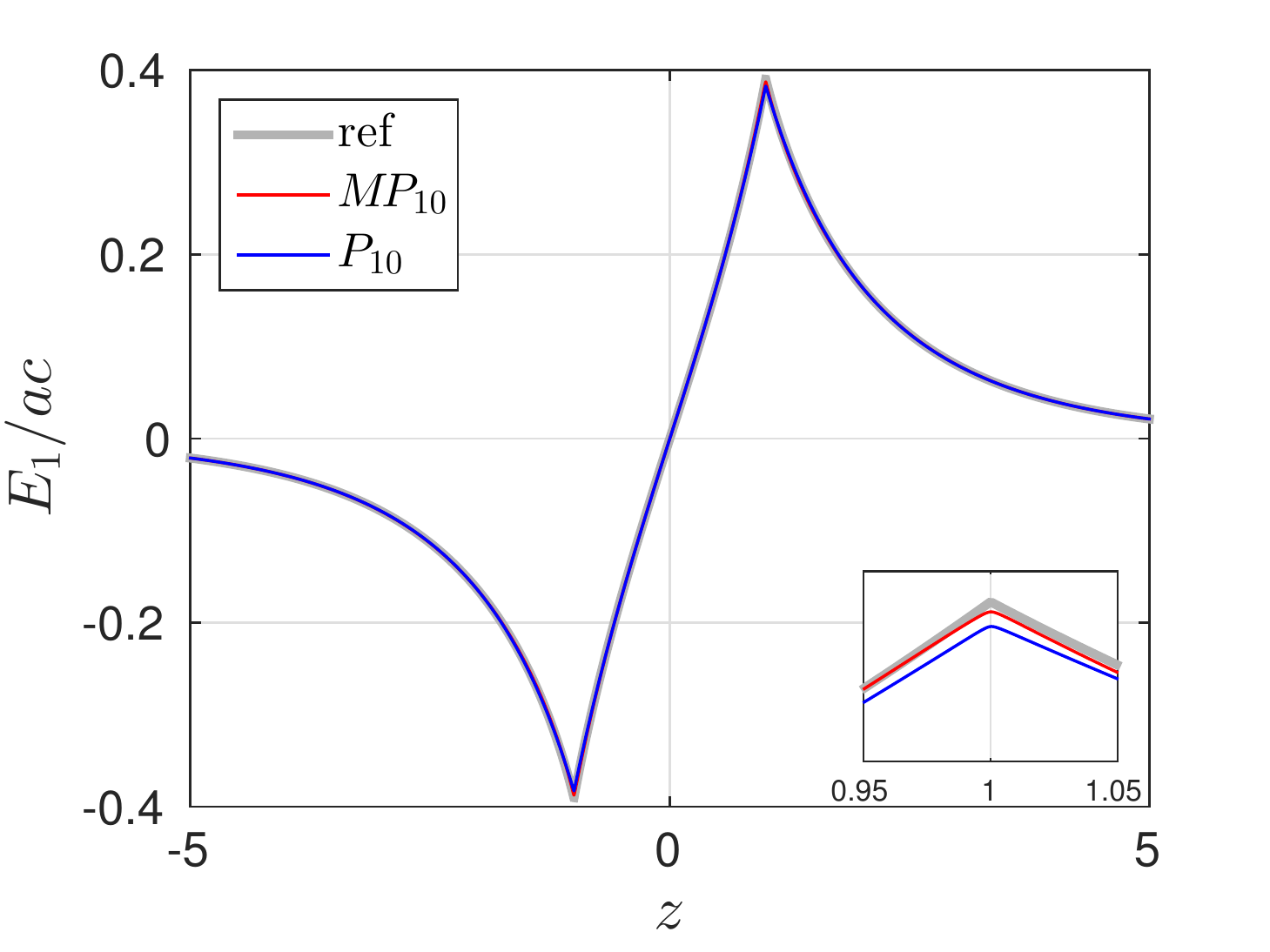}
  }
  \subfloat[${E_1/E_0}$ of \MP{10} and $P_{10}$]{
  \includegraphics[trim={0mm 0mm 0mm 0mm},clip,width=0.33\textwidth,height=0.16\textheight]{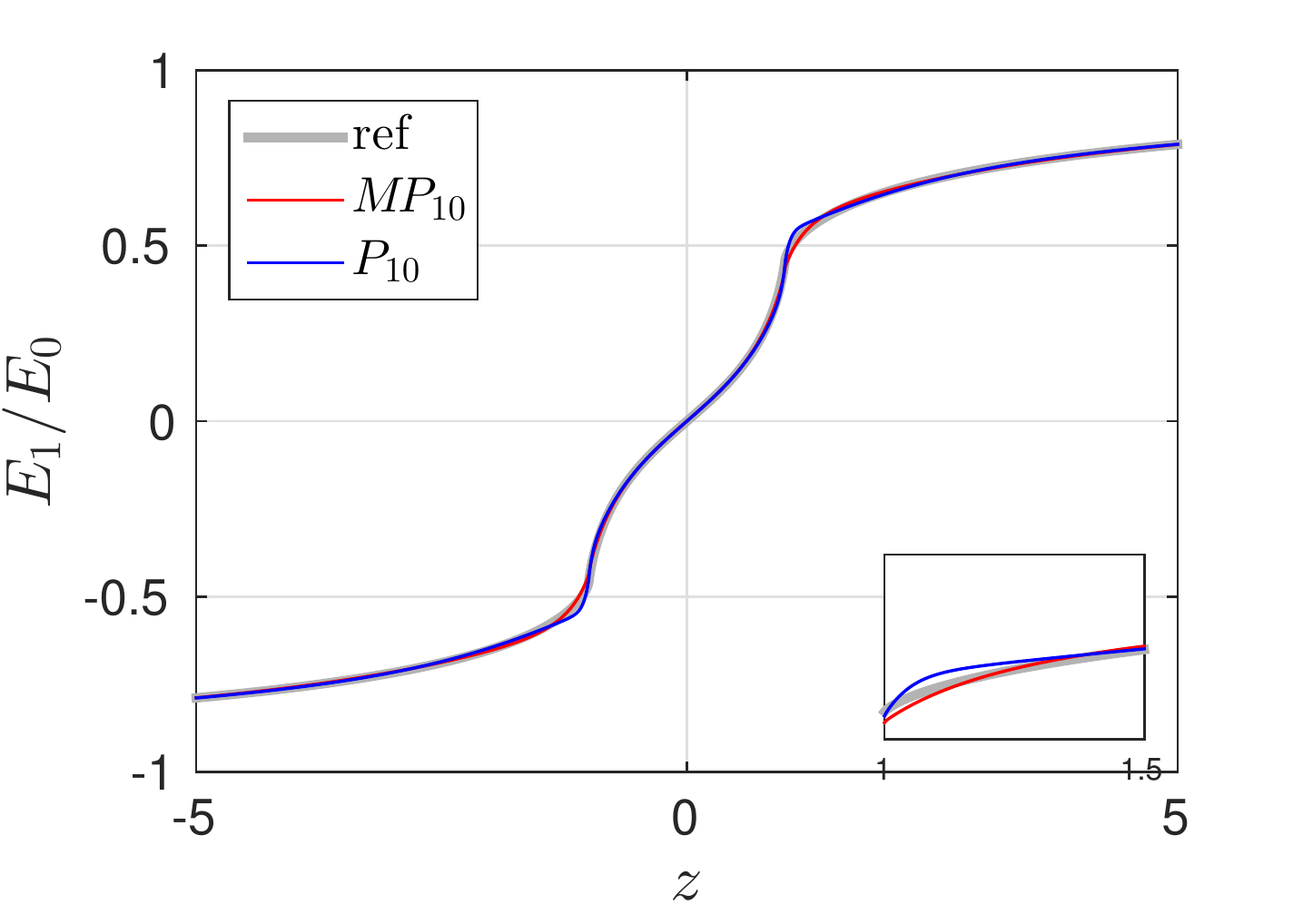}
  }
  \caption{Numerical results of the pure absorbing problem.}
  \label{fig:PureAbsorb_data}
\end{figure}

\begin{figure}[htb]
  \centering
  \subfloat{
  \includegraphics[width=0.4\textwidth,height=0.16\textheight]{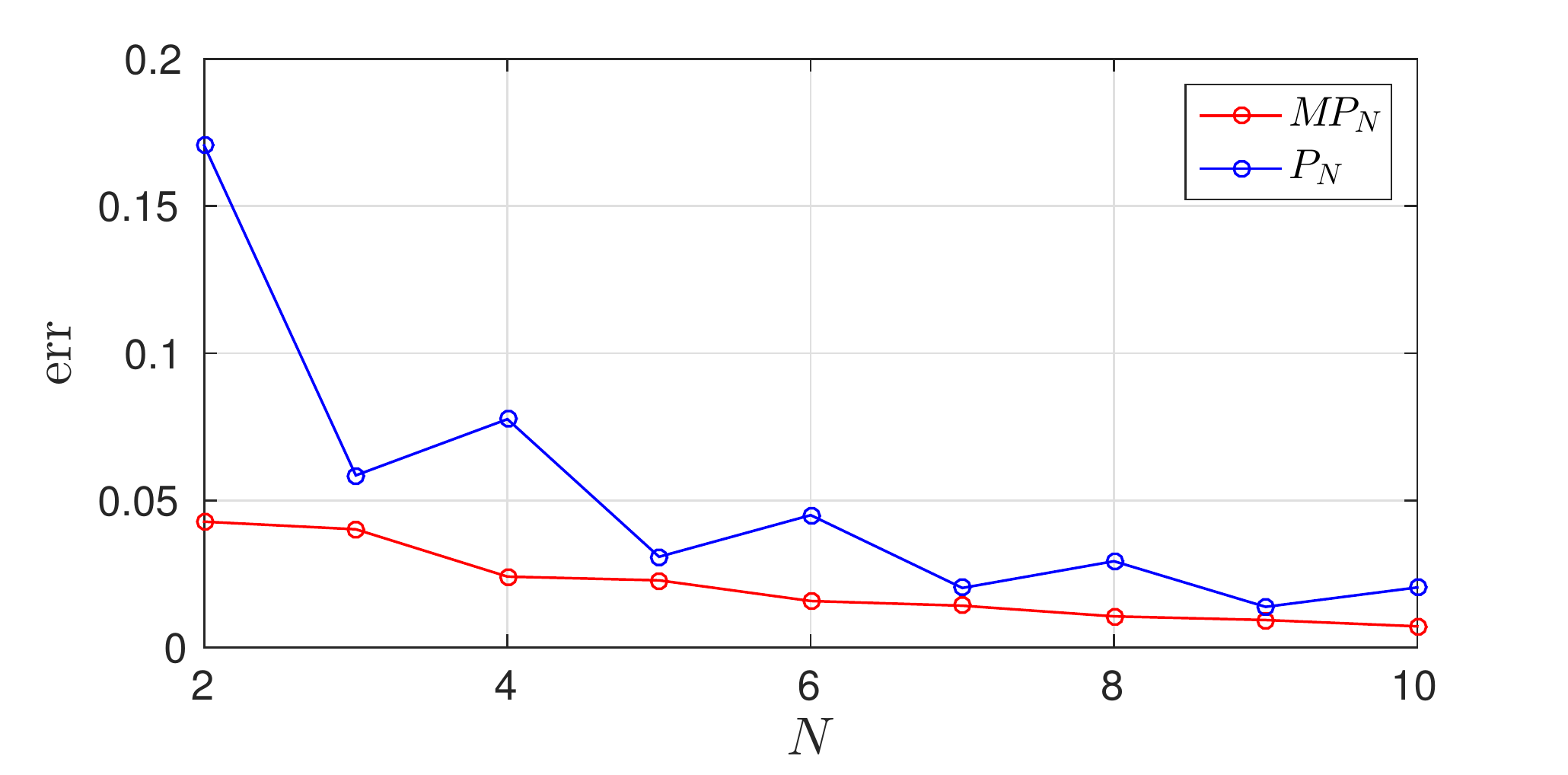}
  }\quad 
  \subfloat{
  \includegraphics[width=0.4\textwidth,height=0.16\textheight]{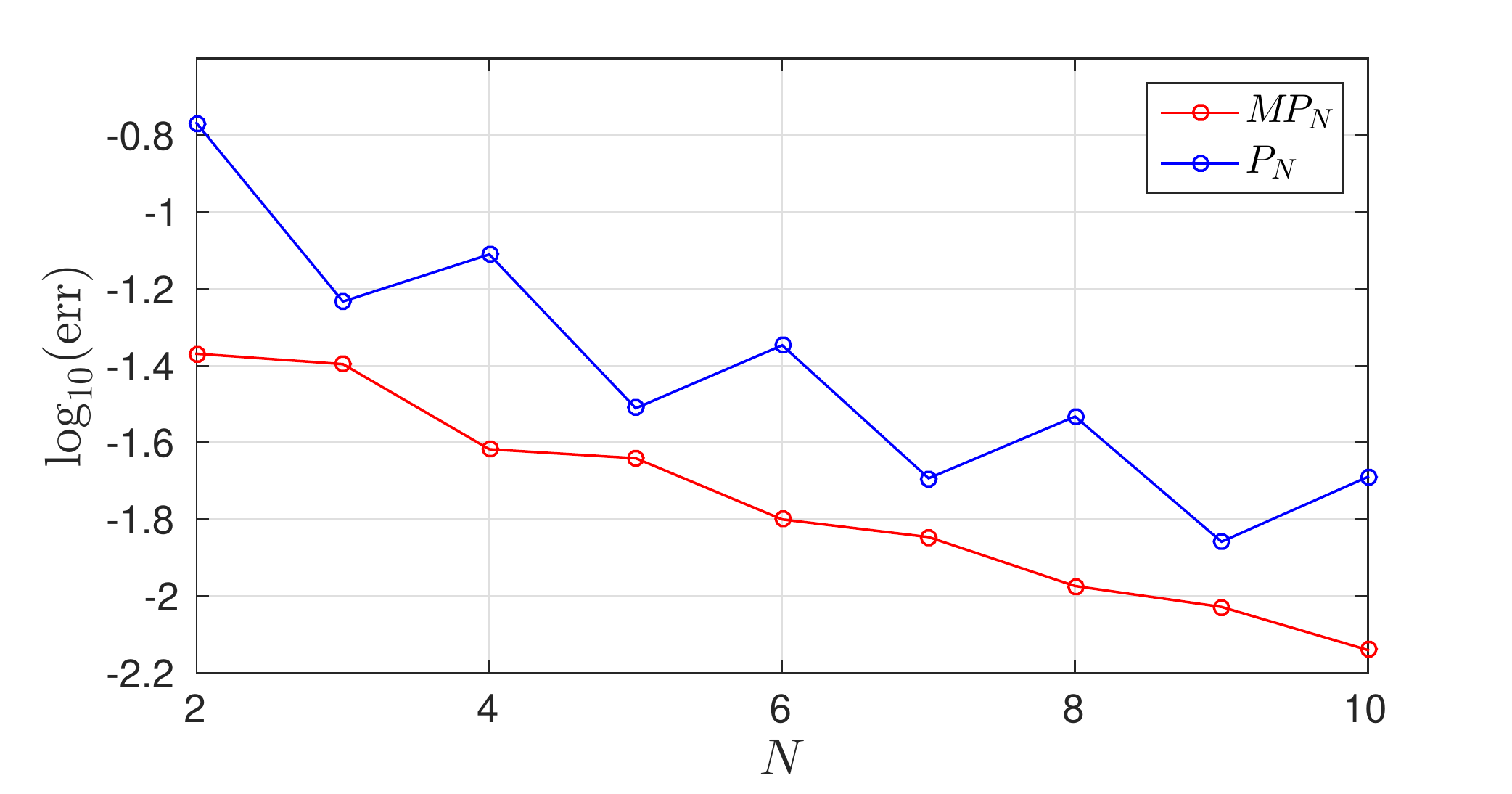}
  }\quad 
  \caption{Relative $\ell_2$ error (err) of $E_0$ of the pure absorbing problem
  in the linear scale (left figure) and the semi-logarithm scale (right
  figure).}
    \label{fig:PureAbsorb_error}
\end{figure}

\Cref{fig:PureAbsorb_data} presents the profile of $E_0$, $E_1$, and
$\frac{E_1}{E_0}$ of the \MPN model and the \PN model with $N=2,10$, and the
relative $\ell_2$ errors of $E_0$ for different $N$ are shown in
\cref{fig:PureAbsorb_error}. All the conclusions in \cref{sec:gaussian}
are also valid for this example. What we want to point out is that 
the results of the \MPtwo model is good enough for this example.

\subsection{Su-Olson problem}
The Su-Olson problem \cite{Olson2000Diffusion} is a non-equilibrium radiative
transfer problem with a material coupling term. The computation domain is $[0,
30]$ and the absorption and scattering opacity are $\sigma_a=1$ and
$\sigma_s=0$, respectively. The external source terms $s$ is given by
\[  
  s(z,t) = \begin{cases}
      ac, &\text{if } 0\leq z\leq \frac{1}{2}, 
      \text{ and } 0\leq c\,t\leq 10,\\
      0, &\text{otherwise}.
  \end{cases}
\]
The relationship between the temperature and the internal energy is given by
$e(T) = aT^4$. The left boundary condition is the reflective boundary condition,
while the right boundary condition is the vacuum boundary condition.

\begin{figure}[ht]
  \centering
  \subfloat{
	\includegraphics[trim={0mm 0mm 0mm 0mm},clip,
	width=0.33\textwidth,height=0.16\textheight]{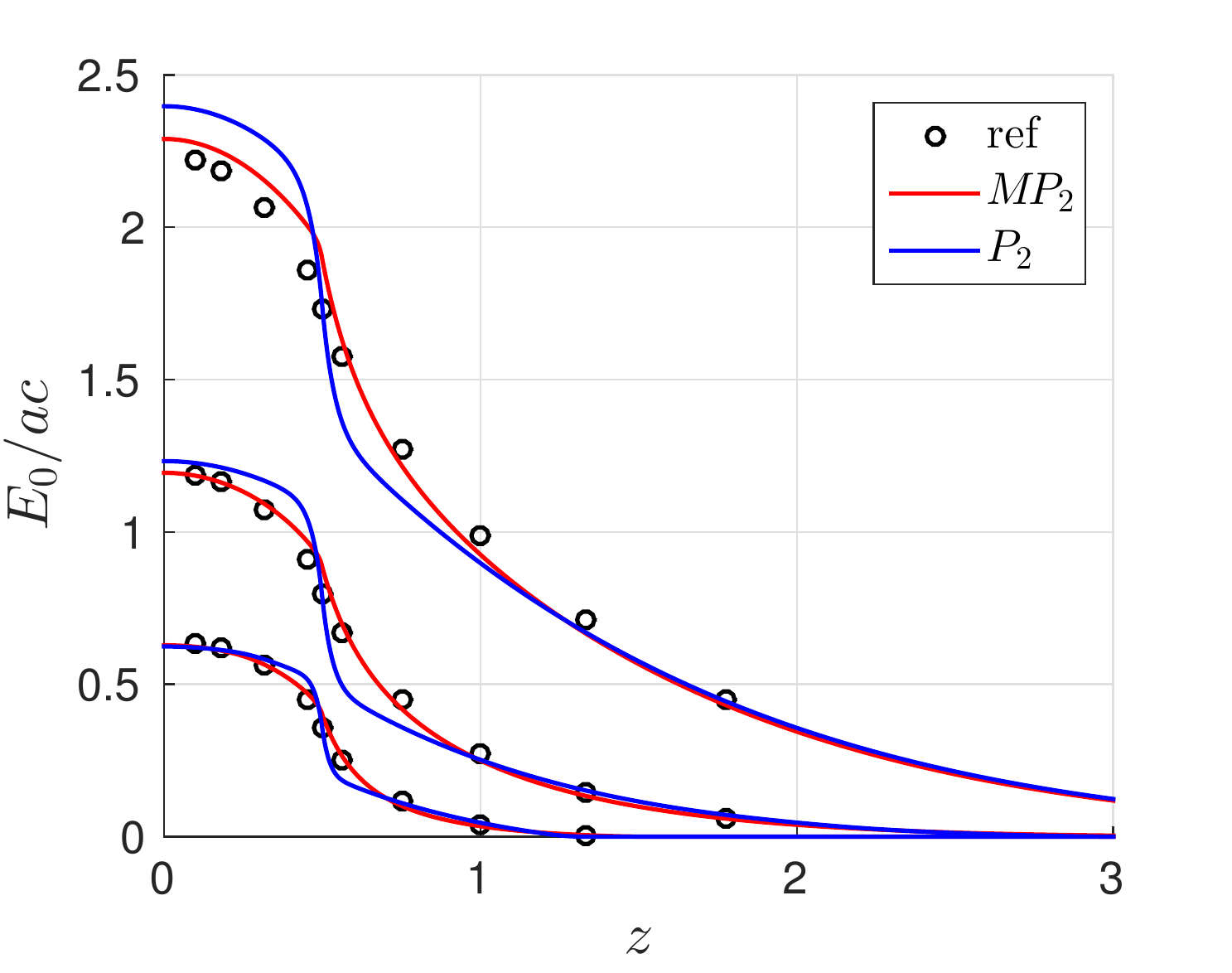}
  }
  \subfloat{
  \includegraphics[trim={0mm 0mm 0mm 0mm},clip,width=0.33\textwidth,height=0.16\textheight]{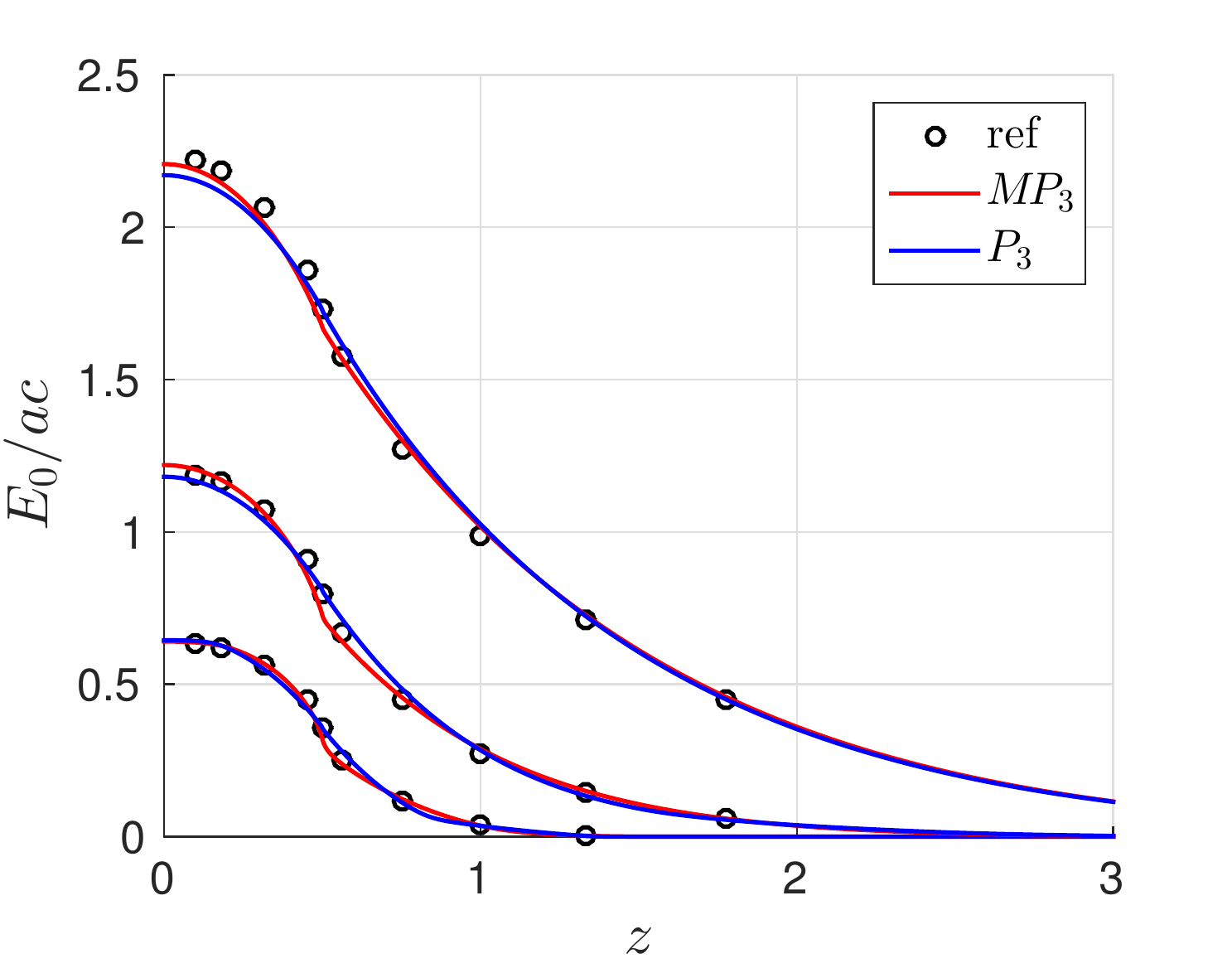}
  }
  \subfloat{
  \includegraphics[trim={0mm 0mm 0mm 0mm},clip,width=0.33\textwidth,height=0.16\textheight]{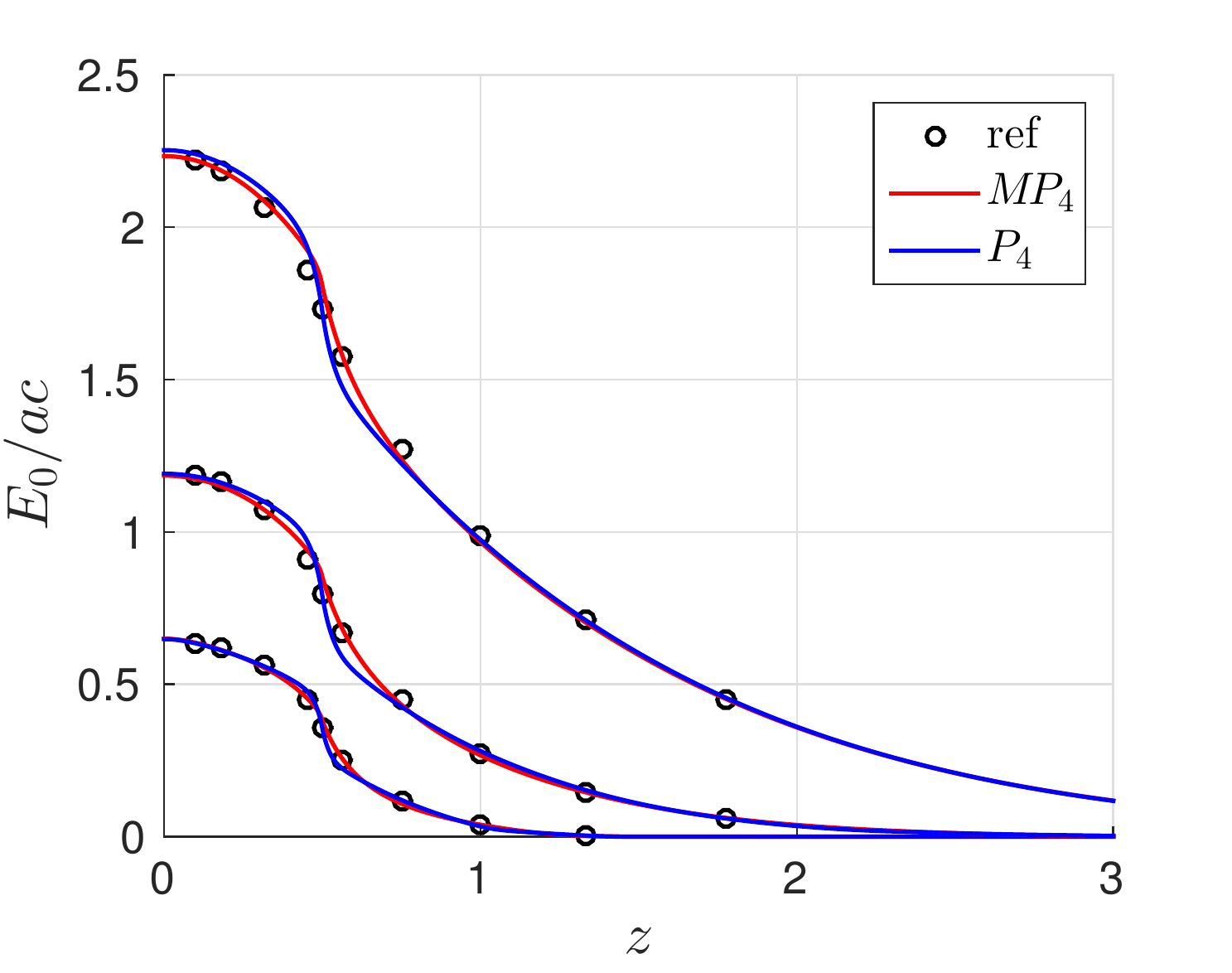}
  }\\
  \subfloat{
  \includegraphics[trim={0mm 0mm 0mm 0mm},clip,width=0.33\textwidth,height=0.16\textheight]{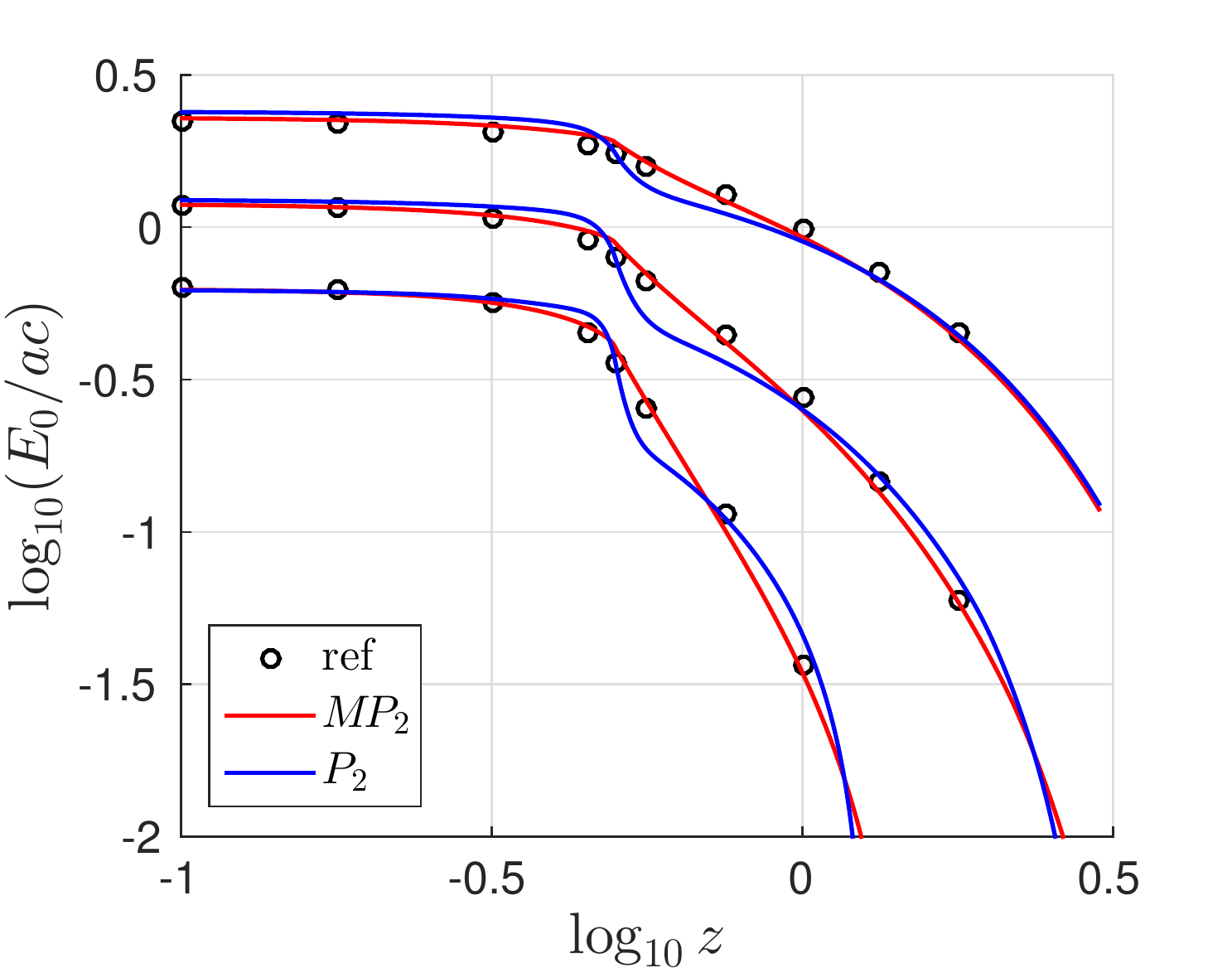}
  }
  \subfloat{
  \includegraphics[trim={0mm 0mm 0mm 0mm},clip,width=0.33\textwidth,height=0.16\textheight]{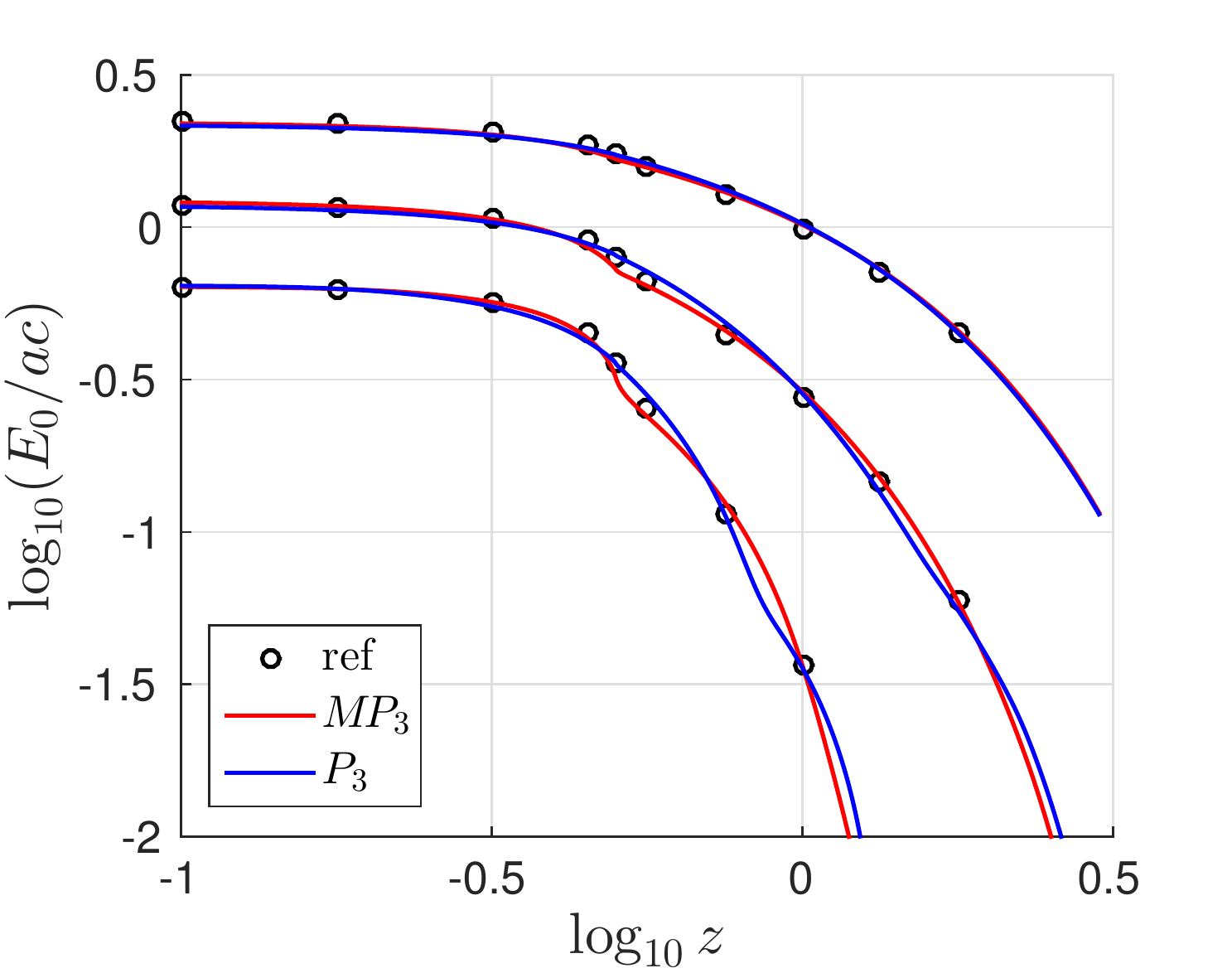}
  }
  \subfloat{
  \includegraphics[trim={0mm 0mm 0mm 0mm},clip,width=0.33\textwidth,height=0.16\textheight]{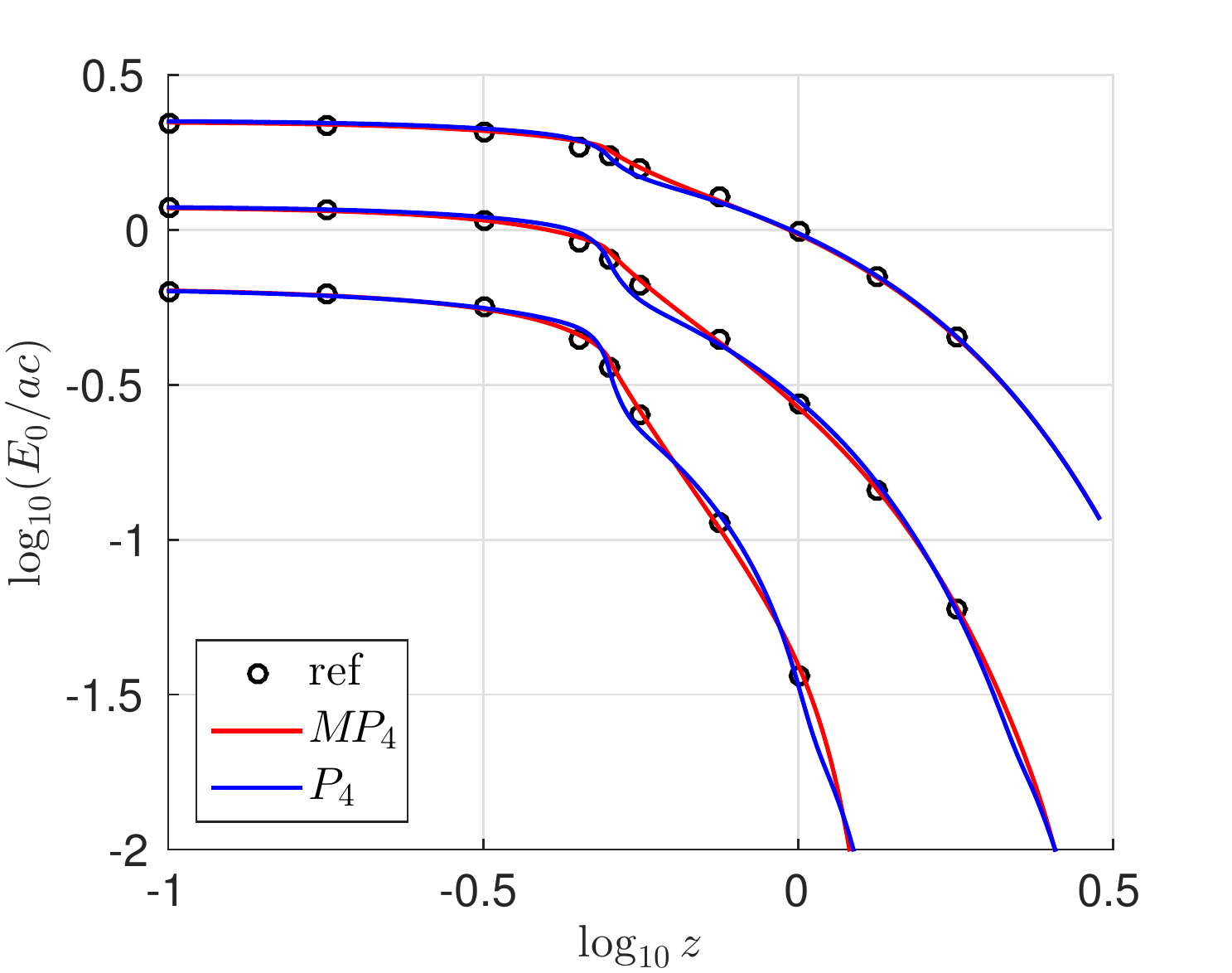}
  }\\
  \caption{Profile of $E_0$ of the \MPN and \PN models with $N=2,3,4$ (from left
    to right) for the Su-Olson problem at different time $\ctend =1$, $3.16$ and
      $10$ (from down to up in each figure) with the linear scale (upper figures)
      and logarithm scale (lower figures).}
  \label{fig:SuOlson_data}
\end{figure}

The setup of the simulation is $N_{\cell}=60000$ with $\Delta z=\frac{1}{2000}$.
\Cref{fig:SuOlson_data} presents the profile of $E_0$ of the \MPN model and the
\PN model with $N=2,3,4$ at different end time $\ctend=1$, $3.16$, and $10$. The
reference solution is the semi-analytic solution in \cite{su1997analytical}.

For both the \PN and \MPN models, the results agree with the reference well as
$N$ increases.  The \MPN model has the capability to simulate such benchmark a
few moments, for instance, $N=2$.  Moreover, in the comparison of the \PN model,
the \MPN model shows superiority to handle such material coupling problem.


%% file: conclusion.tex
\section{Conclusion} \label{sec:conclusion}
According to the two viewpoints on the \PN and \MN models, we proposed the \MPN
model by expanding the specific intensity around the ansatz of the \Mone model
in terms of orthogonal polynomials. The certain selection of the weight function
permitted the \MPN model to simulate the problems with strong anisotropic
distribution. The \MPN model had an explicit expression of the closure for
arbitrary $N$, which allowed us to directly solve it cheaply in the comparison of the
\MN model. In all the numerical tests, only a few moments (for instance, $N=2$
is good enough for many tests) were required to give good numerical results.
Hence, it was believed that the \MPN model could be used
to solve the RTE fast and accurately. The current work focused on the novel idea
on the construction of the \MPN model to approximate the RTE for a grey medium in
the slab geometry. The extension to the general medium and 3D case was in
process.

\section*{Acknowledgements}
The work of Y.F. is partially supported by the U.S. Department of Energy, Office
of Science, Office of Advanced Scientific Computing Research, Scientific
Discovery through Advanced Computing (SciDAC) program and the National Science
Foundation under award DMS-1818449.
The work of R.L. and L.Z. is partially supported by Science Challenge Project,
No. TZ2016002 and the National Natural Science Foundation of China (Grant No.
91630310 and 11421110001, 11421101).


%% file: appendix.tex
\appendix
\section{Proof of the \cref{pro:M1P2}}\label{sec:proof_M1P2}
In this appendix, we prove the \cref{pro:M1P2}. 
Except \cref{pro:M1P2} b), the case $E_1=0$ is trival.
Without loss of generality, we need only to check the case $\alpha\in(0,1)$
for all the properties except \cref{pro:M1P2} b).
Noticing \eqref{eq:E3ofM1P2}, we can obtain the following relationships with 
direct calculation:
\begin{align}
    \pd{E_3}{E_2} &= -\frac{ C_1}{ C_2},\label{eq:dev_E3E2}\\
    \pd{\alpha}{E_0} &= \frac{4 - 2\sqrt{4 - 3(E_1/E_0)^2}}{E_1 \sqrt{4
    - 3(E_1/E_0)^2}},\label{eq:dev_alE0}\\
    \pd{\alpha}{E_1} &= -E_0\frac{4 - 2\sqrt{4 - 3(E_1/E_0)^2}}{E_1^2 \sqrt{4
    - 3(E_1/E_0)^2}}, \label{eq:dev_alE1}
\end{align}
where
\begin{equation*}
    \begin{aligned}
        C_1 &= 2 (\alpha^2 (-6 + 7 \alpha^2) + \alpha (12
        - 13 \alpha^2 + \alpha^4)\text{atanh}(\alpha) + 6 (-1 +
        \alpha^2) \text{atanh}^2(\alpha)),\\
        C_2 &= \alpha( \alpha^2 (-3 + 4 \alpha^2) - 6\alpha (-1 +
        \alpha^2) \text{atanh}(\alpha) + (-3 + 2 \alpha^2 + \alpha^4)
        \text{atanh}^2(\alpha)).
    \end{aligned}
\end{equation*}
Since \eqref{eq:dev_E3E2} is a function of $\alpha$, we let $\kappa(\alpha) :=
\pd{E_3}{E_2}$. Direct calculations yield
\begin{align}
    \pd{\kappa(\alpha)}{\alpha} &=
    \frac{4C_3C_4}{C_2^2},\label{eq:dev_kappa}
\end{align}
where
\begin{equation*}
    \begin{aligned}
        C_3 &=(\text{atanh}(\alpha)-\alpha)
        (\alpha(3-2 \alpha^2)- 3 (1-\alpha^2) \text{atanh}(\alpha)),\\
        C_4 &=
        \alpha^2 (-3 + 8 \alpha^2) + (6\alpha - 8
        \alpha^3)\text{atanh}(\alpha)
        + 3 (-1 + \alpha^4) \text{atanh}^2(\alpha).
    \end{aligned}
\end{equation*}

\begin{enumerate}
  \renewcommand{\theenumi}{\alph{enumi})}
\item[\ref{itm:E3E0}] 
    To prove $\pd{E_3}{E_0}>0$ when $\alpha\in(0,1)$, we need to verify that
    $\pd{^2E_3}{E_2\partial E_0} > 0$ and $\pd{E_3}{E_0}
    \Big|_{E_2=\frac{E_1^{2}}{E_0}} > 0$.
    Since $\pd{E_3}{E_2}$ in \eqref{eq:dev_E3E2} is a function of $\alpha$,
    using the chain rule gives
    \[ 
        \pd{^2E_3}{E_0\partial E_2} =
        \pd{\kappa(\alpha)}{\alpha}\pd{\alpha}{E_0}=
        \frac{4C_3C_4}
        {C_2^2}\pd{\alpha}{E_0}.
    \]
    When $\alpha\in(0,1)$, we have $E_1\in(-1,0)$, so one can check that
    $C_3>0$, $C_4<0$ and $\pd{\alpha}{E_0}<0$.  Thus
    we obtain $\pd{^2E_3}{E_2\partial E_0}>0$.

  Then we consider the situation $E_2=E_1^2/E_0$ to make
  $\pd{E_3}{E_0}$ a function of $\alpha > 0$, which can be written as
  \[ 
    \pd{E_3}{E_0}\Big|_{E_2=\frac{E_1^2}{E_0}}
    =\frac{(1-\alpha^2)C_5}{\alpha^2 (-9 + \alpha^4)
      C_2^2 },
  \]
  where 
  \begin{align*}
    \begin{aligned}
      C_5=&-3 \alpha^5 (-135 + 450 \alpha^2 - 537 \alpha^4 + 222
      \alpha^6 + 8 \alpha^8) \\
      &+ 
      \alpha^4 (-2025 + 6129 \alpha^2 - 6345 \alpha^4 + 2067 \alpha^6 + 230 \alpha^8 + 
      24 \alpha^{10}) \text{atanh}(\alpha)\\
      &- 
      2 \alpha^3 (-2025 + 5508 \alpha^2 - 4671 \alpha^4 + 681 \alpha^6 + 472 \alpha^8 + 
      51 \alpha^{10}) \text{atanh}^2(\alpha)\\
      &+ 
      2 \alpha^2 (-2025 + 4887 \alpha^2 - 3033 \alpha^4 - 630 \alpha^6 +
      689 \alpha^8 + 103 \alpha^{10} + 
      9 \alpha^{12}) \text{atanh}^3(\alpha)\\
      &+ \alpha(2025 - 4266 \alpha^2 + 1431 \alpha^4 + 1668 \alpha^6 -
      601 \alpha^8 - 
      218 \alpha^{10} - 39 \alpha^{12}) \text{atanh}^4(\alpha)\\
      &+ 
      3 (5 + \alpha^2) (-3 + 2 \alpha^2 + \alpha^4)^3 \text{atanh}^5(\alpha),
    \end{aligned}
  \end{align*}
  is an elementary function of $\alpha$. One can
  check that $C_5 < 0$ when $\alpha\in(0,1)$. Therefore, $\pd{E_3}{E_0}>0$
  when $E_2 > \dfrac{E_1^2}{E_0}$.
\item[\ref{itm:E3E1}] 
    Let $G(\alpha)\triangleq \pd{E_3}{E_1}\Big|_{E_2=\frac{E_1^2}{E_0}}$. To
    prove $\pd{E_3}{E_1}\Big|_{E_1=0}>0$, it is sufficient to show
    $\pd{^2E_3}{E_1\partial{E_2}}>0$, and $G(0)>0$. Using the chain rule gives
    \[
        \pd{^2E_3}{E_1\partial{E_2}} = 
        \pd{\kappa(\alpha)}{\alpha}\pd{\alpha}{E_1}.
    \]
    Since $\pd{\alpha}{E_1}<0$ when $\alpha\in(0,1)$, we obtain
    $\pd{^2E_3}{E_1\partial E_2}>0$, according to the proof of
    \cref{pro:M1P2}\ref{itm:E3E0}.
    Direct calculations yield
   \[ 
   G(\alpha) = \frac{\frac{64}{7875}\alpha^{21} + O(\alpha^{23})}
   {\frac{16}{675}\alpha^{21} + O(\alpha^{23})},
   \]
   thus we have $G(0) = \frac{12}{35}>0$.
\item[\ref{itm:E3E2_1}] 
    Noticing \eqref{eq:dev_E3E2}, to prove $\pd{E_3}{E_2}<0$ when
    $\alpha\in(0,1)$, we need only to check that the signs of $C_1$ and
    $C_2$ are same. Because $C_1$ and $C_2$ are elementary
    function of $\alpha$, direct calculation yields $C_1,C_2>0$ when
    $\alpha\in(0,1)$.
\item[\ref{itm:E3E2}]
     Noticing that \eqref{eq:dev_E3E2} is a function of $\alpha$, independent on
     $E_2$, and $\alpha$ is independent on $E_2$, we have that $E_3$ is
     linear dependent on $E_2$, i.e., $\pd{^2E_3}{E_2^2}=0$.

\item[\ref{itm:E3geE1}] 
   With the help of \eqref{eq:dev_E3E2}, we can obtain
   \begin{equation*}
       \mE_3-\mE_1^3-(\mE_2-\mE_1^2)\pd{E_3}{E_2}=
       -\dfrac{3(1-\alpha^2)^2 C_6}
       {\alpha^3(3+\alpha^2)^3 C_2},
   \end{equation*}
   where
   \begin{align*}
    \begin{aligned}
      C_6=& -\alpha^3 (27-72\alpha^2+39\alpha^4+2\alpha^6)
       +\alpha^2(81-171\alpha^2+69\alpha^4+19\alpha^6+2\alpha^8) \text{atanh}(\alpha) \\
      & -3\alpha(1-\alpha^2)^2(27+12\alpha^2+\alpha^4) \text{atanh}(\alpha)^2
       +(1-\alpha^2)^2(3+\alpha^2)^3 \text{atanh}^3(\alpha).
    \end{aligned}
  \end{align*}
  One can show that $C_6>0$ when $\alpha\in(0,1)$ because $C_6$ is
  an elementary function. This results 
  \[
   \mE_3-\mE_1^3-(\mE_2-\mE_1^2)\pd{E_3}{E_2} < 0.
  \]

\end{enumerate}